\documentclass{article}

\usepackage[letterpaper,textwidth=395pt,textheight=640pt,centering]{geometry}

\usepackage{amsmath}
\usepackage{amssymb}
\usepackage{amsthm}
\usepackage{graphicx}
\usepackage[authoryear,round]{natbib}
\usepackage{xcolor}
\usepackage{bm}
\usepackage{booktabs}
\usepackage{tikz}
\usepackage{nicefrac}
\usetikzlibrary{arrows.meta,positioning}
\usepackage{pgfplots}
\pgfplotsset{compat=1.16}
\definecolor{cteal}{HTML}{1D9E75}
\definecolor{cpurple}{HTML}{7F77DD}
\definecolor{cgray}{HTML}{888780}
\definecolor{cgreen}{HTML}{4E9A06}
\usepackage{enumitem}
\usepackage{comment}
\usepackage{appendix}
\usepackage[hidelinks]{hyperref}

\makeatletter
\newtheoremstyle{thmstyleone}
{12pt plus2pt minus1pt}{12pt plus2pt minus1pt}{\normalfont\itshape}{0pt}
{\bfseries}{}{.5em}
{\thmname{#1}\thmnumber{\@ifnotempty{#1}{ }{#2}}
  \thmnote{ {\the\thm@notefont(#3)}}}
\newtheoremstyle{thmstyletwo}
{12pt plus2pt minus1pt}{12pt plus2pt minus1pt}{\itshape}{0pt}
{\normalfont}{}{.5em}
{\thmname{#1}\thmnumber{\@ifnotempty{#1}{ }{#2}}
  \thmnote{ {\the\thm@notefont(#3)}}}
\newtheoremstyle{thmstylethree}
{12pt plus2pt minus1pt}{12pt plus2pt minus1pt}{\normalfont}{0pt}
{\bfseries}{}{.5em}%
{\thmname{#1}\thmnumber{\@ifnotempty{#1}{ }{#2}}
  \thmnote{ {\the\thm@notefont(#3)}}}
\makeatother
\theoremstyle{thmstyleone}
\newtheorem{theorem}{Theorem}
\newtheorem{proposition}{Proposition}
\newtheorem{lemma}{Lemma}
\newtheorem{corollary}{Corollary}
\theoremstyle{thmstyletwo}
\newtheorem{remark}{Remark}
\theoremstyle{thmstylethree}
\newtheorem{assumption}{Assumption}

\providecommand{\R}{}\renewcommand{\R}{\mathbb{R}}
\providecommand{\E}{}\renewcommand{\E}{\mathbb{E}}
\providecommand{\Var}{}\renewcommand{\Var}{\operatorname{Var}}

\providecommand{\W}{}\renewcommand{\W}{\mathcal{W}}
\providecommand{\HH}{}\renewcommand{\HH}{\mathcal{H}}
\providecommand{\mC}{}\renewcommand{\mC}{\mathcal{C}}
\providecommand{\bbeta}{}\renewcommand{\bbeta}{\bm{\beta}}
\providecommand{\spn}{}\renewcommand{\spn}{\operatorname{span}}
\providecommand{\barD}{}\renewcommand{\barD}{\overline{\mathcal{D}}}

\allowdisplaybreaks

\makeatletter
\renewcommand\section{\@startsection{section}{1}{\z@}
  {-3.5ex \@plus -1ex \@minus -.2ex}
  {2.3ex \@plus.2ex}
  {\normalfont\large\bfseries}}
\renewcommand\subsection{\@startsection{subsection}{2}{\z@}
  {-3.25ex\@plus -1ex \@minus -.2ex}
  {1.5ex \@plus .2ex}
  {\normalfont\normalsize\bfseries}}
\renewcommand\subsubsection{\@startsection{subsubsection}{3}{\z@}
  {-3.25ex\@plus -1ex \@minus -.2ex}
  {1.5ex \@plus .2ex}
  {\normalfont\normalsize\bfseries}}
\makeatother

\title{Operator-matched spatial regression}
\author{David Bolin\\[4pt]
  \small\itshape Statistics Program, CEMSE Division, King Abdullah University of Science and Technology\\[-1pt]
  \small\upshape Thuwal, Saudi Arabia.
  \href{mailto:david.bolin@kaust.edu.sa}{\texttt{david.bolin@kaust.edu.sa}}}
\date{}

\begin{document}

\maketitle

\makeatletter
\begin{quotation}
  \small\noindent
  \textbf{Abstract.}\hspace{0.6em}\ignorespaces A standard spatial regression model $Y = X^\top\beta + U$, with covariates $X$ and a Gaussian random field $U$ with Mat\'ern covariance, decouples covariate effects from spatial dependence. In several applications this is mechanistically implausible, as covariates likely drive the dynamics generating the residual covariance, and it is statistically problematic, since the generalised least-squares estimator of $\beta$ may be biased by spatial self-confounding. 
  We propose operator-matched spatial regression, where covariates enter both pointwise and as deterministic forcings in a stochastic partial differential equation (SPDE) that generates the random effect, partitioning the apparent effect into local and operator-mediated components. The formulation arises as the steady state of a transport--relaxation equation forced by $X$, giving the parameters physical interpretations. 
  We show that self-confounding is resolved by construction and that maximum likelihood inference for the regression coefficients, implementable through standard covariance-based methods, is calibrated when the operator parameters are known and, under expanding-domain asymptotics, when they are jointly estimated. A computationally efficient finite element implementation is introduced and justified via bounds on the discretisation error. Applications to temperature--altitude regression and air-quality monitoring show that the method matches or outperforms alternatives and produces physically meaningful results.

  \smallskip\noindent
  \textbf{Keywords:}\hspace{0.6em} Gaussian random fields; spatial confounding; stochastic partial differential equations; maximum likelihood estimation; Whittle--Mat\'ern fields; finite element methods
\end{quotation}
\makeatother

\section{Introduction}\label{sec:intro}

Spatial regression is a fundamental tool in geostatistics, environmental epidemiology, and ecological modelling. In the standard additive formulation
\begin{equation}
\label{eq:additive-intro}
   Y(s) \;=\; X(s)^\top \bbeta \;+\; U(s),
   \qquad s \in \mathcal{D},
\end{equation}
the response $Y$ at location $s$ in the spatial domain $\mathcal{D}$ is decomposed into a fixed effect which is linear in some spatially referenced covariates $X$ and a zero-mean Gaussian random field $U$ that captures spatial dependence the covariates fail to explain. The random effect $U$ is typically assumed to have a Mat\'ern covariance function, as this is a flexible class which contains the exponential covariance function as a special case and the Gaussian covariance function as a limiting case \citep{Matern1960, Stein1999}.

\citet{Whittle1954,Whittle1963} showed that a Gaussian random field on $\R^d$ with a Mat\'ern covariance is the stationary solution of the stochastic partial differential equation (SPDE) $L^{\nicefrac{\alpha}{2}}(\tau U) = \W$, where $L = \kappa^2 - \Delta$ and $\W$ is Gaussian white
noise. The covariance has smoothness $\nu = \alpha - \nicefrac{d}{2}$, practical correlation range $\sqrt{8\nu}/\kappa$ and marginal variance determined by the precision parameter $\tau$. \citet{LRL2011} used this characterisation for a bounded domain $\mathcal{D} \subset \R^d$,
\begin{equation}
\label{eq:spde-bg}
   L^{\nicefrac{\alpha}{2}}(\tau U) = \W, \qquad \text{on $\mathcal{D},$} 
\end{equation}
where the operator is augmented with Neumann boundary conditions, 
to obtain a sparse finite element discretisation if $\alpha\in\mathbb{N}$.
The solution
to \eqref{eq:spde-bg} is a Whittle--Mat\'ern field on $\mathcal{D}$
and recovers the stationary Mat\'ern field as
$\mathcal{D} \uparrow \R^d$. The SPDE representation has the additional benefit of giving the random-effect parameters a physical interpretation as $U$ is the equilibrium response of a damped diffusion with white-noise forcing \citep{LindgrenBakkaEtAl2024}. 

This physical interpretation, however, applies only to the random effect. The formulation \eqref{eq:additive-intro} further encodes the mechanistic assumption that $X$ acts pointwise and plays no role in the dynamics that generate  $U$. In many environmental and physical applications this is difficult to justify. For example, surface temperature responds to altitude through atmospheric mixing, the same diffusive process that generates the observed spatial correlation in temperature, and pollutant concentrations are driven by spatially distributed emissions and dispersed by turbulent diffusion. 
The spatial-confounding literature \citep{ReichHodgesZadnik2006, HodgesReich2010, Paciorek2010, KhanCalder2022} has also noted that the additive formulation \eqref{eq:additive-intro} is problematic statistically as the estimated $\beta$ can shift substantially if $U$ has spatial correlation, and various solutions such as Spatial+ \citep{DupontWoodAugustin2022}, structural-equation reformulations \citep{ThadenKneib2018}, and pre-smoothing $X$ \citep{BolinWallin2026} have been proposed.

To address these issues, we propose a modified regression model where the covariate enters the SPDE directly. Specifically, we propose to model the response $Y$ jointly with the spatial dynamics through
the operator-matched \emph{hybrid} SPDE
\begin{equation}
\label{eq:hybrid-spde}
   L^{\nicefrac{\alpha}{2}}\bigl(\tau\,(Y - \beta_{\mathrm{pw}}\, X)\bigr)(s)
   = \beta_{\mathrm{fc}} X(s) + \W(s),
   \qquad s \in \mathcal{D},
\end{equation}
together with an observation model such as $\widetilde Y_i = Y(s_i) + \epsilon_i$, $\epsilon_i \sim N(0, \sigma_\epsilon^2)$. The covariate $X$ enters both as a pointwise additive channel with coefficient $\beta_{\mathrm{pw}}$, and as an SPDE-forcing channel with coefficient
$\beta_{\mathrm{fc}}$ that drives the response through the  operator that generates the spatial dependence. 
The model \eqref{eq:hybrid-spde} has two important sub-models. 
Setting $\beta_{\mathrm{fc}} = 0$ recovers the additive model \eqref{eq:additive-intro} with $U$ given by \eqref{eq:spde-bg}, and setting $\beta_{\mathrm{pw}} = 0$ yields the operator-matched \emph{forced} SPDE,
\begin{equation}
\label{eq:forced-spde}
   L^{\nicefrac{\alpha}{2}}(\tau Y)(s) = \beta X(s) + \W(s),
\end{equation}
in which $X$ and $\W$ jointly drive the response $Y$ through the operator $L^{\nicefrac{\alpha}{2}}$. Fitting \eqref{eq:hybrid-spde}, which adds a single regression parameter to the additive baseline, and examining the relative magnitudes of $\beta_{\mathrm{pw}}$ and $\beta_{\mathrm{fc}}$ provides an empirical test of which mechanism dominates a given application, without having to commit a priori to either extreme. 

Inverting the operator gives the equivalent formulation ${Y = \beta_{\mathrm{pw}}X + \beta_{\mathrm{fc}}\mathcal{S}X + U}$ with $\mathcal{S} = \tau^{-1}L^{-\nicefrac{\alpha}{2}}$, so the forcing is a smoothing of $X$ by the operator that defines the random effect. We refer to this property as \emph{operator matching}. This is a parsimonious choice that solves the confounding issues raised by \citet{BolinWallin2026}, which does not require any separate tuning and propagates smoothing uncertainty through the joint posterior of the SPDE parameters. If included, the pointwise effect $\beta_{\mathrm{pw}}X$ retains the classical additive interpretation. 
A schematic of the approach is shown in Figure~\ref{fig:schematic}.
Importantly, the hybrid SPDE can often be motivated from physical principles, as we will detail in Section~\ref{sec:physical}, and the parameters have direct physical interpretations. For example, $\kappa^{-1}$ acts as a mixing length and the standardised forcing coefficient $\beta^{*}_{\mathrm{fc}} = \beta_{\mathrm{fc}}/(\tau\kappa^{\alpha})$ as a (sign-reversed) equilibrium lapse rate or source gain. The additive parametrisation cannot offer this since there the same random-effect parameters describe only the residual correlation, and the regression coefficient inherits no mechanistic meaning. 

\begin{figure*}[t]
\centering
\begin{minipage}[c]{0.58\textwidth}\centering
\resizebox{\linewidth}{!}{
\begin{tikzpicture}[
  font=\small,
  bx/.style={draw, rounded corners=3pt, align=center, inner sep=3pt, line width=0.5pt},
  gy/.style={bx, draw=cgray, fill=cgray!12},
  tl/.style={bx, draw=cteal, fill=cteal!12},
  pu/.style={bx, draw=cpurple, fill=cpurple!12},
  gn/.style={bx, draw=cgreen, fill=cgreen!12},
  ar/.style={-{Stealth[length=2mm]}, line width=0.5pt, draw=cgray}
]
\node[gy] (X) at (0,0.95) {$X$\\[-2pt]\scriptsize covariate};
\node[gy] (W) at (0,-0.95) {$\mathcal{W}$\\[-2pt]\scriptsize white noise};
\node[pu] (S) at (3.1,0) {$\mathcal{S}=\tau^{-1}L^{-\nicefrac{\alpha}{2}}$\\[-2pt]\scriptsize operator smoothing};
\node[tl] (bpw) at (3.2,2.1) {$\beta_{\mathrm{pw}}X$\\[-2pt]\scriptsize pointwise, local};
\node[gn] (bfc) at (6.6,0.95) {$\beta_{\mathrm{fc}}\mathcal{S}X$\\[-2pt]\scriptsize forcing, regional};
\node[gn] (U) at (6.6,-0.95) {$U=\mathcal{S}\mathcal{W}$\\[-2pt]\scriptsize random effect};
\node[gy] (Y) at (9.6,0.55) {$Y$\\[-2pt]\scriptsize response};
\draw[ar] (X) -- (S);
\draw[ar] (W) -- (S);
\draw[ar] (X) -- (bpw);
\draw[ar] (S) -- (bfc);
\draw[ar] (S) -- (U);
\draw[ar] (bpw.east) -| (Y.north);
\draw[ar] (bfc) -- (Y);
\draw[ar] (U.east) -| (Y.south);
\end{tikzpicture}}\\[3pt]{\small (a) the two channels}
\end{minipage}\hfill
\begin{minipage}[c]{0.38\textwidth}\centering
\resizebox{\linewidth}{!}{
\begin{tikzpicture}
\begin{axis}[
  width=6.4cm, height=4.8cm,
  axis lines=left,
  xlabel={$\omega$}, ylabel={$\beta(\omega)$},
  xmin=0, xmax=8, ymin=0.15, ymax=1.4,
  xtick={1}, xticklabels={$\omega_{1/2}$},
  ytick={0.4,1.2}, yticklabels={$\beta_{\mathrm{pw}}$,$\beta_{\mathrm{pw}}{+}\beta_{\mathrm{fc}}^{\star}$},
  tick label style={font=\scriptsize}, label style={font=\small},
  clip=false
]
\addplot[cteal, line width=1pt, domain=0:8, samples=120]{0.4+0.8/(1+x)};
\draw[dashed, cgray] (axis cs:0,1.2)--(axis cs:8,1.2);
\draw[dashed, cgray] (axis cs:0,0.4)--(axis cs:8,0.4);
\draw[dashed, cgray] (axis cs:1,0.15)--(axis cs:1,0.9);
\node[font=\scriptsize, anchor=west] at (axis cs:0.5,1.05) {large scales};
\node[font=\scriptsize, anchor=east] at (axis cs:7.7,0.55) {small scales};
\end{axis}
\end{tikzpicture}}\\[3pt]{\small (b) scale-varying coefficient}
\end{minipage}
\caption{Operator-matched spatial regression. (a) The same
operator $\mathcal{S}=\tau^{-1}L^{-\nicefrac{\alpha}{2}}$ maps the covariate $X$ to the regional forcing $\beta_{\mathrm{fc}}\mathcal{S}X$ and the white noise $\mathcal{W}$ to the random effect $U=\mathcal{S}\mathcal{W}$. Together with the pointwise, local channel $\beta_{\mathrm{pw}}X$ these give $Y=\beta_{\mathrm{pw}}X+\beta_{\mathrm{fc}}\mathcal{S}X+U$. (b) The effective coefficient $\beta(\omega)$ (Proposition~\ref{prop:scalefree}(ii)) varies with squared spatial frequency $\omega$, decreasing from $\beta(0)=\beta_{\mathrm{pw}}+\beta_{\mathrm{fc}}^{\star}$ at large scales to $\beta(\infty)=\beta_{\mathrm{pw}}$ at small scales, with half the forcing gain lost at $\omega_{\nicefrac1{2}}=\kappa^2(2^{\nicefrac{2}{\alpha}}-1)$.}
\label{fig:schematic}
\end{figure*}
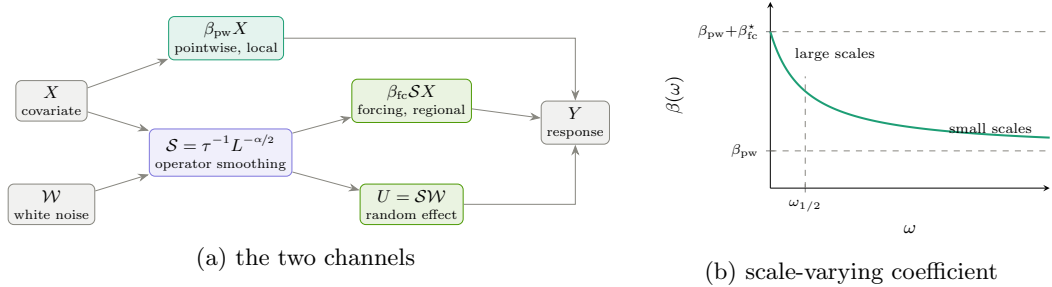

Related constructions appear, for different reasons, in largely disconnected literatures. \citet{HeatonGelfand2011} proposed regression on kernel-averaged covariates with an estimated kernel, and \citet{LindmarkAndersonThorson2026} suggested smoothing covariates by an SPDE-based diffusion operator because organisms respond to spatially averaged habitat. \citet[Section~5.5, p.~629]{ReichEtAl2021} similarly propose to regress the response on $X$ and on a kernel average of $X$, with the kernel fixed by the analyst.
\citet{GiffinEtAl2023} estimate the range of such a kernel jointly with the two coefficients.
\citet{GuanEtAl2023} adjust for spatial confounding by adding a smoothed copy of the covariate to the mean of the response alongside the covariate itself. They obtain the smoother either from a bivariate Mat\'ern model for the covariate and the unmeasured confounder, or as a B-spline mixture for its spectral transfer function. In all of these constructions the smoother is chosen or parametrised separately from the random effect and is thus not operator-matched, and \citet{HeatonGelfand2011} and \citet{GuanEtAl2023} additionally treat the covariate itself as a random field, whereas here $X$ enters as deterministic forcing. 
In econometrics the same structure is classical for time series and discretely indexed spatial models. Although similar ideas appear in both fields, the connections between geostatistics and spatial econometrics are rarely made explicit. To help bridge this gap, we discuss the connections between our proposed models and those in econometrics in detail in Section~\ref{sec:econometrics}.
Covariate smoothing also appears in the spatial-misalignment literature, where covariates observed at other locations are interpolated to the response sites and used as regressors \citep{GryparisEtAl2009, SzpiroSheppardLumley2011}, including in hedonic house-price models \citep{AnselinLozanoGracia2008}. 
There, the smoothed value stands in for the covariate that should have been observed, so smoothing is a source of error to be corrected, rather than a deliberate operator-matched channel.

In Section~\ref{sec:properties}, we further justify the model by deriving various identifiability results and showing asymptotic properties of maximum likelihood parameter estimation. We also discuss testing the operator-matching hypothesis and show that the model indeed solves the self-confounding issues raised by \citet{BolinWallin2026}, where the additive estimator of $\beta$ is driven towards zero when the covariate is rougher than the random effect, which motivates it as a sensible choice even for applications that lack a physical motivation. 
On the computational side, Section~\ref{sec:inference} first shows how the model can be implemented using standard covariance-based methods and then develops a computationally efficient finite element implementation, which fits into the standard \texttt{R-INLA} pipeline \citep{rue2009,lindgren2015bayesian} and thus can be estimated by maximum likelihood or fully Bayesian inference, at the same cost as the additive model. We also quantify the effect of the discretisation, by showing that the finite element model is itself operator-matched, and bound the error that the discretisation induces on the estimated coefficients. 
Section~\ref{sec:app-temp-alt} presents an application to regression of temperature on altitude across three climate regimes. The results demonstrate both predictive gains of the hybrid model over the alternatives and a meaningful physical decomposition of the apparent altitude--temperature relationship into local-column and horizontal-mixing components. 
Section~\ref{sec:app-no2} applies the framework to air-quality monitoring, modelling nitrogen dioxide as a dispersed footprint of measured emission sources. Section~\ref{sec:disc} closes with a discussion. Further details and all proofs are given in four appendices. 
Code reproducing all results is available at \texttt{github.com/davidbolin/OMSR/}.

\section{Physical interpretation}\label{sec:physical}

In Section~\ref{sec:source-like} and Section~\ref{sec:state-like}, we develop the physical interpretation of the hybrid SPDE with $\alpha = 2$ in two complementary cases, distinguished by whether $X$ is a source of the modelled quantity or an external state to which the quantity relaxes. We end the section by discussing the need for models with general smoothness $\alpha\neq 2$.

\subsection{Source-like covariates}\label{sec:source-like}
The forced SPDE has a direct physical interpretation when $Y$ is the concentration of a substance that spreads by isotropic diffusion with constant diffusivity $D$, undergoes first-order decay at rate $\rho$, and is replenished by a deterministic source field $X(s)$. 
Examples include airborne pollutants replenished by surface emissions, solutes in groundwater replenished by infiltration from the surface, and dispersal of radioactive tracers such as radon. Mass balance for $Y(s,t)$ then gives
\begin{equation}\label{eq:tracer-time}
   \partial_t Y = D\Delta Y - \rho Y + \gamma\, X(s) + \dot W(s,t),
\end{equation}
where $\gamma$ is a source-to-concentration conversion coefficient and
$\dot W$ is space--time white noise, the formal time derivative of a cylindrical Wiener process $W_I$ on $L^2(\mathcal{D})$, representing forcing by unresolved processes. In the stationary regime, taking expectations in \eqref{eq:tracer-time} shows that $m = \E[Y]$ satisfies 
$(\rho - D\Delta)\, m = \gamma X$, which is the mean equation
of the forced SPDE \eqref{eq:forced-spde} with $L = \kappa^2 - \Delta$, $\kappa^2 = \nicefrac{\rho}{D}$, $\nicefrac{\beta}{\tau} = \nicefrac{\gamma}{D}$, and $\alpha = 2$.
The fluctuations about the mean require more care. Proposition~\ref{prop:timeavg} below shows that time-averaged responses follow \eqref{eq:forced-spde} with exact mean and Gaussian law, and covariance exact up to a perturbation of order $(\rho T)^{-1}$, with $\tau = D\sqrt{T}$ set
by the averaging window $T$.
The parameter $\kappa^{-1} = \sqrt{\nicefrac{D}{\rho}}$ is the diffusive decay length, the spatial scale at which a localised emission spreads by diffusion before being eliminated by decay. 
From the mean balance we see that, under Neumann boundary conditions, a spatially constant source of unit intensity raises the equilibrium concentration by $\nicefrac{\gamma}{\rho}$, and in the parametrisation of \eqref{eq:forced-spde} this equilibrium gain is $\beta_{\mathrm{fc}}^{\star} := \nicefrac{\beta}{(\tau\kappa^{2})} = \nicefrac{\gamma}{\rho}$. We call $\beta_{\mathrm{fc}}^{\star}$ the \emph{standardised forcing coefficient} as it measures response per unit covariate, and it is a natural scale to report the forcing channel. 

The hybrid form \eqref{eq:hybrid-spde} is appropriate when $Y$ depends on $X$ through a direct local pathway that responds essentially pointwise to $X$ at the same location, and through an indirect regional pathway where $X$ propagates to the measurement location via the dynamics
\eqref{eq:tracer-time}. In urban air-quality monitoring this is a standard decomposition \citep{LenschowEtAl2001} as a station sees both a direct contribution from the immediately adjacent emission and a smoothed urban-background contribution. 

\subsection{State-like covariates}\label{sec:state-like}

When $X$ is not a source of $Y$ but rather a fixed state variable that influences the equilibrium toward which $Y$ relaxes, the forced-SPDE form still arises from a steady-state derivation. The canonical example we use here, and which we will return to in Section~\ref{sec:app-temp-alt}, is the regression of surface temperature on altitude. A further example is an ecological  model in which the density of a dispersing population relaxes toward a habitat-determined equilibrium population size while spreading by diffusion with constant diffusion coefficient \citep{Skellam1951, HefleyEtAl2017}.

Consider a time-evolving heat balance for surface temperature $T(s,t)$ that mixes horizontally and relaxes toward an altitude-dependent equilibrium reference $T_{\mathrm{eq}}(h)$,
\begin{equation}\label{eq:heat-time}
\partial_t T = D\Delta T - \rho\bigl(T - T_{\mathrm{eq}}(h)\bigr) + \dot W(s,t).
\end{equation}
This is a diffusive energy-balance model \citep[Eq.~(53)]{NorthCahalanCoakley1981} in which, following the stochastic-climate paradigm of \citet{Hasselmann1976}, the unresolved weather enters as forcing that is white in space and time, the diffusion term $D\,\Delta T$ represents horizontal turbulent mixing (the eddy transport of heat, with effective diffusivity $D$), the term $-\rho(T - T_{\mathrm{eq}})$ is Newtonian relaxation toward the reference at rate $\rho$, $h(s)$ is altitude, and $T_{\mathrm{eq}}(h) = T_0 - \gamma h$ is the altitude-dependent equilibrium reference, in which $T_0$ is the sea-level reference temperature and $\gamma$ the equilibrium lapse rate.
In the stationary regime the mean $m = \E[T]$ satisfies $(\rho - D\,\Delta)\,m = \rho T_0 - \rho\gamma\,h$. Since constants are eigenfunctions of $\rho - D\Delta$ under the Neumann conditions, the anomaly $T - T_0$ has mean $\tilde m = m - T_0$ satisfying $(\rho - D\,\Delta)\,\tilde m = -\rho\gamma\,h$, the mean equation of \eqref{eq:forced-spde} with $\kappa^2 = \nicefrac{\rho}{D}$, $\nicefrac{\beta}{\tau} = -\nicefrac{\rho\gamma}{D}$, $\alpha = 2$, $X = h$. In practice $T_0$ is unknown and is estimated as an ordinary intercept, the convention used in Section~\ref{sec:app-temp-alt}.
Proposition~\ref{prop:timeavg} below makes the correspondence exact in mean and Gaussian law, up to an $O((\rho T)^{-1})$ covariance perturbation, for time-averaged responses. 
Also here the parameters have physical interpretations. The length scale $\kappa^{-1} = \sqrt{\nicefrac{D}{\rho}}$ is the atmospheric mixing length, i.e., the horizontal distance over which surface temperature equilibrates toward the altitude-implied reference. The ratio
$\nicefrac{\beta}{\tau} = -\nicefrac{\rho\gamma}{D}  = -\gamma\kappa^2$ couples the equilibrium lapse rate $\gamma$ to $\kappa^2$. Thus,  $\beta_{\mathrm{fc}}^{\star} = (\nicefrac{\beta}{\tau})\,\kappa^{-2} = -\gamma$ is minus the lapse rate.
The forcing-free version of \eqref{eq:heat-time}, with a spatially uniform equilibrium, is a standard stochastic model for temperature anomalies \citep{NorthWangGenton2011}. With a uniform equilibrium, every spatial feature of the field must be carried by its covariance, while letting the equilibrium depend on altitude moves the altitude-driven structure into the mean.

Surface temperature is typically not observed instantaneously, but as an average over some time window, which often is months or even multi-decadal for statistical applications in climate. 
The following result shows that for windows long relative to the relaxation time $\rho^{-1}$,  we obtain a Whittle--Mat\'ern field with $\alpha = 2$, which motivates $\alpha=2$ as a default choice. Throughout we make the following assumptions. 

\begin{assumption}\label{ass:A}
$\mathcal{D} \subset \R^d$, $d \le 3$, is a bounded convex domain or a bounded domain with $C^{1,1}$ boundary (locally the graph of a differentiable function with Lipschitz derivative), $L = \kappa^2 - \Delta$ is equipped with homogeneous Neumann boundary conditions, $\kappa, \tau > 0$, and $\alpha > \nicefrac{d}{2}$. Further, $X \in L^2(\mathcal{D})$, and when the pointwise channel is included in the model we additionally require $X \in \HH = D(L^{\nicefrac{\alpha}{2}})$.
\end{assumption}

\begin{proposition}\label{prop:timeavg}
Let Assumption~\ref{ass:A} hold and set $A = D L$ with $D > 0$ and consider 
$$
\mathrm{d}Y_t = \bigl(-A Y_t + \gamma_0 X\bigr)\,\mathrm{d}t + \mathrm{d}W_{I,t},
$$
where $\gamma_0 \in \R$. This equation has a unique invariant law. Let $Y$ denote its stationary solution and let $\bar Y_T = T^{-1}\int_0^T Y_t\,\mathrm{d}t$ be its time average over a window $T$. 
Then, for every $T>0$, $\bar Y_T$ is an $L^2(\mathcal{D})$-valued Gaussian random variable with mean $(\nicefrac{\gamma_0}{D})L^{-1}X$ and
covariance operator $\tau_T^{-2}L^{-2}(I + R_T)$, where $\tau_T = D\sqrt{T}$ and $\|R_T\| \le (\rho T)^{-1}$ in the operator norm $\|\cdot\|$  and $\rho = D\kappa^2$ is the relaxation rate. That is, $\bar Y_T$ follows the forced model \eqref{eq:forced-spde} with $\alpha = 2$, scale $\tau_T$, and forcing coefficient $\beta_T = \gamma_0\sqrt{T}$, perturbed by $I + R_T$. The standardised forcing coefficient $\beta_{\mathrm{fc}}^{\star} = (\nicefrac{\beta_T}{\tau_T})\,\kappa^{-2} = \nicefrac{\gamma_0}{\rho}$ does not depend on $T$ and 
$\sqrt{T}\,\bigl(\bar Y_T - (\nicefrac{\gamma_0}{D})L^{-1}X\bigr)
\xrightarrow{d} U$,
as $\rho T \to \infty$, where $U$ is a Whittle--Mat\'ern field with $\alpha=2$ and $\tau=D$.
\end{proposition}

The operator-matched model with $\alpha = 2$ is thus natural for aggregated responses. A complementary route to $\alpha = 2$ for instantaneous responses under spatially correlated forcing is given in Appendix~\ref{app:colored}.
To motivate the hybrid formulation in this case, note that two physically distinct mechanisms tie surface temperature to altitude. The local mechanism (adiabatic cooling and the column-vertical radiation budget) is essentially pointwise as surface temperature at a station is set by the elevation at that station through the vertical thermodynamics of the overlying air column.
At the same time, surface temperature relaxes toward an altitude-dependent equilibrium, but is horizontally mixed by turbulence on the scale of $\kappa^{-1}$, which is the regional mechanism encoded in \eqref{eq:heat-time}. 
An observed altitude--temperature relationship is a superposition of the two mechanisms, which \eqref{eq:additive-intro} conflates and the hybrid model \eqref{eq:hybrid-spde} makes explicit.

\subsection{Do we need general smoothness?}\label{sec:fractional-alpha}
The motivations above concluded that $\alpha = 2$, i.e., $\nu = 1$ in two dimensions, is the physically motivated smoothness. However, several physically grounded mechanisms lead to a steady-state operator of fractional order $\alpha \neq 2$.
An example is anomalous diffusion: Transport through porous media, fractured rock, or heterogeneous matrices is typically modelled by replacing the Laplacian in \eqref{eq:tracer-time} with a fractional Laplacian \citep{MetzlerKlafter2000, BensonEtAl2004}.
Another example is dimension reduction: a field that is the surface restriction of a three-dimensional diffusion--damping balance inherits its smoothness $\nu = \alpha - \nicefrac{d}{2}$, and with $\alpha = 2$ in $d = 3$, the restriction to a two-dimensional surface has $\alpha = \nu + 1 = \nicefrac32$. Boundary-layer pollutant concentrations and sea surface temperatures restricted from the full water column have this character. As dimension reduction reduces $\alpha$ whereas temporal averaging raises it (Proposition~\ref{prop:timeavg}), it could thus be beneficial to estimate the smoothness from data in applications such as that in Section~\ref{sec:app-temp-alt}.

A third mechanism is the superposition of  processes operating at distinct spatial scales. Surface temperature is for example influenced by processes ranging from sub-metre turbulent diffusion to synoptic systems at thousands of kilometres. If each scale is a diffusion--damping balance with $\alpha = 2$ and its own $(\kappa, \tau)$, their superposition may be close to a  Mat\'ern field with a fractional smoothness. This is the converse of the decomposition in \citet{xiong2022}, which represents a fractional Mat\'ern field as a sum of integer-order components across scales.

\section{Theoretical properties}\label{sec:properties}

The physical motivations are not required for the forced and hybrid models to be useful, as they in any case serve as principled  corrections for spatial self-confounding. In this section, we show this and develop the inferential theory for the models. 

\subsection{Setting and identifiability}\label{sec:cm}

The first result collects the structural properties on which everything else rests. 

\begin{proposition}\label{prop:rep}
Under Assumption~\ref{ass:A}, the hybrid SPDE \eqref{eq:hybrid-spde}
has a unique solution
\begin{equation}
\label{eq:additive}
Y = \beta_{\mathrm{pw}} X + \beta_{\mathrm{fc}} \mathcal{S}X + U,
\end{equation}
where $\mathcal{S} = \tau^{-1} L^{-\nicefrac{\alpha}{2}}$ and $U = \mathcal{S}\W$ is a Gaussian Whittle--Mat\'ern field with covariance operator $\mC = \tau^{-2} L^{-\alpha}$ and continuous modification. Its Cameron--Martin space is $\HH = D(L^{\nicefrac{\alpha}{2}})$, with inner product
$(f, g)_\mC = \tau^2 (L^{\nicefrac{\alpha}{2}} f, L^{\nicefrac{\alpha}{2}} g)_{L^2}$ and norm $\|f\|_\mC = \tau \|L^{\nicefrac{\alpha}{2}} f\|_{L^2}$. Further, $\HH \hookrightarrow C(\barD)$ and $\mathcal{S} : (L^2(\mathcal{D}), \|\cdot\|_{L^2}) \to (\HH, \|\cdot\|_\mC)$ is an isometric isomorphism, so $\mathcal{S}X \in \HH$ for each $X \in L^2(\mathcal{D})$.
\end{proposition}

Thus, the hybrid model regresses on the pair $(X, \mathcal{S}X)$, and its asymptotic behaviour is controlled by the Gram matrix of this pair
in the Cameron--Martin inner product,
\begin{equation}\label{eq:gram}
   \bm{G} =
   \begin{pmatrix}
      \|X\|_\mC^2 & (X, \mathcal{S}X)_\mC \\
      (X, \mathcal{S}X)_\mC & \|\mathcal{S}X\|_\mC^2
   \end{pmatrix}.
\end{equation}
Specifically, Theorem~\ref{thm:infill} below shows that the maximum likelihood estimator of $\bbeta = (\beta_{\mathrm{pw}},\beta_{\mathrm{fc}})$ is Gaussian with asymptotic covariance $\bm{G}^{-1}$.
To characterise when the two channels are identifiable, we compute $\bm{G}$ in the eigenbasis of $L$. For this, we write $\{(\lambda_j, e_j)\}_{j\ge 1}$ for the eigenpairs of $L$, with $\{e_j\}$ an orthonormal basis of $L^2(\mathcal{D})$ and $\lambda_1 \le \lambda_2 \le \cdots \to \infty$. We further write $\omega_j := \lambda_j - \kappa^2 \ge 0$ for the Neumann--Laplacian eigenvalues, so that $\omega_j$ is a squared spatial frequency, and $x_j = (X, e_j)_{L^2}$ for the coefficients of a function $X$.
For $X \neq 0$, the weights $w_j = \nicefrac{x_j^2}{\|X\|_{L^2}^2}\geq 0$ sum to one, so they form a probability distribution which we refer to as the spectral distribution of $X$ over the eigenvalues of $L$. We write $\E_w$ and $\Var_w$ for expectation and variance under this distribution, i.e., for a discrete random variable $\lambda$ which is equal to $\lambda_j$ with probability $w_j$, $\E_w[g(\lambda)] = \sum_j w_j\, g(\lambda_j)$ and $\Var_w[g(\lambda)] = \E_w[g(\lambda)^2] - \E_w[g(\lambda)]^2$.

\begin{theorem}\label{thm:geometry}
Let $X \in \HH \setminus \{0\}$ under Assumption~\ref{ass:A}. Then
\begin{equation}
\label{eq:gram-spectral}
   \|X\|_\mC^2 = \tau^2 \sum_j \lambda_j^{\alpha} x_j^2,
   \qquad
   (X, \mathcal{S}X)_\mC = \tau \sum_j \lambda_j^{\nicefrac{\alpha}{2}} x_j^2,
   \qquad
   \|\mathcal{S}X\|_\mC^2 = \|X\|_{L^2}^2,
\end{equation}
and the channel-separation index
\begin{equation}\label{eq:separation}
   \vartheta(X)
   :=
   1 - \frac{(X,\mathcal{S}X)_\mC^2}{\|X\|_\mC^2\, \|\mathcal{S}X\|_\mC^2}
   =
   \frac{\Var_w\bigl(\lambda^{\nicefrac{\alpha}{2}}\bigr)}{\E_w\bigl[\lambda^{\alpha}\bigr]}
   \in [0, 1),
\end{equation}
the squared sine of the Cameron--Martin angle between $X$ and $\mathcal{S}X$, satisfies: (i) $\bm{G}$ is non-singular if and only if $\vartheta(X) > 0$, which holds if and only if $X$ is not an eigenfunction of $L$; and (ii) $\det \bm{G} = \|X\|_\mC^2\, \|X\|_{L^2}^2\, \vartheta(X)$, and
\begin{equation}\label{eq:varinflation}
   (\bm{G}^{-1})_{11} = \frac{1}{\|X\|_\mC^2\, \vartheta(X)},
   \qquad
   (\bm{G}^{-1})_{22} = \frac{1}{\|X\|_{L^2}^2\, \vartheta(X)}.
\end{equation}
\end{theorem}

The index $\vartheta(X)$ is invariant to rescaling of $X$, is independent of $\tau$, and has a clear interpretation in that the two channels are separable exactly to the extent that the spectral energy of $X$ is spread across the spectrum of $L$. A covariate concentrated near a single spatial scale (e.g., a smooth large-scale trend surface) has $\vartheta(X) \approx 0$ and $\mathcal{S}X$ is then nearly proportional to $X$. By \eqref{eq:varinflation}, both asymptotic channel variances are then inflated by the factor $1/\vartheta(X)$ relative to
the single-channel fits. 

\begin{remark}\label{rem:multivariate}
With $p$ covariates the same geometry holds with the $2p \times 2p$ Cameron--Martin Gram matrix $\bm{G}_p$ of $(X_1,\dots,X_p,\mathcal{S}X_1,\dots,\mathcal{S}X_p)$, whose entries are the cross-covariate analogues of \eqref{eq:gram-spectral}. The channels are jointly identifiable if and only if $\bm{G}_p$ is non-singular, and the results below, including the limit law of Theorem~\ref{thm:infill}, hold verbatim with $\bm{G}_p$ in place of $\bm{G}$. In particular, the asymptotic covariance of $\hat\bbeta$ is then $\bm{G}_p^{-1}$ (justified in Appendix~\ref{app:proofs}). 
Writing $\bm{R}_p$ for the correlation matrix obtained by scaling $\bm{G}_p$ by its diagonal, the $k$th coefficient has variance-inflation factor $(\bm{R}_p^{-1})_{kk}\ge 1$ and \emph{generalised channel-separation index} $\vartheta_k := (\bm{R}_p^{-1})_{kk}^{-1} \in (0,1]$.
\end{remark}

We now return to the standardised forcing coefficient from Section~\ref{sec:source-like}, which extends to any $\alpha$. The constant function is an eigenfunction of $L$ with eigenvalue $\kappa^2$, so $L^{-\nicefrac{\alpha}{2}}\mathbf{1} = \kappa^{-\alpha}\mathbf{1}$. Therefore, writing the implicit smoother in unit-gain form, $\widetilde{\mathcal{S}} := \kappa^{\alpha} L^{-\nicefrac{\alpha}{2}}$, we have $\widetilde{\mathcal{S}}\mathbf{1} = \mathbf{1}$, and the hybrid mean is $\beta_{\mathrm{pw}}X + \beta_{\mathrm{fc}}^{\star}\widetilde{\mathcal{S}}X$ with $\beta_{\mathrm{fc}}^{\star} = \beta_{\mathrm{fc}}/(\tau\kappa^{\alpha})$ in the same units as $\beta_{\mathrm{pw}}$.

\begin{proposition}\label{prop:scalefree}
Under Assumption~\ref{ass:A} we have:
\begin{enumerate}[label=(\roman*),leftmargin=*,itemsep=1pt]
\item  with $\beta_{\mathrm{fc}}^{\star} = \beta_{\mathrm{fc}}/(\tau\kappa^{\alpha})$, the hybrid model \eqref{eq:additive} is
$Y = \beta_{\mathrm{pw}} X + \beta_{\mathrm{fc}}^{\star}\, \widetilde{\mathcal{S}}X + U$. The map $(\beta_{\mathrm{pw}}, \beta_{\mathrm{fc}}) \mapsto (\beta_{\mathrm{pw}}, \beta_{\mathrm{fc}}^{\star})$ is a bijection that leaves the likelihood, the predictions, and the likelihood-ratio statistics unchanged;
\item  the mean $m = \beta_{\mathrm{pw}}X + \beta_{\mathrm{fc}}^{\star}\widetilde{\mathcal{S}}X$ has coefficients $m_j = (m, e_j)_{L^2} = \beta(\omega_j)\, x_j$ with $\beta(\omega) = \beta_{\mathrm{pw}} + \beta_{\mathrm{fc}}^{\star} \bigl(1 + \omega/\kappa^2\bigr)^{-\nicefrac{\alpha}{2}}$.
Thus, the hybrid model is a linear regression whose coefficient varies monotonically across spatial scales, from $\beta(0) = \beta_{\mathrm{pw}} + \beta_{\mathrm{fc}}^{\star}$ for purely large-scale variation to $\beta(\infty) = \beta_{\mathrm{pw}}$ for purely small-scale variation, with half of the forcing gain lost at $\omega_{1/2} = \kappa^2(2^{\nicefrac{2}{\alpha}} - 1)$;
\item $\|\widetilde{\mathcal{S}}X\|_{L^2} = g_X \|X\|_{L^2}$ with $g_X = (\E_w\bigl[(1 + \nicefrac{\omega}{\kappa^2})^{-\alpha} \bigr])^{\nicefrac1{2}} \in (0, 1]$, and ${g_X = 1}$ if and only if $X$ is constant. 
\end{enumerate}
\end{proposition}

We thus have that $\beta_{\mathrm{fc}}^{\star}$ places the channel on the same \emph{units} as $\beta_{\mathrm{pw}}$ whereas the variance-matched $\beta_{\mathrm{fc}}^{\star}g_X$ matches its realised \emph{magnitude} for the covariate at hand, as the forcing contribution $\beta_{\mathrm{fc}}^{\star}\widetilde{\mathcal{S}}X$ has the same $L^2(\mathcal{D})$ norm as a pointwise term with coefficient $\beta_{\mathrm{fc}}^{\star}g_X$. Since $g_X < 1$ whenever $X$ has energy at scales below $\kappa^{-1}$, quoting $\beta_{\mathrm{fc}}^{\star}$ alone overstates the channel's realised contribution for such covariates.

The second property is visualised in Figure~\ref{fig:schematic}. Because of this property, one could view this as a parametric, operator-matched counterpart of the spectral confounding adjustment of \citet{GuanEtAl2023}, in which the regression coefficient is allowed
to vary with spatial frequency. 
An advantage of the operator-matched version is that $\mathcal{S}$ is fixed by the random effect, which removes the non-identifiability that an unrestricted smoother allows, where the smoothed copy may equal $X$ itself \citep[Section~2.3]{GuanEtAl2023}, and leaves identification as
the condition $\vartheta(X) > 0$.

\subsection{Infill asymptotics}\label{sec:wellspec}

We now consider asymptotic properties of the model in the infill regime, where the domain is fixed. The observations and the sampling design are specified as follows.

\begin{assumption}\label{ass:design}
The observations are $\widetilde Y_i = Y(s_i) + \epsilon_i$ with $\epsilon_i \stackrel{\mathrm{iid}}{\sim} N(0, \sigma_\epsilon^2)$,
$\sigma_\epsilon^2 > 0$, independent of $\W$. The design $\{s_i\}_{i \ge 1} \subset \barD$ is a fixed space-filling sequence satisfying $\#\{i \le n : s_i \in B\} \to \infty$ as $n \to \infty$ for every open ball $B$ with $B \cap \mathcal{D} \neq \emptyset$.
\end{assumption}

The first result delimits what any estimator can achieve in this regime. 

\begin{proposition}\label{prop:frontier}
Under Assumption~\ref{ass:A}, the Gaussian law of the field $Y$ in \eqref{eq:additive} with parameters $(\bbeta, \kappa, \tau, \alpha)$,  $\bbeta = (\beta_{\mathrm{pw}}, \beta_{\mathrm{fc}})^\top$, and the law with parameters $(\bbeta', \kappa', \tau', \alpha')$ are equivalent if $(\tau, \alpha) = (\tau', \alpha')$ and mutually singular if $\tau \neq \tau'$ or $\alpha \neq \alpha'$.
\end{proposition}
As equivalent laws cannot be distinguished consistently, and equivalence carries over to the noisy observations (Lemma~\ref{lem:equiv}), neither $\bbeta$ nor $\kappa$ is consistently estimable under infill asymptotics. 
This is the operator-matched analogue of the classical result that $\kappa$ cannot be estimated consistently in Mat\'ern fields \citep{Zhang2004}, and the deterministic forcing does not enlarge the set of consistently estimable parameters. Consistency of $\hat\bbeta_n$ is thus the wrong target. The operator-matched estimator instead has asymptotic unbiasedness and calibrated uncertainty, with a very simple limit law:

\begin{theorem}\label{thm:infill}
Suppose the data are generated by the model \eqref{eq:additive} with coefficients $\bbeta_0 = (\beta_{\mathrm{pw},0}, \beta_{\mathrm{fc},0})$ and known $(\kappa, \tau, \alpha, \sigma_\epsilon^2)$, let Assumptions~\ref{ass:A} and~\ref{ass:design} hold, and assume that $\vartheta(X) > 0$. Let $\hat\bbeta_n$ be the maximum likelihood estimator based on $(\widetilde Y_1, \ldots, \widetilde Y_n)$, and let $\bm{M}_n = \bm{D}_n^\top \bm{\Sigma}_n^{-1} \bm{D}_n$ denote its information matrix, where $\bm{D}_n = [X(s_i), (\mathcal{S}X)(s_i)]_{i \le n}$ and
$\bm{\Sigma}_n$ is the covariance matrix of the observations. Then:
\begin{enumerate}[label=(\roman*)]
\item $\hat\bbeta_n \sim N(\bbeta_0,\, \bm{M}_n^{-1})$, for every $n$ at which $\bm{M}_n$ is non-singular, which by (ii) is the case for all $n$ large enough;
\item $\bm{M}_n \to \bm{G}$ as $n \to \infty$, where $\bm{G}$ is given by \eqref{eq:gram}; consequently $\hat\bbeta_n \xrightarrow{d} N\bigl(\bbeta_0,\; \bm{G}^{-1}\bigr)$.
\end{enumerate}
For the forced model ($\beta_{\mathrm{pw}} \equiv 0$), (i) and (ii) hold for any $X \in L^2(\mathcal{D}) \setminus \{0\}$, with scalar
information $\bm{M}_n \to \|\mathcal{S}X\|_\mC^2 = \|X\|_{L^2}^2$ and limit law $\hat\beta_{\mathrm{fc},n} \xrightarrow{d} N\bigl(\beta_0, \|X\|_{L^2}^{-2}\bigr)$.
\end{theorem}

Since (i) is an exact Gaussian statement, Wald confidence regions built from $\bm{M}_n$ attain nominal coverage at every $n$ at which $\bm{M}_n$ is non-singular, and by (ii) they remain valid in the limit.
The limiting precision of $\hat\beta_{\mathrm{fc},n}$ is simply the squared $L^2$ norm of the covariate, which is completely insensitive to the operator parameters, and $\sigma_\epsilon^2$ affects only the speed of the convergence $\bm{M}_n \to \bm{G}$, but not the limit.
The scale-free parametrisation $\hat\beta_{\mathrm{fc}}^{\star} = \hat\beta_{\mathrm{fc}}/(\tau\kappa^{\alpha})$ of Section~\ref{sec:source-like} is a fixed linear transformation of $\hat\bbeta_n$, so under the conditions of Theorem~\ref{thm:infill}, $\hat\beta_{\mathrm{fc}}^{\star}$ inherits the exact normality and coverage, with limiting law $N\bigl(\beta_{\mathrm{fc},0}^{\star}, (\tau\kappa^{\alpha})^{-2}\|X\|_{L^2}^{-2}\bigr)$ in the forced model.

The requirement $X \in \HH$ in Theorem~\ref{thm:infill} ensures that $X$ can be evaluated at the observation locations and keeps the statement
simple. When $X$ is rough enough to fall outside $\HH$ but still has pointwise meaning, the pointwise channel behaves qualitatively differently, and in fact more favourably.

\begin{proposition}\label{prop:rough}
Let the observations follow Assumption~\ref{ass:design} under the hybrid model \eqref{eq:additive} with  coefficients $\bbeta_0 = (\beta_{\mathrm{pw},0}, \beta_{\mathrm{fc},0})$ and known $(\kappa, \tau, \alpha, \sigma_\epsilon^2)$, and suppose $X \in C(\barD)$ with $X \notin \HH = D(L^{\nicefrac{\alpha}{2}})$ (for $\alpha = d = 2$ this holds for every continuous $X \notin H^{2}(\mathcal{D})$). Then:
\begin{enumerate}[label=(\roman*),leftmargin=*,itemsep=1pt]
\item if $\E_w(\lambda^{\nicefrac{\alpha}{2}}) < \infty$, i.e., $X \in D(L^{\nicefrac{\alpha}{4}})$ (for $\alpha = 2$, $X \in H^1(\mathcal{D})$), the two channels decouple, $\vartheta(X) = 1$;
\item $\hat\beta_{\mathrm{fc}}$ retains the calibrated limit law of
      Theorem~\ref{thm:infill}, with limiting variance $\|X\|_{L^2}^{-2}$;
\item $\hat\beta_{\mathrm{pw}} \to \beta_{\mathrm{pw},0}$ in probability,
      so the pointwise coefficient is consistently estimable.
\end{enumerate}
\end{proposition}

Thus, when $m_0 = \beta_0 X$ with a rough $X$, the hybrid recovers the effect through (iii), and by (ii) keeps $\hat\beta_{\mathrm{fc}}$ centred at zero with calibrated variance.

\subsection{Robustness to operator misspecification}\label{sec:misspec}

An important difference between the operator-matched and additive estimators is their sensitivity to model misspecification. The next result covers an arbitrary data-generating mean in $\mathcal{H}$, which has a misspecified smoother as a  special case.

\begin{theorem}\label{thm:misspec}
Suppose the data are generated by $Y = m_0 + U$ with  $m_0 \in \HH$, and the hybrid model \eqref{eq:additive} is fitted by generalised least-squares (GLS) with known $(\kappa, \tau, \alpha, \sigma_\epsilon^2)$ and $\vartheta(X) > 0$. Then, under Assumptions~\ref{ass:A} and~\ref{ass:design}, $\hat\bbeta_n \xrightarrow{d} N\bigl(\bbeta_\infty,\; \bm{G}^{-1}\bigr)$, with $\bbeta_\infty =  \bm{G}^{-1} \bigl((X, m_0)_\mC, (\mathcal{S}X, m_0)_\mC\bigr)^\top$, where both components of $\bbeta_\infty$ are finite. In particular, for
$m_0 = \beta_0 \mathcal{S}_0 X$ with $\mathcal{S}_0 X \in \HH$:
if $\mathcal{S}_0 = \mathcal{I}$ then $\bbeta_\infty = (\beta_0, 0)^\top$, and if $\mathcal{S}_0 = \mathcal{S}$ then $\bbeta_\infty = (0, \beta_0)^\top$. For the forced model the limit is
$\beta_{\mathrm{fc},\infty} = (\mathcal{S}X, m_0)_\mC / \|X\|_{L^2}^2$,
which is finite for every $m_0 \in \HH$ and every
$X \in L^2(\mathcal{D}) \setminus \{0\}$.
\end{theorem}

\begin{corollary}\label{cor:confounding}
In the setting of Theorem~\ref{thm:misspec}, suppose the data are generated by $Y = \beta_0 X + Z + U$, where $Z \in \HH$ is an unmeasured confounder. Then
$$
\bbeta_\infty = (\beta_0,\, 0)^\top + \bm{G}^{-1}
\begin{pmatrix}
   (X, Z)_\mC & (\mathcal{S}X, Z)_\mC
\end{pmatrix}^\top,
$$
so the asymptotic bias of $\hat\bbeta_n$ is the Cameron--Martin projection of the confounder onto the two channels, and the calibrated limit law $\hat\bbeta_n \xrightarrow{d} N(\bbeta_\infty, \bm{G}^{-1})$ is unchanged.
\end{corollary}
In particular, a confounder aligned with the smoothed covariate, $Z = c\,\mathcal{S}X$, shifts only the forcing coefficient and leaves the pointwise channel unbiased.

\begin{corollary}\label{cor:signs}
Let $m_0 = \beta_0\, \sigma_0(L) X$, where $\sigma_0 : [\lambda_1, \infty) \to [0, \infty)$ is a bounded function, so that $\sigma_0(L)$ is a bounded non-negative smoother, with $\sigma_0(\lambda_j) > 0$ for at least one $j$ with $x_j \neq 0$, and assume $\beta_0 \neq 0$ and $X \in \HH$, so that $m_0 \in \HH$. Then the forced limit satisfies $\beta_{\mathrm{fc},\infty} = \beta_0 \tau \E_w\bigl[\lambda^{\nicefrac{\alpha}{2}} \sigma_0(\lambda)\bigr] \neq 0$ with $\operatorname{sign}(\beta_{\mathrm{fc},\infty}) = \operatorname{sign}(\beta_0)$. 
\end{corollary}
Thus, a wrong smoother rescales the forcing coefficient but cannot reverse its sign or drive it to zero or infinity.

\begin{remark}\label{rem:plugin}
In practice $(\kappa, \tau, \sigma_\epsilon)$ are estimated. By Proposition~\ref{prop:frontier}, under infill any estimator can at best converge to some $(\kappa^*, \tau_0, \sigma_{\epsilon,0})$ with $\kappa^*$ not necessarily $\kappa_0$.
GLS with the limiting plugged-in operator thus has a misspecified but equivalent covariance. As the bias computation in Theorem~\ref{thm:misspec} uses only the mean and the GLS weights, the projection-type limit for the asymptotic bias holds with $\mathcal{S}_0 = \mathcal{S}_{\kappa_0,\tau_0}$, $\mathcal{S} = \mathcal{S}_{\kappa^*,\tau_0}$ so the operator estimation error is a bounded, sign-preserving rescaling of $\hat\bbeta$. A coverage guarantee with \emph{jointly estimated} operator parameters is available under the expanding-domain asymptotics of Section~\ref{sec:incdom}, and Appendix~\ref{app:sim} suggests that empirical coverage is retained under infill in well-specified scenarios.
\end{remark}

\subsection{Expanding-domain asymptotics}\label{sec:incdom}

In contrast to the infill case, the coefficients and the operator parameters $(\kappa, \tau, \sigma_\epsilon^2)$ are jointly consistently estimable under expanding-domain asymptotics, at a fixed $\alpha$. We state the result for the stationary version of the model on growing observation windows. Let $\psi = (\beta_{\mathrm{pw}}, \beta_{\mathrm{fc}}, \kappa, \tau, \sigma_\epsilon^2)$ with true value $\psi_0$ in the interior of a compact parameter set $\Psi \subset \R^2 \times (0,\infty)^3$, let $\alpha > d/2$ be fixed and known, and let the model be \eqref{eq:additive} on $\R^d$ with $U$ the stationary Mat\'ern field and $\mathcal{S} = \tau^{-1}(\kappa^2 - \Delta)^{-\nicefrac{\alpha}{2}}$ on $\R^d$.

\begin{assumption}\label{ass:B}
The observation design is the lattice $\delta \mathbb{Z}^d \cap \mathcal{D}_n$ with fixed spacing $\delta > 0$ and $\mathcal{D}_n = [0, a_n]^d$,
$a_n \to \infty$, so that $n \asymp (a_n/\delta)^d$. The covariate $X$ is a bounded measurable function on $\R^d$ satisfying for every $h \in \R^d$ and $u \in [0,\delta)^d$,
$$
\hat r_n(u, h) := \frac{1}{n}\sum_{i: s_i + u, s_i + u + h \in \mathcal{D}_n} X(s_i + u)\,X(s_i + u + h) \longrightarrow R_X(h),
$$
where the limit does not depend on $u$ and is continuous in $h$. By Bochner's theorem, $R_X(h) = \int_{\R^d} e^{\mathrm{i}\xi^\top h}\, dW_X(\xi)$ for a finite spectral measure $W_X$ on $\R^d$, which we assume is nonzero and not concentrated on a single modulus $\|\xi\|$. Equivalently, the expanding-domain channel-separation index of Remark~\ref{rem:vartheta-infty} satisfies $\vartheta_\infty(X) > 0$.
\end{assumption}

Assumption~\ref{ass:B} admits a rich class of covariates (Remark~\ref{rem:ass-examples}, Appendix~\ref{app:proofs}) and is a spatial analogue of classical time-series conditions on regressors \citep{Grenander1954, GrenanderRosenblatt1957}. Its non-degeneracy requirement is the expanding-domain form of the channel-separation condition of Theorem~\ref{thm:geometry}:

\begin{remark}\label{rem:vartheta-infty}
Let $\E_w$ and $\Var_w$ denote mean and variance under $w := W_X/W_X(\R^d)$, the spectral measure of $X$ normalised to a probability measure. Define the index
$$
\vartheta_\infty(X) = \frac{\Var_w\bigl(\lambda^{\nicefrac{\alpha}{2}}\bigr)}{\E_w(\lambda^{\alpha})}
= 1 - \frac{\bigl[\E_w\bigl(\lambda^{\nicefrac{\alpha}{2}}\bigr)\bigr]^2}{\E_w(\lambda^{\alpha})} \in [0,1],
$$
where $\lambda(\xi) = \kappa^2 + \|\xi\|^2$, and $\vartheta_\infty(X) = 1$ when $\E_w(\lambda^{\alpha}) = \infty$. 
Here $\vartheta_\infty(X) > 0$ if and only if $W_X$ is not concentrated on a single sphere $\|\xi\| = r$, which is a version of the rank condition on regressors of \citet{Grenander1954} for $(X, \mathcal{S}X)$.
\end{remark}

\begin{theorem}\label{thm:incdom}
For the stationary model specified above and under Assumption~\ref{ass:B}, the maximum likelihood estimator $\hat\psi_n$ satisfies
$$
   \hat\psi_n \xrightarrow{p} \psi_0,
   \qquad
   \sqrt{n}\,(\hat\psi_n - \psi_0)
   \xrightarrow{d}
   N\bigl(0, \bm{I}(\psi_0)^{-1}\bigr),
$$
where $\bm{I}(\psi) = \lim_{n\to\infty} n^{-1} \bm{I}_n(\psi)$ is the limiting normalised Fisher information, which is positive definite, and $\bm{I}_n(\psi)$ is the Fisher information of $(\widetilde Y_1, \ldots, \widetilde Y_n)$. 
In particular the regression coefficients $(\hat\beta_{\mathrm{pw}}, \hat\beta_{\mathrm{fc}})$ are $\sqrt{n}$-consistent and asymptotically normal jointly with the operator parameters $(\hat\kappa, \hat\tau)$ and the nugget $\hat\sigma_\epsilon^2$. 
\end{theorem}

With $\bbeta = 0$, Theorem~\ref{thm:incdom} reduces to a zero-mean Mat\'ern field observed with measurement error and yields $\sqrt n$-consistency and joint asymptotic normality of $(\kappa, \tau, \sigma_\epsilon^2)$ of the type established by \citet{Bachoc2014} for noise-free observations. Further, $\beta_{\mathrm{fc}} = 0$ gives the same for an additive model \eqref{eq:additive-intro}.
This also complements the fixed-domain study of \citet{TangZhangBanerjee2021}, in which the nugget is recovered at the $\sqrt n$ rate and the microergodic parameter at rate $n^{1/(2 + 4\nu/d)}$.

\subsection{Testing operator matching and two-scale models}\label{sec:matching-test}

The hybrid model \eqref{eq:additive} has two nested submodels of interest. Setting $\beta_{\mathrm{fc}} = 0$ returns the additive model, and setting $\beta_{\mathrm{pw}} = 0$ returns the forced. The hybrid is also contained in a larger family, obtained by replacing $\mathcal{S}$ by a smoother with range $\kappa_\mu$,
\begin{equation}\label{eq:twoscale}
   Y = \beta_{\mathrm{pw}}X + \beta_{\mathrm{fc}}\,\mathcal{S}_{\kappa_\mu}X + U,
\end{equation}
which we refer to as a two-scale model, and where operator matching is the restriction $\kappa_\mu = \kappa$. This gives three likelihood-ratio tests: hybrid-versus-additive by \eqref{eq:additive-intro} against \eqref{eq:additive}, hybrid-versus-forced by \eqref{eq:forced-spde} against \eqref{eq:additive} and operator matching by \eqref{eq:twoscale} against \eqref{eq:additive}. 
Under the expanding-domain asymptotics of Section~\ref{sec:incdom}, Theorem~\ref{thm:incdom} gives the hybrid-versus-additive and hybrid-versus-forced tests  the usual $\chi^2_1$ calibration, and its argument extends to the two-scale model with free $\kappa_\mu$ provided $\beta_{\mathrm{fc},0} \neq 0$ (which is required for $\kappa_\mu$ to be identified under the null), so the matching test is calibrated in that regime as well. Under infill asymptotics no such calibration is available, as the null involves non-microergodic parameters (Proposition~\ref{prop:frontier}) and the mean is nonlinear in $\kappa_\mu$. Which regime better describes an application depends on the ratio of the domain extent to the fitted ranges and when it is small we consider the $p$-values as descriptive. Appendix~\ref{app:sim} reports an empirical size check of both the matching and the hybrid-versus-additive tests. 

The two-scale model also recovers, as special cases, the diffusion smoother of \citet{LindmarkAndersonThorson2026}, the mean structure of the kernel averaged predictors of \citet{HeatonGelfand2011}, and the fixed buffer-averaging of covariates standard in land-use regression \citep{Hoek2008, Briggs1997}. 

The results of this section resolve the collapse-to-zero failure noted by \citet{BolinWallin2026} for the additive estimator, where a regularity mismatch between $X$ and $\HH$ drives the GLS coefficient to zero under infill. That paper proposes pre-smoothing $X$ by an operator mapping into $\HH$ but leaves the choice of operator open, and shows that a poorly chosen smoother can make the coefficient diverge. Operator matching answers this by construction, with no tuning and no new parameters.
Two-stage remedies such as pre-smoothing or Spatial+ condition on a separately estimated adjustment of $X$, so the uncertainty of that step is not propagated by default, whereas here it enters through the joint estimation of Theorem~\ref{thm:incdom}.
Moreover, unlike an identifying assumption, matching is a restriction that can be tested, as above.

\section{Discrete space and the relation to spatial econometrics}\label{sec:econometrics}

The structure of the operator-matched model is classical in econometrics. 
The autoregressive-distributed-lag model $y_t = \beta_1 y_{t-1} + \gamma_0 x_t + \gamma_1 x_{t-1} + \varepsilon_t$ can be written as $y_t = a\,x_t + b\,\mathcal{S}x_t + \mathcal{S}\varepsilon_t$ with $\mathcal{S} = (1 - \beta_1 B)^{-1}$ for the lag operator $Bx_t = x_{t-1}$, $a = -\gamma_1/\beta_1$ and $b = \gamma_0 + \gamma_1/\beta_1$. This has the same form as \eqref{eq:additive} with the one-sided lag polynomial in place of the symmetric $L^{-\nicefrac{\alpha}{2}}$.
The restriction $b=0$ is the \emph{common factor restriction} which returns the additive counterpart, and it is classically tested rather than assumed \citep{Sargan1980, HendryMizon1978}.
The complementary restriction $\gamma_1 = 0$ returns the corresponding forced model, so the first two tests of Section~\ref{sec:matching-test} have  antecedents there. The operator matching test itself does not, since the two-scale model \eqref{eq:twoscale} lies outside the autoregressive-distributed-lag family.

The spatial versions of these models, in which the lag operator is replaced by a matrix of spatial weights, are the subject of spatial
econometrics, where a regression whose covariate enters both directly and through the operator that generates the dependence has been standard since
\citet{Burridge1981} and \citet{Anselin1988}. 
The rest of this section  makes the connection explicit and discusses what our continuum formulation adds besides being an extension from discrete to continuous space.

\subsection{The lattice family and the correspondence}\label{sec:econ-family}

Spatial econometrics works with $N$ areal units, on which the response is a vector $\bm{y} \in \R^N$ and the covariates an $N \times K$ matrix
$\bm{X}$. Spatial structure enters through an $N \times N$ matrix $\bm{W}$ of non-negative weights with zero diagonal, usually
row-standardised, which encodes which units are neighbours. 
A standard specification is the spatial error model (SEM) $\bm{y} = \bm{X}\bm{b} + \bm{u}$ with $\bm{u} = \lambda \bm{W}\bm{u} + \bm{\varepsilon}$, $\bm{b}$ are regression coefficients, $\lambda$ sets the strength of dependence, and $\bm{\varepsilon}$ has independent Gaussian components.   
Another is the spatial lag or spatial autoregressive model (SAR) $\bm{y} = \rho\bm{W}\bm{y} + \bm{X}\bm{b} + \bm{\varepsilon}$, where 
$\rho$ sets the strength of dependence. A third is the spatial Durbin model (SDM)  
\begin{equation}\label{eq:sdm}
\bm{y} = \rho \bm{W}\bm{y} + \bm{X}\bm{b} + \bm{W}\bm{X}\bm{\theta} + \bm{\varepsilon}.
\end{equation}
See \citet{Anselin1988} or \citet{LeSagePace2009} for details of these and other variants.

With $\bm{S}_\rho := (\bm{I} - \rho\bm{W})^{-1}$, \eqref{eq:sdm} is equivalently $\bm{y} = -\rho^{-1}\bm{X}\bm{\theta} + \bm{S}_\rho \bm{X}\bigl(\bm{b} + \rho^{-1}\bm{\theta}\bigr) + \bm{S}_\rho\bm{\varepsilon}$, which is the representation \eqref{eq:additive} of the hybrid model with $\mathcal{S}$ replaced by $\bm{S}_\rho$, with $\beta_{\mathrm{pw}} = -\rho^{-1}\bm{\theta}$ and $\beta_{\mathrm{fc}} = \bm{b} + \rho^{-1}\bm{\theta}$.
For a row-standardised nearest-neighbour $\bm{W}$ on a regular lattice this is more than an analogy, since $\bm{I} - \rho\bm{W} =
(1-\rho)\bm{I} + \rho(\bm{I} - \bm{W})$ with $\bm{I} - \bm{W}$ a graph Laplacian, so $\bm{S}_\rho$ discretises $L^{-1}$ and the SDM is a discrete hybrid model with $\alpha = 2$.
Further, SEM corresponds to the additive model and SAR to the forced.
The restriction $\bm{\theta} = -\rho\bm{b}$ is the common factor restriction corresponding to $\beta_{\mathrm{fc}} = 0$, and $\bm{\theta} = \bm{0}$ is $\beta_{\mathrm{pw}} = 0$.  Our  hybrid-versus-additive test of Section~\ref{sec:matching-test} is thus the continuous-space form of the common factor test of \citet{Burridge1981}.

Usually three summaries of the impact matrix $\bm{S}_\rho(\bm{I}b_r + \bm{W}\theta_r)$ of the $r$th covariate are reported: the
\emph{direct} impact, defined as the average of its diagonal; the \emph{indirect} impact or spillover, defined as the average row sum of its off-diagonal part; and the \emph{total} impact, which is their sum \citep[Section~2.7]{LeSagePace2009}. 
In our continuous-space formulation $\beta_{\mathrm{pw}}$ corresponds to the direct impact, the standardised forcing coefficient $\beta^{\star}_{\mathrm{fc}} = \beta_{\mathrm{fc}}/(\tau\kappa^{\alpha})$ to the indirect impact and the total impact is their sum.
By Proposition~\ref{prop:scalefree} the direct and total impacts are the two endpoints of the scale-varying coefficient, which gives the scale decomposition an impact interpretation.

\citet[Theorems~3.1--3.2]{Lee2004} provides an expanding domain asymptotic theory for discrete models, for an a priori specified weight matrix and a response observed without measurement error, resting on a non-multicollinearity condition on the covariate and its spatial lag. In Theorem~\ref{thm:geometry} we characterised this by showing that the channels are jointly identified if and only if $\vartheta(X) > 0$. This sharpens the recurring objection that the specifications are hard to distinguish \citep[p.~177]{GibbonsOverman2012}. 
Under infill \citet{Lee2004} finds no general results available and gives examples of an inconsistent estimator, which Proposition~\ref{prop:frontier} shows to be unavoidable for any estimator, while Theorem~\ref{thm:infill} shows what is achievable.

\subsection{What the operator formulation adds}\label{sec:econ-operator}

In the lattice family $\bm{W}$ is exogenous input, and that literature has long identified the choice as a weakness \citep[Section~3.1.4]{Anselin1988}. 
\citet{Wall2004} showed that the covariance implied by a chosen matrix is opaque, and therefore recommended modelling the covariance directly, as in geostatistical models. The matrix does not decide that covariance either, as there is no one-to-one correspondence between the weights and the implied covariance \citep[Section~4.2]{Anselin2002}.
Further, the matrix is not derived from an underlying mechanism, as it is here, and ``research containing a theoretical derivation of the spatial weight matrix is scarce'' \citep[p.~50]{TanKesinaElhorst2025}. 
Also, the implied spillovers need not decay with distance, as the first-order multiplier $\rho\bm{b} + \bm{\theta}$ of the SDM can exceed $\bm{b}$ in absolute value, so decay requires a constraint on the coefficients that is seldom checked \citep{AnselinSereniniAmaral2026}.
Operator matching addresses all these issues: 
The operator is derived from a transport--relaxation balance rather than selected, and the resulting kernel is translation-invariant with a known Whittle--Mat\'ern form. 
As the smoother is the square root of the covariance operator of the random effect, the ``weights'' and the covariance determine each other, and $\mathcal{S}$ is an isometry of $L^2(\mathcal{D})$ onto the Cameron--Martin space (Proposition~\ref{prop:rep}). Further, the spillovers decay with distance on $\mathbb{R}^d$ as the Green's function of $L^{\nicefrac{\alpha}{2}}$ is strictly positive and decreasing in distance.

Another issue is that on a lattice the weights must be aggregated along with the data, and it is well known that the fitted spatial structure depends on the zoning \citep[Section~5.2]{Anselin2002}. Working in continuous space removes the problem, since averaging \eqref{eq:additive} over any set of regions leaves $\beta_{\mathrm{pw}}$, $\beta_{\mathrm{fc}}$ and the operator unchanged. 
Discretising the operator also provides a reasonable default in discrete space: for lattice data, the finite element discretisation below, and for areal data the counterpart $\bm{K} = \kappa^2\bm{C} + \bm{S}$, which takes the place of $\bm{I} - \rho\bm{W}$, with $\bm{C}$ a diagonal matrix of region areas and $\bm{S}$ the graph Laplacian whose weights are the length of the boundary shared by two regions divided by the distance between their centroids.
Because $\bm{C}$ and $\bm{S}$ discretise fixed continuum objects, subdividing or merging regions changes the matrices but leaves $\kappa$ and the coefficients estimating the same quantities, whereas $\rho$ is defined relative to a chosen $\bm{W}$ and is not comparable across partitions.

Estimating rather than fixing the weights has also been investigated. Several studies parametrise $\bm{W}$ by a distance decay, most often  negative-exponential or inverse-distance, and estimate it jointly with
the response parameters \citep[][and the references therein]{TanKesinaElhorst2025}. The kernel's shape is typically kept fixed, which essentially corresponds to fixing $\alpha=2$ for us. For non-lattice data row-standardisation makes each weight depend on the unit's whole neighbour configuration, so the fitted decay is not a translation-invariant kernel and the implied covariance remains opaque.
Further, as the decay is assumed for $\bm{W}$ rather than for $\bm{S}_\rho$, which determines the spillovers, the estimated range is not that of the spillovers. Estimating $\kappa$ in $\bm{K}$ above avoids this, as the fitted range is then that of the spillovers and of the implied covariance by construction. 
The lattice family also allows for non-geographic or asymmetric weights, which we have not addressed.

\section{Inference and computation}\label{sec:inference}

In this section, we introduce two ways of fitting the hybrid model to data. The first combines a numerical approximation of $\mathcal{S}X$ with a
standard covariance-based representation of $U$. The second uses a finite element approximation of the entire model, which reduces computational costs and inherits the entire computational machinery of the standard \texttt{R-INLA} pipeline for additive spatial regression, so the parameters can be estimated by maximum likelihood, as in Section~\ref{sec:app-temp-alt}, or by fully Bayesian inference as in Section~\ref{sec:app-no2}. 

\subsection{A covariance-based approach}\label{sec:covariance-impl}

The hybrid model can be fitted with general-purpose Gaussian-process software whenever the stationary version of Section~\ref{sec:incdom} is used, where $U$ is given a Mat\'ern covariance function $\varrho$. 
Because $L$ is self-adjoint, $\mathcal{S}$ is the operator square root of $U$'s covariance operator, $\mathcal{S} = \mC^{1/2}$.
Equivalently, $\mathcal{S}$ is convolution with a kernel $g$ whose Fourier transform is proportional to $(\kappa^2+\|\xi\|^2)^{-\nicefrac{\alpha}{2}}$, while the Mat\'ern spectral density is proportional to $(\kappa^2+\|\xi\|^2)^{-\alpha}$. With $\nu_g = \nicefrac{(\alpha-d)}{2}$ and $\nu = \alpha - \nicefrac{d}{2}$, 
\begin{equation}\label{eq:kernels-a2}
g(r) = \frac{2^{1-\nu_g}\,(\kappa r)^{\nu_g} K_{\nu_g}(\kappa r)}
      {(4\pi)^{\nicefrac{d}{2}} \Gamma(\nicefrac{\alpha}{2})
      \kappa^{2\nu_g} \tau},
\qquad
\varrho(r) = \frac{2^{1-\nu}\,(\kappa r)^{\nu} K_{\nu}(\kappa r)}
      {(4\pi)^{\nicefrac{d}{2}}\,\Gamma(\alpha) \kappa^{2\nu} \tau^{2}},
\end{equation}
where $K_\nu$ is the modified Bessel function of the second kind and $K_{-\nu} = K_{\nu}$.

Given $X$ on a grid $\{u_k\}$ with quadrature weights $\{\Delta_k\}$, the forced regressor at the observation sites is the Riemann sum
$(\mathcal{S}X)(s_i) = \sum_k g(\|s_i - u_k\|) X(u_k) \Delta_k$. For $\alpha \le d$ the kernel is unbounded at the origin. However, 
as the singularity is integrable, the Riemann sum is valid if evaluation inside the nearest cell is avoided, e.g., by flooring $\|s_i - u_k\|$ at half the grid spacing or by integrating $g$ over the nearest cell. The additive representation \eqref{eq:additive} is then the Gaussian linear mixed model
\begin{equation}\label{eq:cov-lmm}
\widetilde{\bm{Y}} \mid \bbeta \sim N\big(\beta_{\mathrm{pw}}\bm{X} + \beta_{\mathrm{fc}}\,\mathcal{S}\bm{X}, \bm{\Sigma} + \sigma_\epsilon^2 \bm{I}\big), \quad \bm{\Sigma} = [\varrho(\|s_i - s_j\|)]_{i,j=1}^n
\end{equation}
in which, $\mathcal{S}\bm{X}$ is fixed for fixed $(\kappa,\alpha)$, $\bbeta = (\beta_{\mathrm{pw}},\beta_{\mathrm{fc}})^\top$ is obtained by GLS and profiled out, and $(\kappa,\alpha,\tau,\sigma_\epsilon^2)$ are estimated by maximum likelihood or REML. The two-scale model can naturally be fitted using this approach by giving $g$  its own range.

As $X$ is typically not available continuously, one might treat this method as exact given the covariate resolution. However, one could characterise the effect of the quadrature error  similarly to the results of Section~\ref{sec:fem-error}. We leave this for future work as the finite element implementation below is generally preferable.

\subsection{Finite element approximation}
As the covariance-based implementation requires an $O(n^3)$ dense factorisation of $\bm{\Sigma} + \sigma_\epsilon^2\bm{I}$ per likelihood evaluation, we now introduce a more computationally efficient finite element alternative.
Suppose that $\widetilde{Y}_i = Y(s_i) + \epsilon_i$ for $s_i\in\mathcal{D}, \epsilon_i \sim \mathrm{N}(0,\sigma_{\epsilon}^2)$ and $i=1,\ldots,n$, where $Y$ follows \eqref{eq:hybrid-spde} with $\alpha=2$. By \eqref{eq:additive}, we can write this as 
\begin{equation}\label{eq:system}
\widetilde{Y}_i = \beta_{\mathrm{pw}}X(s_i) + Z(s_i) + \epsilon_i 
\end{equation}
where $Z = \beta_{\mathrm{fc}}\mathcal{S}X + U$ is the forced field solving \eqref{eq:forced-spde} with $\beta = \beta_{\mathrm{fc}}$, so that $Y = \beta_{\mathrm{pw}}X + Z$. Let $\{\psi_j\}_{j=1}^m$ be the piecewise-linear basis on a triangulation $\mathcal{T}$ of $\mathcal{D}$, which we take to be polygonal throughout Section~\ref{sec:inference} so that $\mathcal{T}$ covers $\mathcal{D}$ exactly. 
This is for simplicity only, to avoid the additional boundary-approximation terms that a curved $\partial\mathcal{D}$ would contribute \citep[Ch.~10]{BrennerScott2008}. We let $\bm{C}$ and $\bm{S}$ be the mass and stiffness matrices with elements $C_{ij} = (\psi_i, \psi_j)_{L^2}$ and $S_{ij} = (\nabla \psi_i, \nabla \psi_j)_{L^2}$. Then, Galerkin discretisation of \eqref{eq:forced-spde} yields the sparse linear system $\bm{K} (\tau \bm{z}) = \beta_{\mathrm{fc}}\bm{C}\bm{x} + \bm{w}$, where $\bm{K} = \kappa^2 \bm{C} + \bm{S}$, $\bm{z}, \bm{x}$ are the finite element coefficient vectors for $Z$ and $X$ and $\bm{w} \sim N(0, \bm{C})$ is the projection of $\W$ onto the basis. Introducing the observation matrix $\bm{A}$ with elements $A_{ij} = \psi_j(s_i)$ and the vectors $\widetilde{\bm{Y}}$ and $\bm{X}$ with the observations and covariate evaluated at the observation locations, we obtain $\widetilde{\bm{Y}} | \bm{z} \sim \mathrm{N}(\beta_{\mathrm{pw}}\bm{X} + \bm{A}\bm{z}, \sigma_{\epsilon}^2\bm{I})$, and $\bm{z} \sim \mathrm{N}(\beta_{\mathrm{fc}}\mathcal{S}_{\mathrm{FEM}}\bm{x}, \bm{Q}^{-1})$ as the approximation of \eqref{eq:system}, where $\bm{Q} = \tau^2 \bm{K}^{\top}\bm{C}^{-1}\bm{K}$ and $\mathcal{S}_{\mathrm{FEM}} = \tau^{-1}\bm{K}^{-1}\bm{C}$.
This is a standard latent Gaussian model and relative to the additive \texttt{SPDE} model ($\beta_{\mathrm{fc}} = 0$), the single algorithmic change is the term $\beta_{\mathrm{fc}}\mathcal{S}_{\mathrm{FEM}}\bm{x}$ in the latent mean, evaluated by one sparse solve with $\bm{K}$. The likelihood thus has the same computational scaling as the additive finite element model. See \citet{LBR2022} for details on likelihood computations and computational scalings for SPDE models like these. 

For a general $\alpha>d/2$, the latent model in \eqref{eq:system} is $Z = \beta_{\mathrm{fc}}\tau^{-1} L^{-\nicefrac{\alpha}{2}}X + U$
where $U$ is a centred Whittle--Mat\'ern field with covariance operator $\tau^{-2}L^{-\alpha}$. 
To discretise this model, the finite element discretisation is combined with a rational approximation of the fractional powers of $L$ \citep{BolinKirchner2020}, which replaces the operator $L^{\nicefrac{\alpha}{2}}$ by a low-order rational approximation $\hat r(L) = p(L) q(L)^{-1}$ with $p, q$ polynomials in $L$. 
We apply it to the mean term and use a separate rational approximation for the covariance operator of $U$ as in \citet{xiong2022}, which avoids an unnecessary approximation of $U$ in cases such as $\alpha = 1$. 
The resulting method is implemented in the \texttt{rSPDE} package \citep{rSPDE} and integrates with the standard \texttt{R-INLA} pipeline, with the mean term added in the latent field as for $\alpha = 2$, and facilitates joint likelihood-based inference of $\alpha$. See the references above for details on the likelihood computation.

\subsection{Discretisation error}\label{sec:fem-error}

The results of Section~\ref{sec:properties} concern the continuous model, whereas the numerical implementation above fits its Galerkin discretisation. This subsection states the consequence of this for the estimator, which we establish for $\alpha = 2$ to keep the proofs brief. 
Throughout, $\{V_h\}_{h > 0}$ denotes the continuous piecewise-linear finite element spaces on a quasi-uniform family of triangulations of
$\mathcal{D}$ with mesh width $h$, $P_h : L^2(\mathcal{D}) \to V_h$ is the $L^2$-projection, and $L_h : V_h \to V_h$ is the Galerkin
operator defined by $(L_h v, w)_{L^2} = \kappa^2 (v, w)_{L^2} + (\nabla v, \nabla w)_{L^2}$ for $v, w \in V_h$. For $\alpha = 2$ the discrete model replaces $\mathcal{S}$ by $\mathcal{S}_h = \tau^{-1} L_h^{-1} P_h$, with matrix representation $\mathcal{S}_{\mathrm{FEM}}$, and $U$ by
$U_h = \tau^{-1} L_h^{-1} P_h \W$, whose covariance kernel we denote by $\varrho_h$.

Note that $\mathcal{S}_h$ and $U_h$ are built from the \emph{same} discrete operator $L_h$, so the fitted model is itself an operator-matched
hybrid model on $V_h$: the discrete analogues of Proposition~\ref{prop:rep} (with the isometry $\|\mathcal{S}_h X\|_{\mC_h} = \|P_h X\|_{L^2}$ for the discrete Cameron--Martin norm $\|v\|_{\mC_h} = \tau\|L_h v\|_{L^2}$), of Theorem~\ref{thm:geometry}, and of Theorem~\ref{thm:infill} hold verbatim on $V_h$ for the covariate represented in $V_h$, with the eigenpairs of $L_h$ replacing $(\lambda_j, e_j)$. Within the
fitted model the failure modes of \citet{BolinWallin2026} therefore cannot occur at any mesh width. The effect of fitting data from the continuous model is instead a quantitative shift of the estimates, which Theorem~\ref{thm:fem-estimator} bounds.

\begin{theorem}\label{thm:fem-estimator}
Let the data be generated by the continuous hybrid model of Theorem~\ref{thm:infill} with $\alpha = 2$, and let $\hat\bbeta_{n,h}$ be the estimator of Theorem~\ref{thm:infill} for the discretised model, with $\mathcal{S}$ and $\varrho$ replaced by $\mathcal{S}_h$ and $\varrho_h$. 
Set $\varepsilon_h = \|\mathcal{S}X - \mathcal{S}_hX\|_{L^\infty}$, $\delta_h = \|\varrho - \varrho_h\|_{L^\infty(\barD^2)}$, $\eta_{n,h} = n\,\varepsilon_h + n^2 \delta_h$. Then $\varepsilon_h = O(h^{2 - \nicefrac{d}{2}})$ and $\delta_h = O(h^{2 - \nicefrac{d}{2}})$ and constants $C, \eta_0 > 0$ exist such that for all pairs $(n, h)$ with $\eta_{n,h} \le \eta_0$ and $n$ large enough:
\begin{enumerate}[label=(\roman*),leftmargin=*,itemsep=1pt]
\item $\hat\bbeta_{n,h} \sim N\bigl(\bbeta_0 + b_{n,h}, \bm{V}_{n,h}\bigr)$, with $\|b_{n,h}\| \le C \sqrt{n}\,\varepsilon_h \le C\,\eta_0/\sqrt{n}$ and $\|\bm{V}_{n,h} - \bm{M}_n^{-1}\| \le C\, \eta_{n,h}$, where $\bm{M}_n$ is the information matrix of the continuous model;
\item if $(h_n)$ satisfies $\eta_{n, h_n} \to 0$, then $\hat\bbeta_{n, h_n} \xrightarrow{d} N(\bbeta_0, \bm{G}^{-1})$, and the Wald intervals built from the discrete information $\bm{M}_{n,h_n}$ are asymptotically exact.
\end{enumerate}
The same statements hold for the forced model with the single
regressor $\mathcal{S}_h X$, for any $X \in L^2(\mathcal{D}) \setminus \{0\}$ and without the condition $\vartheta(X) > 0$.
\end{theorem}

The condition $\eta_{n, h_n} \to 0$ holds if $h_n = o(n^{-4/(4-d)})$, i.e., $h_n = o(n^{-2})$ in $d = 2$. Sharper rates are likely available, which would relax this to $h_n = o(n^{-2/(4-d)})$ (Remark~\ref{rem:fem-fractional} in Appendix~\ref{app:proofs}). By (i), the bias requires only $\sqrt{n}\,\varepsilon_h$ to be small relative to the limiting $O(1)$ standard deviations, suggesting the milder scaling $h \propto n^{-1/(4-d)}$. As $\varepsilon_h$ is computable for a given mesh, this can be checked, and the adequacy of the meshes used below is assessed by comparison with the covariance-based method.

\section{Temperature and altitude across climate regimes}
\label{sec:app-temp-alt}
As a first application we consider regression of temperature on altitude. This is, statistically, a spatial-confounding problem as the regression slope of temperature on elevation is contaminated by a smooth spatial field. Standard methods such as restricted spatial regression or Spatial+ return an adjusted slope. The operator-matched model instead resolves the slope into a local thermodynamic lapse rate ($\hat\beta_{\mathrm{pw}}$) and a horizontal-mixing contribution ($\hat\beta_{\mathrm{fc}}^{\star}$) acting on a scale $\kappa^{-1}$. 
What makes the example a genuine test of the method, rather than just a check for predictive quality, is that the local lapse rate and the mixing scale have been estimated independently in the atmospheric-science literature, so the method's output can be checked against independent physical estimates. 
We arrange three regimes along a gradient of increasing horizontal mixing: dry-continental Colorado, maritime Norway, and the Norwegian winter. The effective forcing $\hat\beta_{\mathrm{fc}}^{\star}\hat g_X$ should then be small in Colorado, where spring surface lapse rates track the free-air rate \citep{PepinLosleben2002}, and grow across the gradient, as maritime and especially polar-winter temperatures are set increasingly by advected air \citep{PepinNorris2005}.
A supporting simulation study is presented in Appendix~\ref{app:sim}. 

\subsection{Data and model comparison}
Figure~\ref{fig:data-mesh} shows the data, which for Colorado are $n=213$ observations from NOAA Cooperative Observer stations of mean
March--May climatology over 1960--1990, distributed with the \texttt{fields} package \citep{fields}. The data for Norway are $n=64$ GHCN-M v4 stations, restricted to below $65^\circ\mathrm{N}$ to remove the latitudinal trend, and either the  $1991$--$2020$ annual means or the December--February means. 

\begin{figure}[t]
\centering
\includegraphics[width=\textwidth]{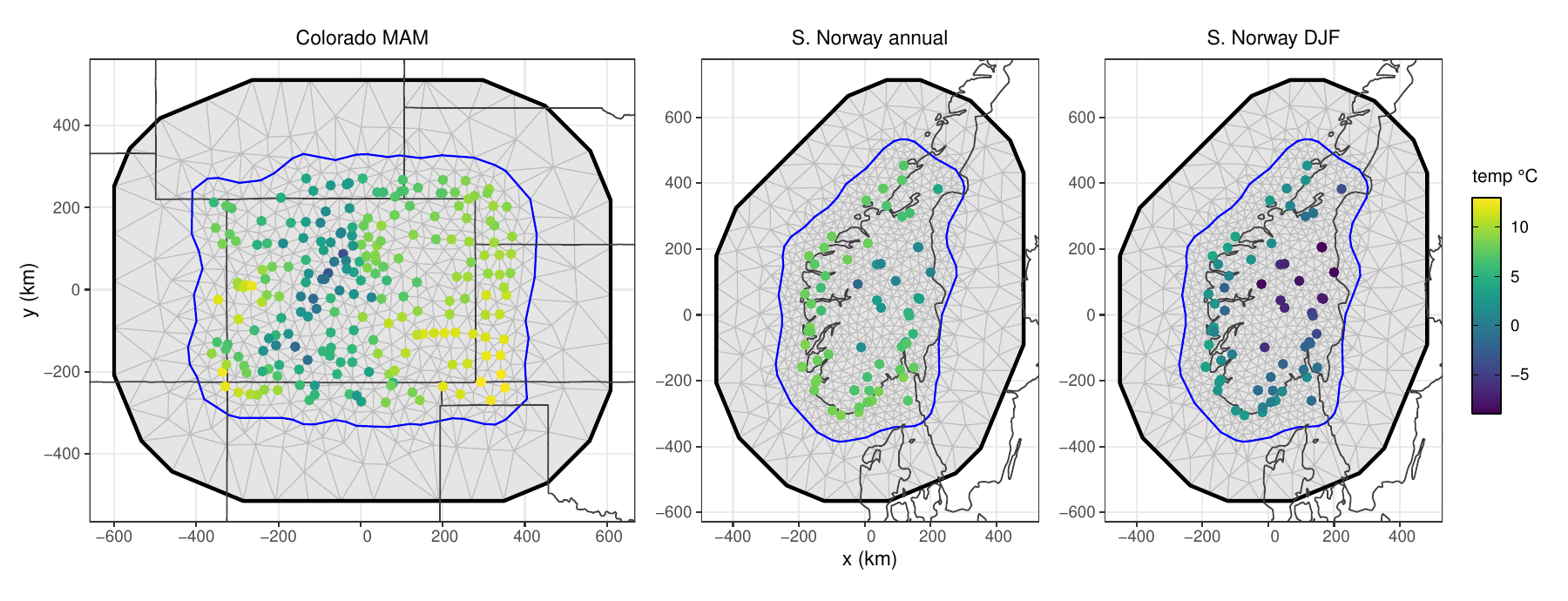}
\caption{Data and triangulations for the three temperature regimes.}
\label{fig:data-mesh}
\end{figure}

The additive, hybrid, and forced models are compared against ordinary least squares (OLS), a fixed-bandwidth version of the empirical-smoother approach of \citet{BolinWallin2026} (BW; elevation pre-smoothed with a Gaussian kernel with $20\,$km bandwidth), and Spatial+  (a thin-plate spline is fitted to elevation through REML via \texttt{mgcv} \citep{Wood2011} and the residual is used as covariate). We additionally fit the two-scale model \eqref{eq:twoscale} to test the operator matching assumption, and for Norway we also fit a `coast' model, which is the hybrid augmented with a pointwise distance-to-coast covariate, to test whether the steep winter forcing coefficient merely stands in for an unmeasured continentality gradient. Each model is fitted by maximum likelihood using both the covariance-based implementation and the finite element implementation on a common triangulation per dataset (Figure~\ref{fig:data-mesh}). Elevations over the meshes are taken from a digital elevation model, a $205\times119$ grid for Colorado \citep{fields} and a $30$~arc-second SRTM-derived model for Norway \citep{FarrEtAl2007, geodata}.  The smoothness $\alpha$ is either fixed at $2$ or estimated jointly. All models include an intercept as an ordinary fixed effect, which is not included in the forcing channel as $\vartheta(\mathbf{1}) = 0$ by Theorem~\ref{thm:geometry}(i). The same convention is used in Section~\ref{sec:app-no2}. We assess predictive performance by a buffered leave-group-out cross-validation (CV) \citep{LiuVanNiekerkRue2025} following \citet{xiong2022}, where each station is predicted after deleting every training station within a radius $d$, so that increasing $d$ checks progressively longer-range extrapolation. We report RMSE and the logarithmic score \citep{GneitingRaftery2007} as functions of $d$, both negatively oriented, with $d=0$ the ordinary leave-one-out (LOO) score.

\subsection{Results and physical interpretation}
The estimated smoothness parameters and a comparison of $10$-fold CV performance for the general smoothness and $\alpha=2$ models are shown in Table~\ref{tab:tempalt-hyper} of Appendix~\ref{app:details}. As fixing $\alpha=2$ is physically motivated, and as the general-smoothness models have similar performance to the $\alpha=2$ models, we focus on $\alpha=2$ below. For $\alpha=2$, the hypothesis of operator matching in the two-scale model is not rejected in either Colorado ($p$ value $0.07$) or Norway~DJF ($p=0.48$). It is weakly rejected ($p=0.03$) in Norway annual; however, the two-scale model provides no cross-validated gain in any of the regimes (Table~\ref{tab:tempalt-hyper}), so the parsimonious matched hybrid is preferred throughout. The fitted ranges are not small relative to the domain extents here, so the application is closer to the infill setting. The empirical size check in Appendix~\ref{app:sim} nevertheless shows that the test performs as it should at this sample size.

The CV results are shown in Figure~\ref{fig:temp-cvcurve}.
At $d=0$ (LOO) the additive baseline, hybrid, and the coast-augmented hybrid models are statistically indistinguishable in every regime (paired Diebold--Mariano tests give $p>0.3$), which is expected as each held-out station has near neighbours, so the task is essentially short-range interpolation. 
The models separate as the prediction range grows in Norway, where the coastal hybrid performs best.
Table~\ref{tab:lapse-decomp} gives the coefficients for the OLS, additive, and hybrid models. For the hybrid the apparent altitude effect splits into a local pointwise component $\hat\beta_{\mathrm{pw}}$ and a horizontal-mixing channel summarised by $\hat\beta_{\mathrm{fc}}^{\star}$ (the sign-reversed equilibrium lapse rate $-\gamma$ of Section~\ref{sec:state-like}). 
As the standard errors in Table~\ref{tab:lapse-decomp} show, $\hat\beta_{\mathrm{fc}}^{\star}$ is only weakly identified (its standard error is large as $\kappa$ is not consistently estimable under infill), so we read the channel through the variance-matched $\hat\beta_{\mathrm{fc}}^{\star}\hat g_X$, which is the fitted magnitude of the forcing contribution, 
$\|\hat\beta_{\mathrm{fc}}\mathcal{S}X\|_{L^2}/\|X\|_{L^2}$, and thus is less sensitive to $\kappa$.
\begin{figure}[t]
\centering
\includegraphics[width=\textwidth]{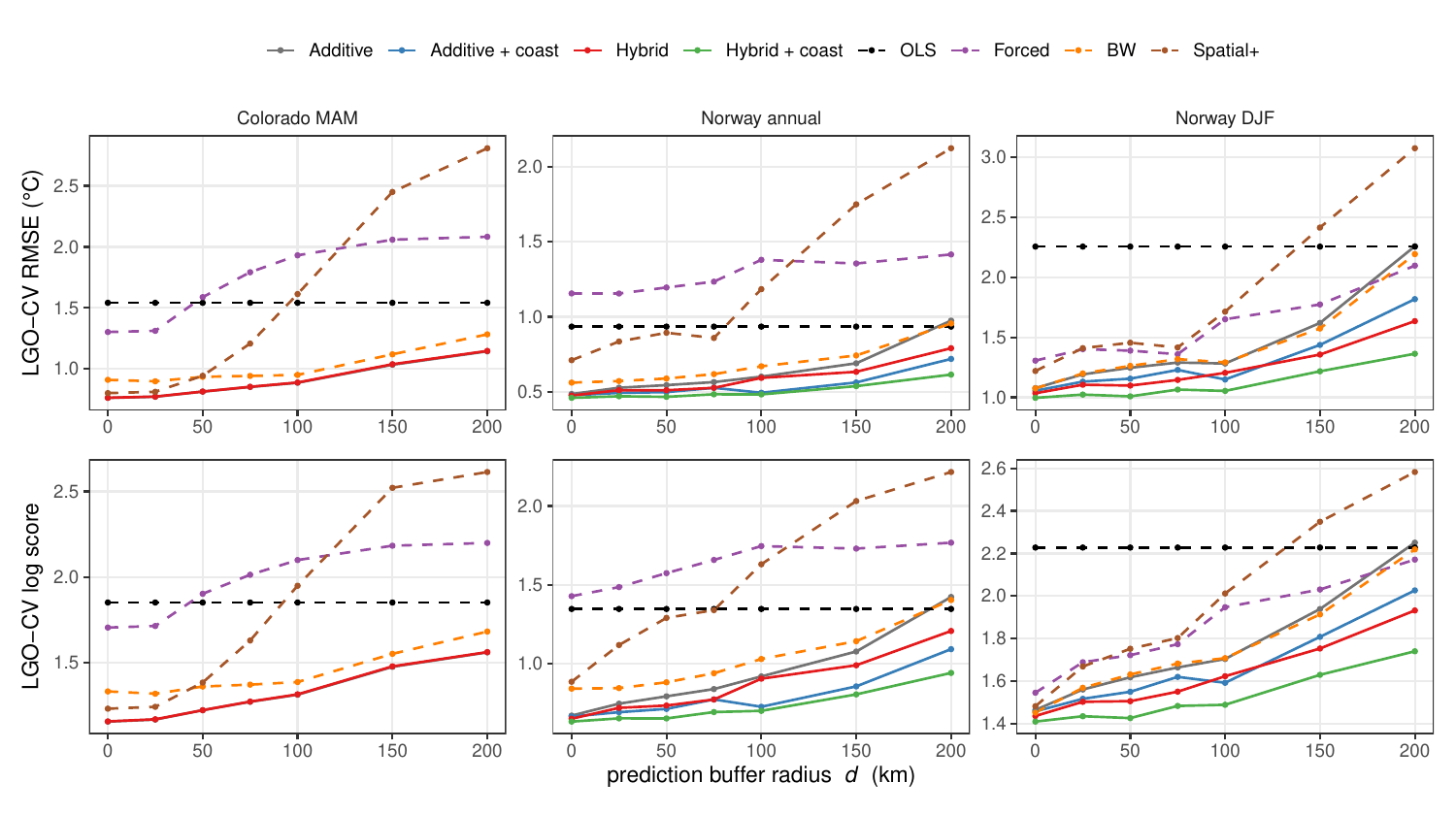}
\caption{Buffered leave-group-out CV across the three regimes.}
\label{fig:temp-cvcurve}
\end{figure}

\begin{table}[t]
\begin{center}
\setlength{\tabcolsep}{4pt}
\resizebox{\textwidth}{!}{%
\begin{tabular}{lrrrrrrrr}
\toprule
Regime & $\hat\beta_{\mathrm{OLS}}$ & $\hat\beta_{\mathrm{add}}$
       & $\hat\beta_{\mathrm{pw}}$ & $\hat\beta_{\mathrm{fc}}^{\star}$
       & $\hat g_X$ & $\hat\beta_{\mathrm{fc}}^{\star}\hat g_X$
       & $\hat\vartheta(X)$ & $2\Delta\ell$\\
\midrule
Colorado MAM & $-5.43$ & $-6.63$ & $-6.69$ {\footnotesize$(0.24)$} & $2.06$ {\footnotesize$(2.58)$} & $0.48$ & $0.98$ {\footnotesize$(1.22)$} & $0.89$ & $0.6$\\
Norway annual & $-7.86$ & $-5.99$ & $-5.65$ {\footnotesize$(0.34)$} & $-7.76$ {\footnotesize$(3.97)$} & $0.40$ & $-3.07$ {\footnotesize$(0.96)$} & $0.90$ & $8.9$\\
Norway DJF & $-10.54$ & $-4.78$ & $-4.10$ {\footnotesize$(0.75)$} & $-29.65$ {\footnotesize$(15.84)$} & $0.35$ & $-10.28$ {\footnotesize$(2.37)$} & $0.91$ & $16.5$\\
\bottomrule
\end{tabular}}\end{center}
\caption{Parameters for the OLS, additive, and hybrid models with $\alpha=2$. Here, $\hat\beta_{\mathrm{OLS}}$ and $\hat\beta_{\mathrm{add}}$ are the OLS and additive-GLS slopes of $T$ on $h$, $\hat\beta_{\mathrm{pw}}$ is the pointwise lapse rate, $\hat\beta_{\mathrm{fc}}^{\star}$ is the scale-free forcing coefficient, $\hat g_X\in(0,1]$ is the effective gain, $\hat\beta_{\mathrm{fc}}^{\star}\hat g_X$ is the variance-matched effective forcing, $\hat\vartheta(X)$ is the channel-separation index, and $2\Delta\ell$ is the likelihood-ratio statistic for the hybrid against the additive baseline. Parenthetical values are delta-method standard errors. 
The coefficients have units  $^\circ\mathrm{C\,km^{-1}}$.}
\label{tab:lapse-decomp}
\end{table}

Colorado is the regime in which local column thermodynamics should dominate, and the fits show this. The additive and hybrid models tie for best, while the forced model and Spatial+ lose by smoothing away the sharp altitude contrast. The decomposition explains the tie as $\hat\beta_{\mathrm{add}}$ and $\hat\beta_{\mathrm{pw}}$ are similar and the forcing channel is not statistically significant. The fitted $\hat\beta_{\mathrm{pw}}\approx -6.7\,^\circ\mathrm{C\,km^{-1}}$ is close to the standard environmental lapse rate of $-6.5\,^\circ\mathrm{C\,km^{-1}}$, and to the March--May value implied by the Colorado Front Range temperature normals of \citet{PepinLosleben2002}.
For Norway annual, the hybrid and additive models have similar short range predictive performance, but the additive increasingly falls behind the hybrid model at longer ranges. The hybrid splits the signal into a local lapse rate and a non-negligible mixing channel. The variance-matched $\hat\beta_{\mathrm{fc}}^{\star}\hat g_X$ ($95\%$ CI $[-5.0,-1.2]$) is about half the local term, and the split is well identified ($\hat\vartheta(X)=0.90$) and significant ($2\Delta\ell=8.9$, $p<0.005$). Adding a distance-to-coast covariate leaves the channel essentially intact ($\hat\beta_{\mathrm{fc}}^{\star}\hat g_X=-2.5$, $[-4.0,-0.9]$), so it is not merely a proxy for continentality.

Norway DJF is the regime where the forced mechanism is strongest as expected. The hybrid and coast models degrade far more slowly than the additive baseline as the prediction range grows, and in winter the forcing channel is the larger of the two long-range effects, outweighing continentality. The pure forced model, which has no pointwise lapse term, is not competitive.
The decomposition here is the most striking of the three: the OLS slope of $-10.5\,^\circ\mathrm{C\,km^{-1}}$ exceeds any free-atmosphere lapse rate, reflecting a relationship set by spatial structure rather than by local thermodynamics. The additive GLS attenuates this to $-4.78$, and the hybrid partitions it into a modest local lapse rate ($\hat\beta_{\mathrm{pw}}=-4.1$) and a dominant mixing channel with 
$\hat\beta_{\mathrm{fc}}^{\star}\hat g_X=-10.3$ ($[-14.9,-5.6]$). Controlling for distance to coast leaves it largely unchanged ($-9.3$, $[-13.5,-5.0]$). The small pointwise slope is what one expects in winter, when temperature at a station is only loosely tied to how temperature falls with height in the atmosphere above it, and is set largely by which air has been carried to the station from elsewhere \citep{PepinNorris2005}. 

By Section~\ref{sec:physical}, $\kappa^{-1}=\sqrt{D/\rho}$, where values of $D$ determined from atmospheric diffusion data lie between $1.5\times10^{4}$ and $1.2\times10^{5}\,\mathrm{m^{2}\,s^{-1}}$ \citep{Gifford1982}, and near-surface temperature anomalies relax on a timescale
of two to five days \citep{Hasselmann1976, MeyerKantz2019}, placing $\kappa^{-1}$ between $50$ and $230\,$km, a range that contains the fitted values for both Norway DJF and annual ($205$ and $179\,$km). For Colorado, the forcing channel is not active so this $\kappa^{-1}$ interpretation does not hold.

The results above are based on the finite element approach. 
Table~\ref{tab:temp-impl} of Appendix~\ref{app:details} shows that these agree on every well-identified quantity with those from the covariance-based approach of Section~\ref{sec:covariance-impl}. 
Colorado has the largest differences, as the forcing channel is not significant, which makes $\kappa$ more difficult to identify.
As the datasets are small, the two approaches are also comparable in cost per likelihood evaluation except in Colorado where the finite element likelihood is cheaper.

In summary, the effective forcing shifts across regimes as the physical motivation of Section~\ref{sec:state-like} predicts and agrees with the finding that surface lapse rates are shallowest in winter \citep{Rolland2003} and that maritime ranges have shallower rates \citep{MinderMoteLundquist2010}.
The predictive gain shows the forcing channel is needed, and the agreement of $\hat\beta_{\mathrm{pw}}$, and of $\kappa^{-1}$ where the channel is identified, with independent physical estimates supports its interpretation as a real effect, which is statistically identified.

\section{Air-quality monitoring}\label{sec:app-no2}

The previous application is an example of the state-like mechanism of Section~\ref{sec:state-like}. As a second application we consider air-quality monitoring, where a pollutant is emitted by spatially distributed sources, which is an example of the source-like mechanism of Section~\ref{sec:source-like}. The operator-matched forcing is then the steady-state dispersion of the emission field, and the hybrid model's two channels are the near-source and urban-background contributions of \citet{LenschowEtAl2001}. As the temperature application was done in a frequentist setting, we now instead consider fully Bayesian inference.

\subsection{Data and model comparison}
The response is the calibrated annual $\mathrm{NO}_2$ concentration at
$n=17{,}886$ sites of the CurieuzeNeuzen Vlaanderen citizen-science
campaign \citep{CurieuzeNeuzen}, in which volunteers across Flanders measured $\mathrm{NO}_2$ with duplicate Palmes diffusion tubes
in May~2018, calibrated against the Flemish reference network. 
A standard approach to mapping air pollution from such data is land-use regression (LUR) \citep{Briggs1997, Hoek2008}. A LUR model regresses the measured concentration, by OLS, on predictors derived from emission inventories and surrounding land use, population density, and urban and natural land-cover fractions. Because the spatial scale on which each predictor acts is not known a priori, each variable is evaluated over a range of circular buffers and the most predictive buffer sizes are retained by supervised forward selection \citep{Beelen2013}.

As emission sources, we use road traffic, taken as the OpenStreetMap \citep{OpenStreetMap} major-road network with each road class weighted by an approximate annual average daily traffic volume (of order $70{,}000$ vehicles\,day$^{-1}$ for motorways down to $4{,}000$ for tertiary roads), and the reported $\mathrm{NO}_x$ releases of the $56$ facilities in the domain \citep{EPRTR}. 
The road-class weights are nominal figures reflecting the typical traffic hierarchy and as the covariate is standardised, only their relative magnitudes matter.
Each source is aggregated to an emission-density field on the analysis triangulation, where the value at each node is the
average over its dual cell, giving forcing covariates $X_{\mathrm{traf}}$ and $X_{\mathrm{ind}}$. 
We also include gridded residential data \citep{WorldPop} as a proxy for diffuse low-level emissions. Figure~\ref{fig:no2-data} shows the data.

\begin{figure}[t]
\centering
\includegraphics[width=\textwidth]{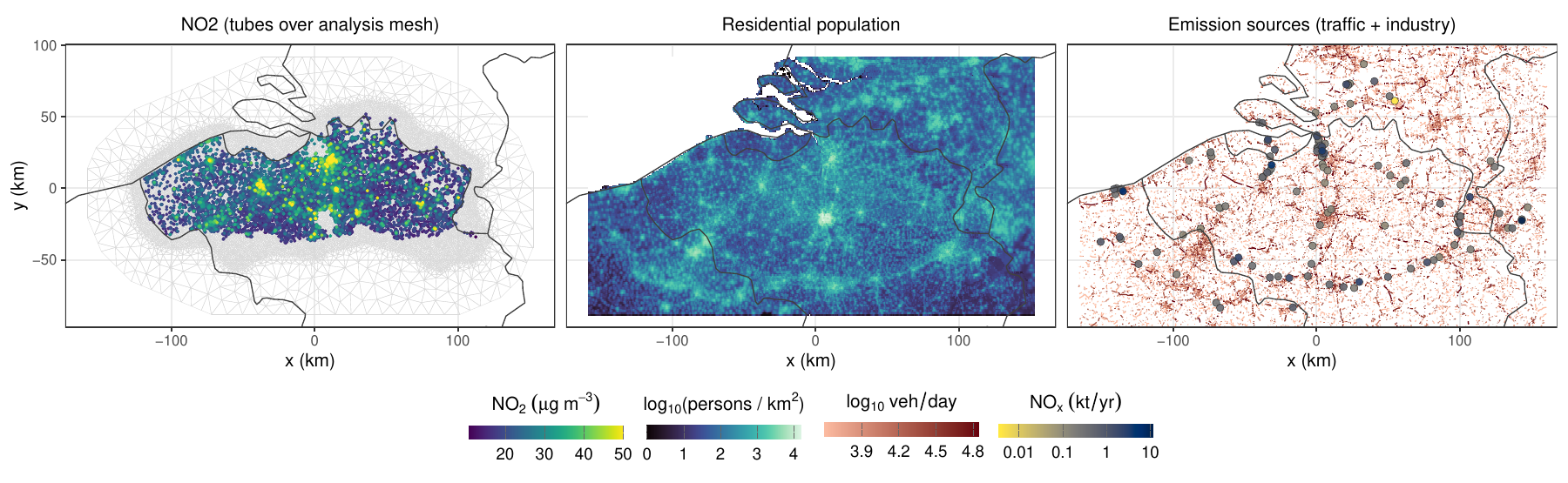}
\caption{The measurements coloured by annual $\mathrm{NO}_2$ (capped from $75.3\,\mu\mathrm{g\,m^{-3}}$ to $50$, which $59$ tubes exceed), mesh, population density, road network shaded by traffic intensity, and E-PRTR industrial facilities coloured by reported $\mathrm{NO}_x$ release.}
\label{fig:no2-data}
\end{figure}

There are a few reasonable choices for operator-matched models in this case. The simplest are forced or hybrid models that take all three covariates as forcings. 
This is, however, difficult to justify physically as it would assume that the pollution from each source disperses at the same rate, which is unlikely as industrial pollutants are released from tall chimneys and traffic pollution at ground level. A more realistic model has one hybrid field for industry and a second that takes population and traffic as forcings, with their sum giving the total pollution. We refer to this as an industry-split two-field model.
A more complex model has separate hybrid models for all three pollution sources, which we refer to as the three-field model. In these multi-field specifications operator matching is imposed per component, and we note that the theory of Section~\ref{sec:properties} covers the common-operator case only (see Section~\ref{sec:disc}).

We fit the forced, hybrid, three-field and industry-split two-field models with $\alpha=2$ on a common mesh, all by Bayesian inference through \texttt{R-INLA}. The finite element approach is used here as the size of the dataset makes the covariance-based approach too computationally heavy.
We also fit four reference models. The first is a ``field only'' model, which is a  Whittle--Mat\'ern field with an intercept. The second is a LUR model with emission and land-use buffer covariates. We include 22 potential covariates, formed from the data in Figure~\ref{fig:no2-data} as pointwise or buffered traffic, industry and population variables together with land-use fractions and the distance to the coast. The full list is given in Appendix~\ref{app:no2}. The covariates that are actually used are selected through supervised forward selection. We also include an additive model \eqref{eq:additive-intro} with the same covariates as the LUR model (LUR+field), and a two-scale model \eqref{eq:twoscale} to test the operator matching assumption. 

All models are fitted with weakly informative priors, specified in Appendix~\ref{app:no2}. Models are again compared through grouped CV, where each observation is predicted based on all data except all observations within a radius of $d$~km. The CV is done based on fixed parameters from the fits, and the forward selection of buffer covariates for the LUR and LUR+field models is only done once on the full data.

\subsection{Results}

\begin{table}[t]
\centering\small
\setlength{\tabcolsep}{3pt}
\begin{tabular}{lccccc}
\toprule
Model & \shortstack{Pre-set\\candidates} & \shortstack{Parameters\\(cov.)} & \shortstack{CV\\log-score} & \shortstack{CV\\RMSE} & \shortstack{$\Delta$ log-score\\vs.\ best (SE)}\\
\midrule
Field only (kriging)  & $0$ & $4\,(0)$  & $2.924$ & $4.46$ & $0.0227\,(0.0021)$\\
Standard LUR & $22\,(5)$ & $7\,(5)$  & $2.979$ & $4.76$ & $0.0779\,(0.0072)$\\
LUR ${}+{}$ field        & $22\,(5)$ & $9\,(5)$  & $2.906$ & $4.40$ & $0.0046\,(0.0018)$\\
Operator-matched forced & $0$ & $7\,(3)$  & $2.910$ & $4.41$ & $0.0085\,(0.0011)$\\
Operator-matched hybrid & $0$ & $10\,(3)$ & $2.909$ & $4.41$ & $0.0077\,(0.0009)$\\
Two-scale forcing       & $0$ & $11\,(3)$ & $2.908$ & $4.41$ & $0.0067\,(0.0010)$\\
Three-field (per source) & $0$ & $14\,(3)$ & $2.902$ & $4.38$ & $0.0004\,(0.0002)$\\
Two-field (industry split) & $0$ & $12\,(3)$ & $2.902$ & $4.38$ & ---\\
\bottomrule
\end{tabular}
\caption{Model complexity and predictive skill for the $\mathrm{NO}_2$
analysis. The first column is the number of buffer and land-use covariates in the LUR forward selection, with the number selected in parentheses. The second is the number of parameters, including the number of covariates (in parentheses) in the models. The LOO CV log-score and RMSE are both negatively oriented. The final column shows the paired difference in mean log-score from the best model (industry-split two-field), with a standard error in parentheses from a spatial block bootstrap ($12\times12$ blocks, $2000$ resamples) that accounts for spatial dependence between the scores. Each difference is at least two SE, so all differences are statistically significant.}
\label{tab:no2-models}
\end{table}

Table~\ref{tab:no2-models} shows LOO CV results together with the number of parameters for each model. The results of the full CV are shown in Figure~\ref{fig:no2-cvcurve}.
The addition of a spatial field is important for short-range prediction, as the standard LUR model falls well behind the other models. Further, the field-only model deteriorates at long-range prediction, and the forced, hybrid, LUR+field, and two-scale forcing models have  similar predictive quality. 
In contrast to the temperature application, the fitted ranges are here small relative to the domain, so this application is effectively in the expanding-domain setting of Section~\ref{sec:incdom}. The scale-free decomposition (Table~\ref{tab:no2-decomp}) shows that all three forcing coefficients are positive and well separated from zero, and traffic and population act through both a near-source and a dispersed channel, whereas industry's near-source coefficient is indistinguishable from zero, so industry enters only through dispersion. Allowing the forcing its own range  does not improve the fit and the operator matching assumption is not rejected, which supports the matched specification and its physical interpretation. 

\begin{table}[t]
\centering\small
\begin{tabular}{lcccccc}
\toprule
Source & $\hat\beta_{\mathrm{pw}}$ & $\hat\beta_{\mathrm{fc}}^{\star}$
       & $\hat g_X$ & $\hat\beta_{\mathrm{fc}}^{\star}\hat g_X$
       & $\hat\vartheta(X)$ & $\hat\vartheta_m$\\
\midrule
Traffic     & $0.60\,(0.50,\,0.70)$ & $3.07\,(2.33,\,3.86)$ & $0.67$ & $2.07\,(1.57,\,2.60)$ & $0.48$ & $0.34$\\
Industry    & $0.01\,(-0.07,\,0.09)$ & $3.10\,(1.50,\,4.90)$ & $0.34$ & $1.04\,(0.50,\,1.64)$ & $0.23$ & $0.23$\\
Population  & $0.54\,(0.40,\,0.67)$ & $2.16\,(1.36,\,3.06)$ & $0.77$ & $1.66\,(1.04,\,2.36)$ & $0.52$ & $0.37$\\
\bottomrule
\end{tabular}
\caption{Scale-free decomposition of the three forcing channels in the operator-matched hybrid model, with $95\%$ posterior credible intervals in parentheses, the marginal channel-separation index $\hat\vartheta(X)$ and its joint counterpart $\hat\vartheta_m$ from Remark~\ref{rem:multivariate}.}
\label{tab:no2-decomp}
\end{table}

\begin{figure}[t]
\centering
\includegraphics[width=\textwidth]{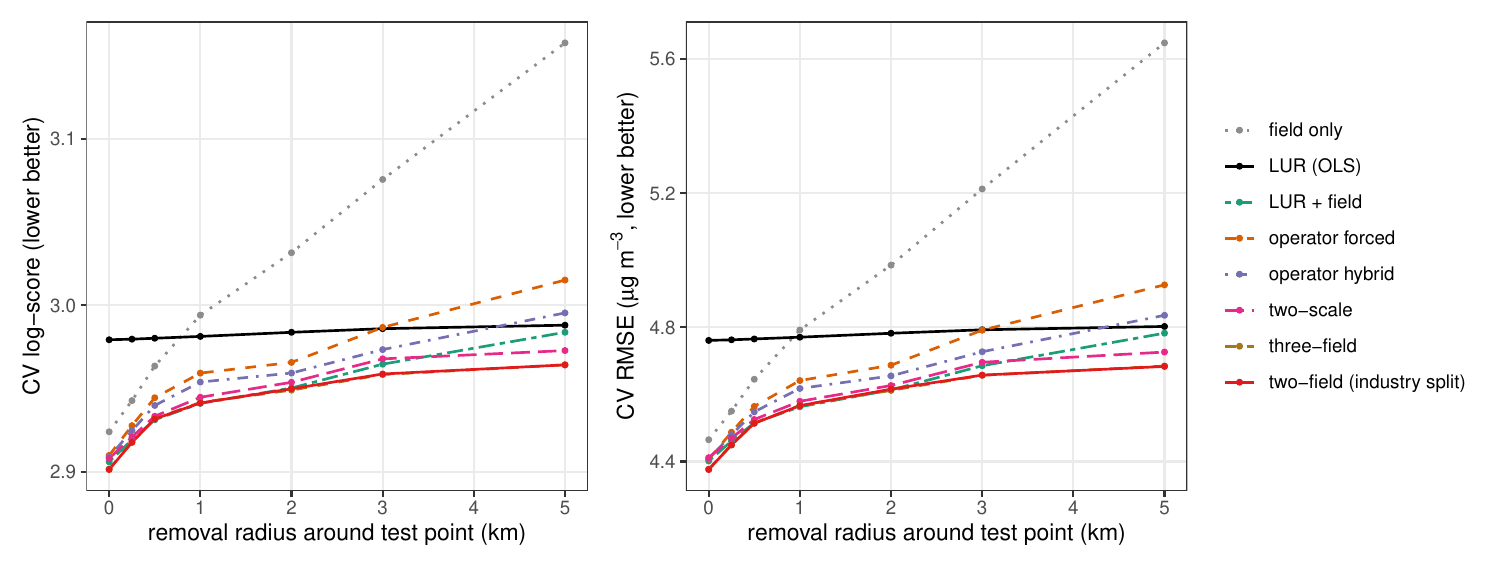}
\caption{Buffered leave-group-out CV results for the $\mathrm{NO}_2$ application.}
\label{fig:no2-cvcurve}
\end{figure}
The three-field and two-field models have essentially identical performance, outperform all other models, and are the only two that clearly beat standard LUR at long-range prediction.
The three-field model estimates $\kappa^{-1}\approx 32\,$km for the industrial field while traffic and population scales are both $\approx 1.5\,$km, confirming that the elevated industrial sources disperse an order of magnitude further than the ground-level sources, and the pooled $4.3\,$km of the hybrid is a compromise that no individual process occupies.
The traffic and population forcing channels are the most collinear pair in the Cameron--Martin geometry (cross-channel correlation $0.47$,
Remark~\ref{rem:multivariate}).
This affects identification rather than predictive performance, but it motivates the two-field model, which merges the two redundant local scales and reproduces the three-field's predictive skill with two fewer parameters, fitting a local field at $\kappa^{-1}\approx 1.4\,$km and an industrial field at $32\,$km.

We can thus conclude that the industry-split two-field model is the best choice, outperforming both LUR and LUR+field, without a need to prespecify buffer ranges or do manual forward selection. 
It has been noted that buffer sizes in LUR models should ideally be chosen to reflect known dispersion patterns \citep{Hoek2008}, and our models take this a step further as the smoother is the dispersion operator itself, and its scale is estimated jointly rather than selected from a candidate list.

\section{Discussion}\label{sec:disc}

We have introduced operator-matched spatial regression as a new tool for spatial data analysis, implemented in the \texttt{rSPDE} package, in which $X$ acts through the same dynamics that generate the spatial dependence in $U$. 

A growing literature treats spatial confounding explicitly as a problem of causal identification with unmeasured spatial confounders 
\citep{SchnellPapadogeorgou2020}.
There, identification is typically obtained by restricting the joint structure of the exposure and the unmeasured confounder. Examples include that confounding dissipates at small scales \citep{GuanEtAl2023}, that the exposure has variation not shared by the confounder \citep{GilbertEtAl2024, GilbertOgburnDatta2025}, that the covariate and the spatial effect are explicitly correlated random fields \citep{MarquesKneibKlein2022}, or that the adjustment acts at a selected spatial scale \citep{KellerSzpiro2020}.
Operator matching is a structural assumption of this type, but one derived from a transport--relaxation mechanism rather than imposed for identification, and it specifies through the fitted operator how the influence of $X$ is distributed across scales (Proposition~\ref{prop:scalefree}). Unlike an untestable identifying assumption, however, matching can be tested against a two-scale model.
By Corollary~\ref{cor:confounding}, the asymptotic bias of each channel under an unmeasured confounder $Z$ is the corresponding coordinate of the Cameron--Martin projection of $Z$ onto $(X, \mathcal{S}X)$. That up-weights small scales, so a smooth $Z$ enters mainly through the forcing coordinate. The hybrid therefore localises where confounding bias concentrates.
A full characterisation of the causal estimands the decomposition recovers, and of what the forcing channel identifies when the confounder shares the
operator with $U$, remains open.

Further open problems include joint estimation of $(\alpha, \zeta)$ in the two-exponent extension of Appendix~\ref{app:colored}; a complete theory for the multi-field models of Section~\ref{sec:app-no2}, where each covariate is matched to the operator of its own random-effect component; and coverage guarantees under infill asymptotics when the operator parameters are estimated.
Several extensions could also be considered. For example, in advection-dominated transport or in problems with  non-stationarities, a modified operator (e.g.\ adding an advection term to $L$, having a function-valued $\kappa$, or replacing $\Delta$ by $\nabla\cdot(\mathbf{H}(s)\nabla)$ for a matrix-valued function $\mathbf{H}$) may be preferable. Ecological diffusion, in which $\Delta$ is replaced by $\nabla^2(\mu(s)\,\cdot\,)$ with a habitat-dependent motility $\mu$, is particularly interesting as this would extend the approach to species-distribution and disease-spread models \citep{HefleyEtAl2017}.
One may also consider non-Gaussian versions by replacing the driving noise with non-Gaussian noise such as type-G L\'evy noise \citep{bw20}, and the approach naturally extends to non-Euclidean domains such as metric graphs \citep{BolinSimasWallin2026} and manifolds \citep{LRL2011} as Whittle--Mat\'ern fields may also be defined there.
The construction is likewise not tied to Gaussian, point-referenced responses. Because the forcing enters the model as the mean of the Gaussian field, it can be used directly in the latent-Gaussian-model machinery of \texttt{R-INLA}, and using a Poisson or binomial likelihood yields operator-matched regression for count data at no additional modelling cost.

Finally, in the introduction we listed several other papers that considered smoothed covariates, which, like our proposed models, can be viewed as continuous versions of specifications in the lattice family in Section~\ref{sec:econometrics}. 
Those connections are rarely drawn, even when the lattice family is implemented \texttt{R-INLA}, which also supports SPDE models \citep{GomezRubioBivandRue2021}. 
\citet[Section~3]{PimPratesCarvalho2026} also observe that the econometrics literature is largely absent from work on spatial confounding, even though
the SDM has been derived from an omitted-variable argument \citep[Section~2.2]{LeSagePace2009}.
Our discussion in Section~\ref{sec:econometrics} aims to bring these fields closer, and econometrics-oriented extensions (e.g., continuous-space SLX or SDEM models) are another interesting topic for future work.

\begin{appendices}

\section{Coloured driving noise and a two-exponent extension}\label{app:colored}

In Section~\ref{sec:physical} we assumed that the unresolved forcing was spatially white. However, unresolved processes may have spatial correlation of their own. An extension of the forced model is thus one where the driving noise has a spatial covariance operator proportional to $L^{-\zeta}$, $\zeta \ge 0$ (the white-noise case being $\zeta = 0$). The colour of the noise raises the smoothness of the random effect while leaving the forcing channel unchanged, decoupling the two: 

\begin{proposition}\label{prop:colored}
In the setting of Proposition~\ref{prop:timeavg}, let the driving
noise be $W^{\zeta}$ with spatial covariance operator $L^{-\zeta}$,
$\zeta \ge 0$. Then:
\begin{enumerate}[label=(\roman*),leftmargin=*,itemsep=1pt]
\item $\E[\bar Y_T] = \gamma_0 A^{-1}X = \gamma_0 D^{-1}L^{-1}X$ for every $T$ and every $\zeta$;
\item the instantaneous stationary field has covariance operator
      $(2D)^{-1}L^{-(1+\zeta)}$, and is thus a Whittle--Mat\'ern field with
      $\alpha = 1 + \zeta$ when $\zeta > \nicefrac{d}{2} - 1$;
\item the covariance operator of $\sqrt{T}(\bar Y_T - \gamma_0 A^{-1}X)$
      equals $D^{-2}L^{-(2+\zeta)}(I + R_T)$ with $\|R_T\| \le (\rho T)^{-1}$ in operator norm, so the time average is asymptotically Whittle--Mat\'ern with $\alpha = 2 + \zeta$ and $\tau = D$;
\item $\beta_{\mathrm{fc}}^{\star} = (\beta_T/\tau_T)\kappa^{-2} = \gamma_0/\rho$, and the mean $\gamma_0 A^{-1}X$ lies in the Cameron--Martin space of the limit field in (iii) if and only if $X \in D(L^{\zeta/2})$.
\end{enumerate}
\end{proposition}

Non-integer $\zeta$ thus produces fractional $\alpha$ directly. Therefore, for $\zeta > 0$ the time-averaged response follows a hybrid model in which the forcing channel lies in the Cameron--Martin space $D(L^{1+\zeta/2})$ of
the random effect when $X \in D(L^{\zeta/2})$, which is a weaker requirement than the $X \in D(L^{1+\zeta/2})$ that the additive formulation would need, but no longer automatic. Fitting the single-$\alpha$ operator-matched model to such data is an instance of the operator misspecification covered by Theorem~\ref{thm:misspec}.

The coloured-noise model also returns $\alpha = 2$ itself, for instantaneous responses, through a specific and physically meaningful choice of $\zeta$. 
By Proposition~\ref{prop:colored}(ii) with $\zeta = 0$, the equilibrated white noise-driven balance has spatial covariance proportional to $L^{-1}$. If the unresolved forcing is itself such an equilibrated field on a finer timescale, refreshed at an event-driven rate so that it acts as temporally white noise carrying its equilibrium spatial covariance, then $\zeta = 1$ and the instantaneous response has $\alpha = 2$. This procedure can be iterated. 

\begin{corollary}\label{cor:cascade}
For $m \ge 1$ define a hierarchy of stationary fields as follows:
level $1$ is the balance of Proposition~\ref{prop:timeavg} with spatially white driving noise, and for $m \ge 2$, level $m$ solves the same balance driven by temporally white noise whose spatial covariance is proportional to that of the equilibrium of level $m - 1$. Then the instantaneous field at level $m$ is a Whittle--Mat\'ern field with $\alpha = m$ whenever $m > \nicefrac{d}{2}$, and its time average over windows long relative to the relaxation time is asymptotically Whittle--Mat\'ern with $\alpha = m + 1$. 
In particular, $\alpha = 2$ arises either as the time average when $m=1$ or as the instantaneous response when $m = 2$.
\end{corollary}

Thus, $\alpha$ counts the number of passes the innovations make through the operator, a structure that fits several geophysical applications with processes occurring at different temporal scales. For example boundary-layer turbulence equilibrates in hours, mesoscale fields in days, the synoptic environment in weeks, and the forcing felt by each level is refreshed by events (cloud fields, precipitation, synoptic disturbances) at a rate that is roughly scale-independent.

\section{Simulation study}\label{app:sim}

This appendix reports a simulation study with  three goals:
(i) demonstrate unbiased and calibrated estimation of $\beta$ under the well-specified forced-SPDE model in agreement with Theorem~\ref{thm:infill};
(ii) quantify performance under operator misspecification, showing a bounded, sign-preserving rescaling rather than collapse as Theorem~\ref{thm:misspec} and Corollary~\ref{cor:signs} predict;
(iii) compare against the additive model and existing remedies for self-confounding.

\subsection{Setup}
We fix the spatial domain to $\mathcal{D} = [0, 10]^2$ and have $n \in \{50, 100, 250, 500, 1000\}$ observation locations sampled uniformly in the domain. The observations are $\widetilde Y_i = Y(s_i) + \epsilon_i$ with $\epsilon_i \sim N(0, \sigma_\epsilon^2)$ and $\sigma_\epsilon = 0.1$, where $Y$ is generated in four different scenarios using $\beta=1$ and a covariate $X$ that is a realisation of a unit-variance Mat\'ern field with practical range $2.5$ and a scenario-specific smoothness $\nu_X$. In all scenarios, we have $\alpha=2$ and $\kappa$ is chosen so the practical range of the random effect $U$ is $2.5$, with $\tau$ chosen so that marginal variance is $1$. The scenarios are 
(S1) $Y$ generated by \eqref{eq:forced-spde} with $\nu_X = 1$; 
(S2) $Y = X\beta + U$ with $\nu_X = 1/2$; 
(S3) $Y = X\beta + U$ with $\nu_X = 2$; and  
(S4) $Y = \beta \mathcal{S}_0 X + U$ where
$\mathcal{S}_0$ is a Gaussian kernel smoother and $\nu_X = 1$.
Thus, for (S1) the forced model is well-specified, for (S2), the covariate is rougher than $U$, for (S3) it is smoother than $U$, and in (S4) we have a mismatched smoothing operator. The data are generated from the finite element discretisation of Section~\ref{sec:inference} on the same mesh used for estimation. The data-generating model is then exactly operator-matched at the discrete level (see Section~\ref{sec:fem-error}), so the study isolates the statistical behaviour of the estimators from discretisation error.

We compare seven different models for the data, with the goal of estimating $\beta$. In all models, the random effect has a fixed smoothness $\alpha=2$: The additive model \eqref{eq:additive-intro} (Additive-GLS); an additive model with $X$  pre-smoothed by a Gaussian kernel of fixed bandwidth equal to the practical range, matching the smoother used in Section~\ref{sec:app-temp-alt} (BW-empirical); Spatial+; RSR \citep{ReichHodgesZadnik2006}; the forced model with the operator fixed at the true $(\kappa,\tau)$, an oracle that isolates the cost of estimating the operator (Forced-fixed); the forced model with $(\kappa,\tau)$ and $\beta_{\mathrm{fc}}$ estimated (Forced), and the hybrid model \eqref{eq:hybrid-spde} with $(\kappa,\tau)$ and $(\beta_{\mathrm{pw}}, \beta_{\mathrm{fc}})$ estimated. 

The additive, BW-empirical, Spatial+, Forced and Hybrid fits are obtained through maximum likelihood using the \texttt{rSPDE} package. Forced-fixed is fitted through a profile GLS fit at the true operator, and RSR is a restricted spatial regression. We report $\hat\beta$ (for the hybrid $\hat\beta_{\mathrm{pw}}$ and $\hat\beta_{\mathrm{fc}}$ separately) with truths $(\beta_{\mathrm{pw}}, \beta_{\mathrm{fc}}) = (0, 1)$ under S1 and S4, and $(1, 0)$ under S2 and S3, bias, RMSE and empirical coverage of nominal $95\%$ intervals. The two hybrid coefficients are compared to these truths componentwise. Every other estimator reports a single coefficient, and is compared to $\beta = 1$, the size with which $X$ enters the data in all four scenarios. Standard errors are obtained from the numerical Hessian of the negative log-likelihood at the MLE, except for RSR which uses its model-based $\widehat\sigma_\epsilon^2 (X^\top X)^{-1}$.

The channel-separation indices of Theorem~\ref{thm:geometry} are
$\vartheta(X) = 0.45, 0.51, 0.35, 0.46$ for S1--S4, so the pointwise and forcing channels are well separated in every scenario and the hybrid is
well posed. Note that these indices are computed from the discrete operator and can differ from their continuum counterparts. In S2 the covariate has $\nu_X = \nicefrac12 < \nu$, so in the continuum $X \notin \HH$ and the two channels decouple asymptotically (Proposition~\ref{prop:rough}(ii)--(iii)), whereas on the mesh all functions are finite-dimensional and $\vartheta(X) = 0.51$.

\subsection{Results}\label{sec:sim-results}

\begin{figure}[t]
\centering
\includegraphics[width=\textwidth]{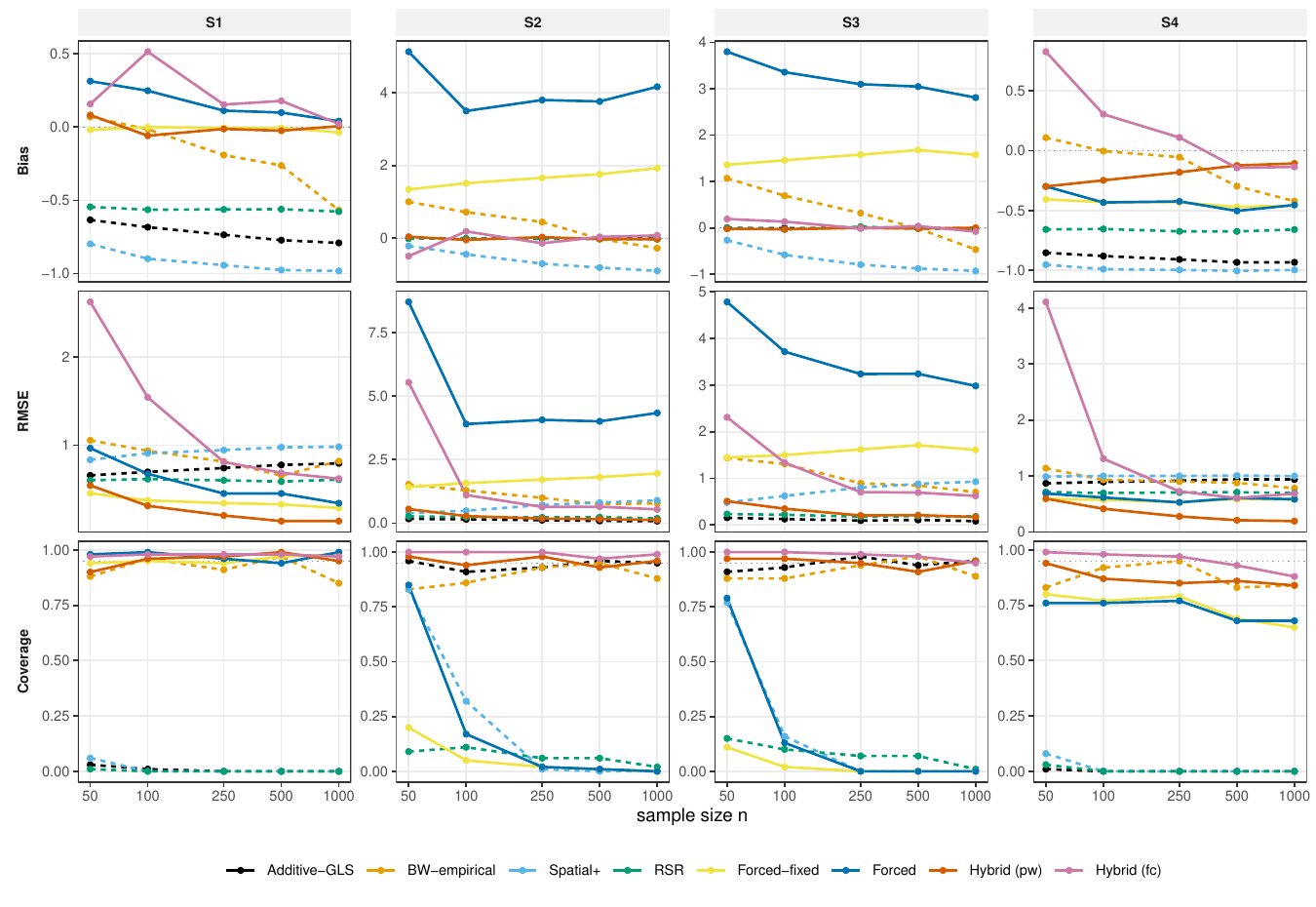}
\caption{Simulation results as functions of the sample size $n$. The rows show bias, RMSE, and empirical $95\%$ coverage of $\hat\beta$ and the columns show the four scenarios. Competitors (Additive-GLS, BW-empirical,
Spatial+, RSR) are dashed whereas the operator-matched Forced-fixed, Forced and the two hybrid coefficients are solid.}
\label{fig:sim-curves}
\end{figure}

Figure~\ref{fig:sim-curves} reports the results, which are obtained by running $100$ Monte Carlo replicates per (scenario, $n$) pair.  
As the spatial domain is fixed, increasing $n$ puts us in the infill asymptotics regime, where Proposition~\ref{prop:frontier} shows that the regression coefficients and the range $\kappa$ are not consistently estimable. Therefore, the right target is not vanishing RMSE but asymptotic unbiasedness and calibration, and the RMSE of a well-specified estimator decreases toward the positive limits $(\bm{G}^{-1})_{jj}^{\nicefrac1{2}}$ implied by Theorem~\ref{thm:infill}(ii) rather than to zero. This can for example be seen for the $\hat\beta_{\mathrm{pw}}$ curve in S1, which plateaus near $0.14$ from $n=500$.

Under (S1) the forced and hybrid models are well specified. The Forced-fixed and hybrid coefficients (known operator) are unbiased with coverage at the nominal $0.95$ for every $n$, as predicted by Theorem~\ref{thm:infill}(i). The Forced fit, which estimates $(\kappa,\tau)$, carries only a small additional bias that decays with $n$ (the  operator-estimation effect of Remark~\ref{rem:plugin}) and its empirical coverage remains close to the nominal $0.95$ throughout. 
The Additive-GLS, Spatial+ and RSR estimators instead are attenuated towards zero, with coverage collapsing to zero and no improvement as $n$ grows: entering the covariate pointwise when it acts on the response through the smoothing operator lets the random effect absorb the smooth signal. This is the spatial self-confounding analysed by \citet{BolinWallin2026}, which operator matching removes.

In (S2) and (S3) the additive model is correctly specified. The Additive-GLS and the hybrid pointwise coefficient are unbiased and calibrated at all $n$, and the hybrid assigns the effect to the pointwise channel ($\hat\beta_{\mathrm{pw}} \to 1$, $\hat\beta_{\mathrm{fc}} \to 0$). The forced-only estimators instead carry a positive bias that does not vanish. RSR keeps the OLS point estimate and is therefore unbiased here, but its reported uncertainty is  much smaller than the residual variation and the resulting intervals are anti-conservative, so the coverage collapses towards zero even where the point estimate is unbiased. This behaviour of RSR was documented by \citet{KhanCalder2022}.

Under operator misspecification (S4) the forced estimators remain bounded, with a stable bias near $-0.45$ and an RMSE flat in $n$. This is the bounded, sign-preserving attenuation guaranteed by Theorem~\ref{thm:misspec} and Corollary~\ref{cor:signs}.
Additive-GLS, Spatial+ and RSR again attenuate to zero coverage, while the hybrid isolates the forcing channel, with $\hat\beta_{\mathrm{pw}} \to 0$ and $\hat\beta_{\mathrm{fc}}$ close to the truth.

Across all scenarios the hybrid is the only method simultaneously unbiased and calibrated whenever the truth lies in its span. The BW-empirical smoother is applied to the observed covariate, so the effective amount of smoothing depends on the sampling density and its bias drifts with $n$ and attenuates at the larger sample sizes. This is expected for a fixed smoother and would be mitigated by estimating the smoothing range jointly, as in the two-scale model.

\subsection{Size of the likelihood-ratio tests}\label{sec:sim-lrt}
\begin{table}[t]
\centering
\caption{Empirical size of the two likelihood-ratio tests against the
$\chi^2_1$ reference, over $200$ replicates per cell (Monte Carlo
standard error $\approx 0.015$ at the $5\%$ level).}
\label{tab:lrt-size}
\begin{tabular}{llcccccc}
\toprule
 & & \multicolumn{3}{c}{nominal $5\%$} & \multicolumn{3}{c}{nominal $1\%$}\\
\cmidrule(lr){3-5}\cmidrule(lr){6-8}
Test & Truth & $n=64$ & $100$ & $250$ & $n=64$ & $100$ & $250$\\
\midrule
Matching ($\kappa_\mu=\kappa$) & S1 & $0.045$ & $0.055$ & $0.040$ & $0.005$ & $0.005$ & $0.010$\\
Hybrid vs.\ additive & S2 & $0.030$ & $0.075$ & $0.100$ & $0.005$ & $0.030$ & $0.035$\\
Hybrid vs.\ additive & S3 & $0.060$ & $0.055$ & $0.090$ & $0.020$ & $0.020$ & $0.020$\\
\bottomrule
\end{tabular}
\end{table}

The applications report likelihood-ratio statistics for two nested comparisons: the hybrid model against the additive baseline ($\beta_{\mathrm{fc}} = 0$), and the two-scale model against the matched hybrid ($\kappa_\mu = \kappa$). The fitted ranges are not small relative to the domain in the temperature application, so the expanding-domain calibration of Theorem~\ref{thm:incdom} cannot be taken for granted there. We therefore check the finite-sample size of both tests empirically, in the setting of this appendix. For the matching test, data are generated under the matched truth (S1) and the matched hybrid is tested against the two-scale model. For the hybrid-versus-additive test, data are generated under the additive truths (S2 and S3) and the additive model is tested against the hybrid. All fits estimate $(\kappa, \tau, \sigma_\epsilon^2)$ by maximum likelihood with $\alpha = 2$ fixed, and the sample sizes $n \in \{64, 100, 250\}$ cover those of Section~\ref{sec:app-temp-alt}. Table~\ref{tab:lrt-size} reports the empirical rejection rates at nominal levels $5\%$ and $1\%$. The operator-matching test is calibrated at all sample sizes considered.
The hybrid-versus-additive test is calibrated at $n = 64$, the sample size of the temperature application, and becomes mildly liberal as $n$
grows, reaching $0.10$ at the nominal $5\%$ level under S2 at $n = 250$. The effect is modest and immaterial for the strongly significant statistic reported in Section~\ref{sec:app-temp-alt}.

\section{Additional details for the applications}\label{app:details}

\subsection{Temperature application}

This section collects two supporting tables for the temperature application of Section~\ref{sec:app-temp-alt}: the fitted hyperparameters together with the effect of estimating the smoothness (Table~\ref{tab:tempalt-hyper}), and a comparison of the finite-element and covariance-based implementations of Section~\ref{sec:inference}
(Table~\ref{tab:temp-impl}).

\begin{table}[!ht]
\centering
\small
\begin{tabular}{llrrrrrr}
\toprule
 & & & & & & \multicolumn{2}{c}{CV RMSE}\\
\cmidrule(lr){7-8}
Regime & Model & $\hat\alpha$ & $\hat\kappa^{-1}$ (km) & range (km) & $\hat\sigma_\epsilon$ ($^\circ$C) & $\alpha{=}2$ & $\hat\alpha$\\
\midrule
\textit{Colorado} &
Additive & $1.89$ & 279 & 790 & 0.70 & 0.765 & 0.767\\
& Forced   & $1.37$ & 28 & 79 & 0.86 & 1.365 & 1.370\\
& Hybrid   & $2.01$ & 259 & 732 & 0.70 & 0.767 & 0.772\\
& Two-scale & -- & 324 & 916 & 0.70 & 0.763 & --\\
\midrule
\textit{Norway annual} & 
Additive & $2.41$ & 234 & 662 & 0.36 & 0.497 & 0.489\\
& Forced   & $1.88$ &  56 & 158 & 0.51 & 1.303 & 1.217\\
& Hybrid   & $2.91$ & 179 & 505 & 0.36 & 0.495 & 0.497\\
& Two-scale & -- & 79 & 225 & 0.34 & 0.502 & --\\
& Coast    & $3.03$ & 216 & 610 & 0.40 & 0.512 & 0.598\\
\midrule
\textit{Norway DJF} &
Additive & $1.92$ & 264 & 746 & 0.75 & 1.120 & 1.129\\
& Forced   & $1.29$ &  93 & 263 & 0.54 & 1.390 & 1.319\\
& Hybrid   & $1.83$ & 205 & 580 & 0.81 & 1.110 & 1.089\\
& Two-scale & -- & 169 & 478 & 0.79 & 1.118 & --\\
& Coast    & $1.32$ & 242 & 685 & 0.86 & 1.082 & 1.091\\
\bottomrule
\end{tabular}
\caption{Fitted hyperparameters and the effect of estimating the smoothness. Here $\hat\alpha$ is the jointly estimated smoothness, while $\hat\kappa^{-1}$, the practical range, and $\hat\sigma_\epsilon$ are from the $\alpha=2$ fit. The last two columns give the $10$-fold CV RMSE at $\alpha=2$ and at $\hat\alpha$. The two-scale model is fitted with $\alpha = 2$ only; its $\hat\kappa^{-1}$ is the range of the random field, and its separately estimated mean-channel ranges $\hat\kappa_\mu^{-1}$ are $20$, $585$, and $273$~km for the three regimes, respectively.}
\label{tab:tempalt-hyper}
\end{table}

\begin{table}[!ht]
\centering
\small
\begin{tabular}{llrrrrrr}
\toprule
Regime & Method & $\hat\kappa^{-1}$ (km) & $\hat\beta_{\mathrm{pw}}$ & $\hat\beta_{\mathrm{fc}}^{\star}\hat g_X$ & $\hat\sigma_\epsilon$ & CV RMSE & eval (ms)\\
\midrule
\textit{Colorado MAM} &
FEM  & 259 & $-6.69$ & $0.98$ & 0.70 & 0.767 & 27\\
($n=213$) & covariance & 387 & $-6.66$ & $0.59$ & 0.71 & 0.765 & 175\\
\midrule
\textit{Norway annual} & 
FEM  & 179 & $-5.65$ & $-3.07$ & 0.36 & 0.495 & 34\\
($n=64$) & covariance & 191 & $-5.73$ & $-2.27$ & 0.37 & 0.509 & 40\\
\midrule
\textit{Norway DJF}  &
FEM  & 205 & $-4.10$ & $-10.28$ & 0.81 & 1.110 & 36\\
($n=64$) & covariance & 211 & $-4.16$ & $-9.09$ & 0.81 & 1.145 & 40\\
\bottomrule
\end{tabular}
\caption{Parameters of the hybrid model with $\alpha=2$ fitted through the finite-element and covariance-based approaches of Section~\ref{sec:inference}. The  RMSE values are based on a $10$-fold CV, and eval is the median wall-clock of a single log-likelihood evaluation.}
\label{tab:temp-impl}
\end{table}

\subsection{Air quality application}
\label{app:no2}

The LUR models of Section~\ref{sec:app-no2} use 22 candidate covariates. The first 15 are constructed from the data in Figure~\ref{fig:no2-data} as
the pointwise traffic variable or buffered values with radii $0.1/0.3/0.5/1/2$ km, buffered industry values with radii $1/5/10/25$ km, the distance to the nearest industry, the number of industries within a $10$ km radius, and population with radii $1/5/10$ km. The remaining seven are six land-use covariates, namely urban and natural land-use fractions with radii $1/5/10$ km each, together with the distance to the coast.

All models are fitted with weakly informative priors. The regression coefficients have independent $N(0,10^{5})$ priors, with a flat intercept, and the observation precision has the default \texttt{R-INLA} log-Gamma prior. The field-only and LUR${}+{}$field models use the default \texttt{rSPDE} priors, i.e., independent Gaussian priors on the log standard deviation and log range with precision $0.1$. The operator-matched fields have Gaussian priors on the log range and log precision, centred at the values of a preliminary field-only fit and with precision $5$, and  each $\beta_{\mathrm{fc}}$ is $N(0,10^{3})$.

\section{Proofs}\label{app:proofs}

Throughout the proofs we use the following notation. 
For positive sequences, $a_j \asymp b_j$ means that constants $0 < c \le C < \infty$ exist with $c b_j \le a_j \le C b_j$ for all large $j$.
For two normed spaces, $V \hookrightarrow W$ denotes a continuous embedding, i.e., $V \subseteq W$ and the inclusion is bounded,
$\|v\|_W \le C\|v\|_V$ for all $v \in V$, so that convergence in $V$ implies convergence in $W$. We let $H^{s}(\mathcal{D})$ be the $L^2$-based Sobolev space of order $s$ (for integer $s$, the functions whose weak derivatives up to order $s$ are square-integrable, extended to non-integer $s$ by interpolation) and $C^{0,\varsigma}(\barD)$ the space of $\varsigma$-H\"older-continuous functions. 
The embedding $H^{s}(\mathcal{D}) \hookrightarrow C^{0,\varsigma}(\barD)$ holds for $0 < \varsigma \le s - \nicefrac{d}{2}$, $\varsigma < 1$.

We write $\omega_j$ for the eigenvalues of the negative Neumann Laplacian $-\Delta$ on $\mathcal{D}$, ordered increasingly and repeated by multiplicity, with orthonormal eigenfunctions $\{e_j\}$, so that $L e_j = \lambda_j e_j$ with $\lambda_j = \kappa^2 + \omega_j$ and $0 = \omega_1 \le \omega_2 \le \cdots \to \infty$, hence $\lambda_1 = \kappa^2 > 0$. 
The coefficients of $f \in L^2(\mathcal{D})$ are $f_j = (f, e_j)_{L^2}$.
The fractional power $L^{s}$ is defined by the spectral calculus, acting as multiplication by $\lambda_j^{s}$ on the $j$th eigencoordinate
$f_j = (f, e_j)_{L^2}$, with domain $D(L^{s}) = \{ f \in L^2(\mathcal{D}) : \sum_j \lambda_j^{2s} f_j^2 < \infty \}$. 
A bounded self-adjoint operator $T$ with eigenvalues $\mu_j$ has operator (spectral) norm $\|T\| = \sup_j |\mu_j|$. If moreover $T \succeq 0$ (so $\mu_j \ge 0$) it is trace class when $\operatorname{tr} T = \sum_j \mu_j < \infty$, and Hilbert--Schmidt, with Hilbert--Schmidt norm $\|T\|_{\mathrm{HS}} = (\sum_j \mu_j^2)^{1/2}$, when $\sum_j \mu_j^2 < \infty$. A centred Gaussian field takes values in $L^2(\mathcal{D})$ when its covariance operator is trace class. 
The Loewner order $S \preceq T$ means $T - S \succeq 0$, i.e., $(Sg, g)_{L^2} \le (Tg, g)_{L^2}$ for all $g$. Finally, we write $\spn$ for linear span and $V^{\perp}$ for orthogonal complement.

\subsection{Proofs for Section~\ref{sec:physical}}

Recall from Proposition~\ref{prop:rep} that $\HH = D(L^{\nicefrac{\alpha}{2}})$ is the Cameron--Martin space of $U$, equivalently the reproducing-kernel Hilbert space of its covariance kernel $\varrho$, with inner product $(\cdot,\cdot)_\mC$ and norm $\|\cdot\|_\mC$.

\begin{lemma}\label{lem:embed}
Under Assumption~\ref{ass:A} we have 
(a) $\omega_j \asymp j^{\nicefrac{2}{d}}$, and consequently $\lambda_j \asymp j^{\nicefrac{2}{d}}$; 
(b) $\sum_j \lambda_j^{-s} < \infty$ if and only if $s > \nicefrac{d}{2}$; 
and (c) $D(L^{\nicefrac{\alpha}{2}}) \hookrightarrow H^{\min(\alpha, 2)}(\mathcal{D}) \hookrightarrow C^{0,\varsigma}(\barD)$ for every $0 < \varsigma \le \min(\alpha, 2) - \nicefrac{d}{2}$ with $\varsigma < 1$. 
In particular, since $d \le 3$ and $\alpha > \nicefrac{d}{2}$, point evaluation is a bounded linear functional on $(\HH, \|\cdot\|_\mC)$ and there is a constant $C_E < \infty$ with $\|f\|_{C^{0,\varsigma}(\barD)} \le C_E \|f\|_\mC$ for all $f \in \HH$.
\end{lemma}

\begin{proof}
(a) is Weyl's law for the Neumann Laplacian, which holds under Assumption~\ref{ass:A} \citep[Remark~2]{xiong2022}. 
(b) follows from (a) by comparison with $\sum_j j^{-2s/d}$.
For (c), under Assumption~\ref{ass:A}, Proposition~B.2 in \citet{xiong2022} gives $D(L^{\sigma/2}) \hookrightarrow H^{\sigma}(\mathcal{D})$ for $0 \le \sigma \le 2$. Taking $\sigma = \alpha$ when $\alpha \le 2$, and $D(L^{\nicefrac{\alpha}{2}}) \subset D(L) \hookrightarrow H^2(\mathcal{D})$ when $\alpha > 2$ (domains of powers are nested decreasingly), yields the first inclusion. The second is the Sobolev embedding
$H^{s}(\mathcal{D}) \hookrightarrow
C^{0,\varsigma}(\barD)$ for $0 < \varsigma \le s - \nicefrac{d}{2}$, $\varsigma < 1$ (\citealp[Thm.~4.12]{AdamsFournier2003}; see also \citealp[Sec.~3.1]{CoxKirchner2020}), applied with $s = \min(\alpha, 2)$.
The norm bound defines $C_E$, and boundedness of point evaluation is immediate.
\end{proof}

\begin{proof}[Proof of Proposition~\ref{prop:timeavg}]

To simplify the notation write $f := \gamma_0 X \in L^2(\mathcal{D})$, so that the drift is $-A Y_t + f$. We let $\rho = D\kappa^2$ denote the relaxation rate of Section~\ref{sec:physical}, while $\lambda_j$ are the eigenvalues of $L$. Note that $a_j := D\lambda_j \ge D\kappa^2 = \rho$
for all $j$.
Since a cylindrical Wiener process has independent standard Brownian coordinates $w_j$ in any orthonormal basis, in the eigenbasis $\{e_j\}$
the dynamics decouple into 
$$
\mathrm{d}y_j = \bigl(-a_j  y_j + f_j\bigr)\mathrm{d}t + \mathrm{d}w_j, \qquad j \ge 1,
$$
whose stationary solution is the invariant measure of the associated OU semigroup \citep[Theorems~11.17 and~11.20]{DaPratoZabczyk2014},
taken on the space $H^{-s}(\mathcal{D}) := D(L^{-s/2})$ (with $\|g\|_{H^{-s}(\mathcal{D})}^2 = \sum_j \lambda_j^{-s} g_j^2$). The $y_j$ are independent stationary OU processes, each with stationary law $N\bigl(f_j/a_j,  (2a_j)^{-1}\bigr)$ and covariance $C_j(u) = (2a_j)^{-1} e^{-a_j |u|}$. The assembled field $Y_t := \sum_j y_j(t) e_j$ is then, for every $s > \nicefrac{d}{2} - 1$, a well-defined random element of 
$H^{-s}(\mathcal{D})$, since in stationarity, 
$$
\E\|Y_t\|_{H^{-s}(\mathcal{D})}^2 = \sum_j \lambda_j^{-s}\bigl[(2a_j)^{-1} + (f_j/a_j)^2\bigr] < \infty,
$$
by Lemma~\ref{lem:embed}(b).
Its coordinates are recovered as the duality pairings $y_j(t) = \langle Y_t, e_j\rangle$, which are well defined since $e_j \in D(L^{s/2})$. 

Each coordinate average $\bar y_j = T^{-1}\int_0^T y_j(t)\,\mathrm{d}t$ is Gaussian with $\E[\bar y_j] = f_j/a_j$, so the time average
$\bar Y_T := \sum_j \bar y_j e_j$ has mean $\sum_j (f_j/a_j) e_j = A^{-1}f$. This mean lies in $L^2(\mathcal{D})$ because $\sum_j (f_j/a_j)^2 \le \rho^{-2}\|f\|_{L^2}^2 < \infty$. Further, by direct calculation,
$$
\Var(\bar y_j) = \frac{1}{T^2}\int_0^T\!\!\int_0^T C_j(t - s)\,\mathrm{d}t\,\mathrm{d}s
= \frac{1}{a_j^2 T}\bigl(1 + r_j\bigr),
\qquad
r_j := -\frac{1 - e^{-a_j T}}{a_j T}.
$$
Hence, $0 \le \Var(\bar y_j) \le (a_j^2 T)^{-1}$, so $\sum_j \Var(\bar y_j) \le (D^2 T)^{-1}\sum_j \lambda_j^{-2} < \infty$ for $d \le 3$ by Lemma~\ref{lem:embed}(b). Therefore, $\bar Y_T = \sum_j \bar y_j e_j$ converges in $L^2(\Omega; L^2(\mathcal{D}))$ and is an $L^2(\mathcal{D})$-valued
Gaussian variable with mean $A^{-1}f$.

Finally, as $\bar y_j$ are independent across $j$, the covariance operator of $\sqrt{T}(\bar Y_T - A^{-1}f)$ is diagonal in the eigenbasis $\{e_j\}$ and has corresponding eigenvalues $\{a_j^{-2}(1 + r_j)\}$. It thus equals $A^{-2}(I + R_T)$ where $R_T$ is diagonal in $\{e_j\}$ with eigenvalues $\{r_j\}$. Therefore, $\|R_T\| = \sup_j |r_j| \le (\rho T)^{-1}$, proving the stated bound.
In particular, $\bar Y_T$ is Gaussian with mean $A^{-1}f = (\nicefrac{\gamma_0}{D})L^{-1}X$ and covariance operator $T^{-1}A^{-2}(I + R_T) = \tau_T^{-2}L^{-2}(I + R_T)$, where $\tau_T = D\sqrt{T}$. The forced model \eqref{eq:forced-spde} with parameters $(\alpha, \tau, \beta) = (2, \tau_T, \beta_T)$ has mean $(\nicefrac{\beta_T}{\tau_T})L^{-1}X$ and covariance operator $\tau_T^{-2}L^{-2}$, so matching the mean gives the identification $\nicefrac{\beta_T}{\tau_T} = \nicefrac{\gamma_0}{D}$, i.e., $\beta_T = \gamma_0\sqrt{T}$, with the covariance matched up to
the factor $I + R_T$. The standardised coefficient is $\beta_{\mathrm{fc}}^{\star} = \nicefrac{\beta_T}{(\tau_T\kappa^{2})}
= \nicefrac{\gamma_0}{(D\kappa^{2})} = \nicefrac{\gamma_0}{\rho}$, which does not depend on $T$. 
The leading term $A^{-2} = D^{-2}L^{-2}$ is trace class for $d \le 3$, since $\sum_j \lambda_j^{-2} < \infty$ by Lemma~\ref{lem:embed}(b) (with $s = 2 > \nicefrac{d}{2}$). By \citet[Example~3.8.13(iii)]{Bogachev1998}, a sequence of centred Gaussian measures on a separable Hilbert space with covariance operators $K_n$ converges weakly to the centred Gaussian measure with covariance operator $K$ if and only if $K_n^{1/2} \to K^{1/2}$ in the Hilbert--Schmidt norm. Here $\Sigma_T = A^{-2}(I + R_T)$ and $(A^{-2})^{1/2} = A^{-1}$ are diagonal in $\{e_j\}$, and as
$|\sqrt{1 + r_j} - 1| \le |r_j| \le (\rho T)^{-1}$, 
$$
\bigl\|\Sigma_T^{1/2} - A^{-1}\bigr\|_{\mathrm{HS}}^2
= \sum_j a_j^{-2}\bigl(\sqrt{1 + r_j} - 1\bigr)^2
\le (\rho T)^{-2}\operatorname{tr}\bigl(A^{-2}\bigr)
\longrightarrow 0,
$$
so $\sqrt{T}(\bar Y_T - A^{-1}f) \xrightarrow{d} N(0, A^{-2})$. The limit $N(0, D^{-2}L^{-2})$ is precisely the law of the Whittle--Mat\'ern
field with $\alpha = 2$ and $\tau = D$.
\end{proof}

\subsection{Proofs for Section~\ref{sec:properties}}

\begin{proof}[Proof of Proposition~\ref{prop:rep}]

Since $\beta_{\mathrm{fc}}X$ lies in $L^2(\mathcal{D})$, the right-hand side of \eqref{eq:hybrid-spde} has the form $g + \W$ with $g = \beta_{\mathrm{fc}}X \in L^2(\mathcal{D})$. Viewed as an equation for $\tau(Y - \beta_{\mathrm{pw}}X)$, it falls under \citet[Remark~2.4]{BolinKirchnerKovacs2020}, which yields the unique $L^2(\mathcal{D})$-valued solution $\tau(Y - \beta_{\mathrm{pw}}X) = L^{-\nicefrac{\alpha}{2}}\bigl(\beta_{\mathrm{fc}}X + \W\bigr)$.
With $\mathcal{S} = \tau^{-1}L^{-\nicefrac{\alpha}{2}}$ this is the additive representation \eqref{eq:additive}, $Y = \beta_{\mathrm{pw}}X + \beta_{\mathrm{fc}}\mathcal{S}X + U$, in which $U = \mathcal{S}\W = \tau^{-1}\sum_j \lambda_j^{-\nicefrac{\alpha}{2}}\xi_j e_j$ (with $\xi_j = \W(e_j)$ i.i.d.\ $N(0,1)$) is the Whittle--Mat\'ern field of \citet[Section 3]{xiong2022}, with covariance operator $\mC = \tau^{-2}L^{-\alpha}$, trace class for $\alpha > \nicefrac{d}{2}$ by Lemma~\ref{lem:embed}(b).
As all eigenvalues $\tau^{-2}\lambda_j^{-\alpha}$ are positive its unique positive square root is $\mC^{1/2} = \tau^{-1} L^{-\nicefrac{\alpha}{2}}$, with range $D(L^{\nicefrac{\alpha}{2}})$. So $\HH = \mC^{1/2}(L^2(\mathcal{D})) = D(L^{\nicefrac{\alpha}{2}})$ and
$$
\|f\|_\mC = \|\mC^{-1/2} f\|_{L^2}
= \tau \|L^{\nicefrac{\alpha}{2}} f\|_{L^2}
= \tau\Bigl(\sum_j \lambda_j^{\alpha} f_j^2\Bigr)^{\nicefrac1{2}}.
$$

The space $\HH$ is thus the Cameron--Martin space of $U$ \citep[\S~2.4]{Bogachev1998}, which coincides with the RKHS of its covariance kernel $\varrho$ \citep{vanderVaartVanZanten2008}, and in particular $\varrho(\cdot, s) \in \HH$ and $(\varrho(\cdot, s), f)_\mC = f(s)$ for every $f \in \HH$ and $s \in \barD$.
By Lemma~\ref{lem:embed}(c), $\HH = D(L^{\nicefrac{\alpha}{2}}) \hookrightarrow C^{0,\varsigma}(\barD)$ for every $0 < \varsigma \le \min(\alpha,2) - \nicefrac{d}{2}$ with $\varsigma < 1$, and $\|f\|_{C^{0,\varsigma}(\barD)} \le C_E\|f\|_\mC$ on $\HH$. Since $\varrho(\cdot, s) \in \HH$ for every $s$, the reproducing property and this bound give $\E|U(s) - U(t)|^2 = \|\varrho(\cdot, s) - \varrho(\cdot, t)\|_\mC^2 \le C_E^2 |s - t|^{2\varsigma}$. Hence, $\varrho$ is continuous on $\barD^2$ and, by Kolmogorov's continuity theorem, $U$ admits a modification with $\varsigma'$-H\"older continuous sample paths for every $\varsigma' < \varsigma$, hence for every $\varsigma' < \min\{\min(\alpha, 2) - \nicefrac{d}{2},  1\}$. We work with this modification throughout. 
In particular $\widetilde Y_i = Y(s_i) + \epsilon_i$ is well defined, with $Y$ continuous since $\mathcal{S}X\in\HH$ and also $X\in\HH$ in the hybrid case by assumption.
Finally, since $\mathcal{S} = \tau^{-1}L^{-\nicefrac{\alpha}{2}} = \mC^{1/2}$, it is the canonical isometric isomorphism from $L^2(\mathcal{D})$ onto
$(\HH, \|\cdot\|_\mC) = \mC^{1/2}(L^2(\mathcal{D}))$ and every $f \in \HH$ equals $\mathcal{S}(\tau L^{\nicefrac{\alpha}{2}} f)$.
\end{proof}

\begin{proof}[Proof of Theorem~\ref{thm:geometry}]
Since $X \in \HH$, $\|X\|_\mC^2 = \tau^2\sum_j\lambda_j^\alpha x_j^2 < \infty$. As $L^{\nicefrac{\alpha}{2}}\mathcal{S}X = \tau^{-1}X$, 
\begin{align*}
(X, \mathcal{S}X)_\mC
   &= \tau^2(L^{\nicefrac{\alpha}{2}}X, L^{\nicefrac{\alpha}{2}}\mathcal{S}X)_{L^2}
   = \tau^2(L^{\nicefrac{\alpha}{2}}X, \tau^{-1}X)_{L^2} 
   = \tau\sum_j \lambda_j^{\nicefrac{\alpha}{2}} x_j^2 ,
\end{align*}
where the last sum is finite by the Cauchy--Schwarz inequality.
Together with $\|\mathcal{S}X\|_\mC^2 = \|X\|_{L^2}^2 = \sum_j x_j^2$ from the isometry of Proposition~\ref{prop:rep}, this is \eqref{eq:gram-spectral}.
Let $\theta(X)$ be the Cameron--Martin angle between $X$ and $\mathcal{S}X$, defined by $\cos\theta(X) = (X,\mathcal{S}X)_\mC/(\|X\|_\mC \|\mathcal{S}X\|_\mC)$, so that $\vartheta(X) = \sin^2\theta(X) = 1 - \cos^2\theta(X)$. Dividing the numerator and denominator of $\cos^2\theta(X)$ by $\|X\|_{L^2}^4$ and writing $w_j = x_j^2/\|X\|_{L^2}^2$,
\begin{align*}
\cos^2\theta(X)
&= \frac{\bigl(\tau\sum_j\lambda_j^{\nicefrac{\alpha}{2}}x_j^2\bigr)^2}
       {\tau^2\sum_j\lambda_j^{\alpha}x_j^2 \cdot \sum_j x_j^2}
= \frac{\bigl(\sum_j w_j \lambda_j^{\nicefrac{\alpha}{2}}\bigr)^2}
       {\sum_j w_j \lambda_j^{\alpha}}
= \frac{(\E_w[\lambda^{\nicefrac{\alpha}{2}}])^2}{\E_w[\lambda^{\alpha}]},\\
\vartheta(X) &= 1 - \cos^2\theta(X)
= \frac{\E_w[\lambda^{\alpha}] - (\E_w[\lambda^{\nicefrac{\alpha}{2}}])^2}{\E_w[\lambda^{\alpha}]}
= \frac{\Var_w(\lambda^{\nicefrac{\alpha}{2}})}{\E_w[\lambda^{\alpha}]} \ge 0.
\end{align*}
Moreover, $\vartheta(X) < 1$, since $\E_w[\lambda^{\nicefrac{\alpha}{2}}] \ge \lambda_1^{\nicefrac{\alpha}{2}} = \kappa^{\alpha} > 0$ and $\E_w[\lambda^{\alpha}] < \infty$ for $X \in \HH$.

To prove (i), note that $\bm{G}$ is the Gram matrix of $X, \mathcal{S}X$ and hence positive semi-definite, and singular if and only if the vectors are linearly dependent. Now, 
$
\vartheta(X) = 0 \Leftrightarrow \Var_w(\lambda^{\nicefrac{\alpha}{2}}) = 0 \Leftrightarrow \lambda_j^{\nicefrac{\alpha}{2}} \text{is $w$-a.s.\ constant},
$
which holds if and only if there is $c > 0$ with $\lambda_j = c$ for every $j$ with $x_j \neq 0$. This means that $X \in \ker(L - c)$, i.e., that $X$ is an eigenfunction of $L$. 

To prove (ii), write $a = \|X\|_\mC^2$, $b = \|\mathcal{S}X\|_\mC^2$,
$c = (X, \mathcal{S}X)_\mC$, so that $\bm{G} = \bigl(\begin{smallmatrix} a & c \\ c & b\end{smallmatrix}\bigr)$ and $c^2 = ab\cos^2\theta$. Then
$\det \bm{G} = ab - c^2 = ab (1 - \cos^2\theta) = \|X\|_\mC^2  \|X\|_{L^2}^2  \vartheta(X)$, using $b = \|X\|_{L^2}^2$. For $\vartheta(X) > 0$ the inverse is 
$\bm{G}^{-1} = 
(\det \bm{G})^{-1}\bigl(
   \begin{smallmatrix} b & -c \\ -c & a\end{smallmatrix}\bigr)$,
so
$$
(\bm{G}^{-1})_{11} = \frac{b}{ab \vartheta(X)}
= \frac{1}{\|X\|_\mC^2 \vartheta(X)},
\qquad
(\bm{G}^{-1})_{22} = \frac{a}{ab \vartheta(X)}
= \frac{1}{\|X\|_{L^2}^2 \vartheta(X)}.
$$
\end{proof}

\begin{proof}[Proof of Proposition~\ref{prop:scalefree}]
As $\beta_{\mathrm{fc}}\mathcal{S}X
= \beta_{\mathrm{fc}} \tau^{-1}L^{-\nicefrac{\alpha}{2}}X
= \bigl(\nicefrac{\beta_{\mathrm{fc}}}{(\tau\kappa^{\alpha})}\bigr)
  \kappa^{\alpha}L^{-\nicefrac{\alpha}{2}}X
= \beta_{\mathrm{fc}}^{\star}\widetilde{\mathcal{S}}X$, the mean
functions, and hence the models, coincide. For fixed
$(\kappa, \tau, \alpha)$ the map
$(\beta_{\mathrm{pw}}, \beta_{\mathrm{fc}}) \mapsto
(\beta_{\mathrm{pw}}, \beta_{\mathrm{fc}}^{\star})$ is a linear bijection of $\R^2$, so the likelihood at corresponding parameter values is identical, and maximised likelihoods, likelihood-ratio statistics, and fitted predictive distributions are invariant, as is $\vartheta(X)$, which by \eqref{eq:separation} depends only on $(X, L, \alpha)$.
To prove (ii), note that in the eigenbasis, $(\widetilde{\mathcal{S}}X)_j = \kappa^{\alpha}\lambda_j^{-\nicefrac{\alpha}{2}}x_j = (\nicefrac{\kappa^2}{\lambda_j})^{\nicefrac{\alpha}{2}}x_j = (1 + \nicefrac{\omega_j}{\kappa^2})^{-\nicefrac{\alpha}{2}}x_j$ with $\omega_j = \lambda_j - \kappa^2$, so
$$
m_j = \beta_{\mathrm{pw}}x_j + \beta_{\mathrm{fc}}^{\star}(1 + \nicefrac{\omega_j}{\kappa^2})^{-\nicefrac{\alpha}{2}}x_j
= \beta(\omega_j)  x_j.
$$
The map $\omega \mapsto (1 + \nicefrac{\omega}{\kappa^2})^{-\nicefrac{\alpha}{2}}$ is strictly decreasing from $1$ at $\omega = 0$ to $0$ as
$\omega \to \infty$, so $\beta(\omega)$ is monotone between $\beta(0) = \beta_{\mathrm{pw}} + \beta_{\mathrm{fc}}^{\star}$ and
$\beta(\infty) = \beta_{\mathrm{pw}}$ (strictly when $\beta_{\mathrm{fc}}^{\star} \neq 0$). Solving
$(1 + \nicefrac{\omega}{\kappa^2})^{-\nicefrac{\alpha}{2}} = 1/2$ gives
$1 + \nicefrac{\omega}{\kappa^2} = 2^{\nicefrac{2}{\alpha}}$, i.e.\
$\omega_{\nicefrac1{2}} = \kappa^2(2^{\nicefrac{2}{\alpha}} - 1)$.

For (iii), by Parseval,
$$
\|\widetilde{\mathcal{S}}X\|_{L^2}^2
= \sum_j (1 + \nicefrac{\omega_j}{\kappa^2})^{-\alpha} x_j^2
= \|X\|_{L^2}^2 \E_w\bigl[(1 + \nicefrac{\omega}{\kappa^2})^{-\alpha}\bigr]
= \|X\|_{L^2}^2 g_X^2 .
$$
Each weight $(1 + \nicefrac{\omega_j}{\kappa^2})^{-\alpha}$ lies in $(0, 1]$, so $g_X \in (0, 1]$, with $g_X = 1$ if and only if 
$(1 + \nicefrac{\omega_j}{\kappa^2})^{-\alpha} = 1$ for every $j$ with $x_j \neq 0$, i.e.\ iff all spectral mass of $X$ sits at $\omega = 0$, whose eigenspace is the constants since $\mathcal{D}$ is connected. 
\end{proof}

\begin{lemma}
\label{lem:equiv}
Let $P$ and $P'$ be the laws of the continuous field $Y$ on $\barD$ under two parameter values, and let $P^{(n)}, P'^{(n)}$ be the corresponding joint laws of $(\widetilde Y_1, \ldots, \widetilde Y_n)$ under Assumption~\ref{ass:design}, with $P^{(\infty)}, P'^{(\infty)}$ the laws of the full observation sequences. If $P \sim P'$ (mutual absolute continuity), then $P^{(n)} \sim P'^{(n)}$ for every $n$, including $n = \infty$, and there exists no estimator sequence
$T_n = T_n(\widetilde Y_1, \ldots, \widetilde Y_n)$ that is consistent in probability for a parameter taking different values under $P$ and $P'$.
\end{lemma}

\begin{proof}
Let $P \sim P'$ on the path space $E = C(\barD)$ and let $P^\epsilon$ denote the law of the noise sequence $(\epsilon_i)_{i \ge 1}$ on $\R^{\mathbb{N}}$. We claim $P \otimes P^\epsilon \sim P' \otimes P^\epsilon$ with density $\frac{d(P \otimes P^\epsilon)}{d(P' \otimes P^\epsilon)}(y, e) = \frac{dP}{dP'}(y)$. Indeed, for measurable $A \subset E \times \R^{\mathbb{N}}$, the Fubini--Tonelli theorem and then the substitution $dP = \frac{dP}{dP'}\,dP'$ (the density exists by the hypothesis $P \sim P'$) 
give
$$
(P \otimes P^\epsilon)(A)
= \int \Bigl(\int \mathbf{1}_A(y, e)\, dP(y)\Bigr) dP^\epsilon(e)
= \int\!\!\int \mathbf{1}_A(y, e)\, \frac{dP}{dP'}(y)\, dP'(y)\, dP^\epsilon(e),
$$
which is the integral of $\mathbf{1}_A \cdot \frac{dP}{dP'}$ against $P' \otimes P^\epsilon$; the density is positive $P' \otimes P^\epsilon$-a.s.\ because $\frac{dP}{dP'} > 0$ $P'$-a.s., so the measures are mutually absolutely continuous.

The observation sequence $(\widetilde Y_i)_{i \ge 1} = T(Y, (\epsilon_i)_i)$, with $T(y, e) = (y(s_i) + e_i)_{i \ge 1}$,  is measurable from
$E \times \R^{\mathbb{N}}$ to $\R^{\mathbb{N}}$ as each coordinate is the sum of a continuous evaluation functional and a coordinate
projection. If $\mu \ll \nu$ and $T$ is measurable, then $T_\#\mu := \mu \circ T^{-1}$ and $T_\#\nu := \nu \circ T^{-1}$  satisfy $T_\#\mu \ll T_\#\nu$. Applying this in both directions to the equivalent product measures  shows that the infinite observation laws $P^{(\infty)} = T_\#(P \otimes P^\epsilon)$ and $P'^{(\infty)}$ are equivalent, and restricting to the first $n$ coordinates gives $P^{(n)} \sim P'^{(n)}$ for every $n$.
Finally, as $P^{(\infty)} \sim P'^{(\infty)}$ have the same null sets, no estimator sequence $T_n = T_n(\widetilde Y_1, \ldots, \widetilde Y_n)$ can converge almost surely, or in probability, to different limits under $P^{(\infty)}$ and $P'^{(\infty)}$, and none is consistent for a parameter taking different values under $P$ and $P'$.
\end{proof}

Thus, equivalence of Gaussian measures of the continuous fields transfers to the observation sequences. The results that follow use the following finite-sample notation. For a continuous function $f$ on $\barD$ we write $\bm{f}_n = (f(s_1), \ldots, f(s_n))^\top \in \R^n$ for its evaluation
vector at the first $n$ observation locations, and $\bm{\Sigma}_n = \bm{K}_n + \sigma_\epsilon^2 \bm{I}_n$ for the covariance matrix of the observations $\widetilde Y_i = Y(s_i) + \epsilon_i$, where $\bm{K}_n = [\varrho(s_i, s_j)]_{i,j \le n}$.

\begin{lemma}\label{lem:qform}
Let Assumptions~\ref{ass:A} and~\ref{ass:design} hold and let $f, g \in \HH$. Then the symmetric bilinear forms $Q_n(f, g) := \bm{f}_n^\top \bm{\Sigma}_n^{-1} \bm{g}_n$ satisfy $Q_n(f, g) \to (f, g)_\mC$ as $n \to \infty$.
\end{lemma}

\begin{proof}
By the reproducing property of Proposition~\ref{prop:rep}, writing $k_i = \varrho(\cdot, s_i) \in \HH$, we have $(k_i, h)_\mC = h(s_i)$ for all $h \in \HH$, and $(k_i, k_j)_\mC = \varrho(s_i, s_j) = (\bm{K}_n)_{ij}$. For every $f \in C(\barD)$ and every $n$ we have
\begin{equation}\label{eq:variational-app}
Q_n(f, f) = \bm{f}_n^\top \bm{\Sigma}_n^{-1} \bm{f}_n
= \min_{h \in \HH} J_n(h; f),
\end{equation}
for $J_n(h; f) := \|h\|_\mC^2 + \frac1{\sigma_\epsilon^{2}} \sum_{i=1}^n \bigl(f(s_i) - h(s_i)\bigr)^2$.
This is the identity of \citet[Theorem~1]{ZhdanovKalnishkan2013}, applied for each $n$ with kernel $\varrho$, whose RKHS is $(\HH, \|\cdot\|_\mC)$ by Proposition~\ref{prop:rep}, data $y_i = f(s_i)$, and regularisation parameter $a = \sigma_\epsilon^2$, after dividing both sides by $\sigma_\epsilon^2$. The minimum is attained by their Proposition~2.

For a fixed $h$, $J_{n+1}(h; f) \ge J_n(h; f)$ because the data-fit sum adds a non-negative term, and taking minima preserves the inequality, so
$Q_n(f,f)$ is non-decreasing in $n$ for the fixed nested design of Assumption~\ref{ass:design}. If $f \in \HH$, the choice $h = f$ in
\eqref{eq:variational-app} gives $Q_n(f, f) \le \|f\|_\mC^2$ for every $n$. Hence $Q_n(f,f) \uparrow Q_\infty(f) \le \|f\|_\mC^2$.
To find a lower bound, let $h_n \in \HH$ denote a minimiser in \eqref{eq:variational-app}. As $J_n(h_n; f) = Q_n(f,f) \le \|f\|_\mC^2$,  
\begin{equation}\label{eq:two-bounds-app}
\|h_n\|_\mC^2 \le \|f\|_\mC^2,
\qquad
\sum_{i \le n} \bigl(f - h_n\bigr)^2(s_i) \le \sigma_\epsilon^2  \|f\|_\mC^2.
\end{equation}
Take a subsequence along which $\|h_n\|_\mC^2 \to \liminf_n \|h_n\|_\mC^2$. Being bounded in $\HH$ by \eqref{eq:two-bounds-app}, it has a further subsequence $\{h_{n_k}\}$ converging weakly to some $h^* \in \HH$, and the reproducing property gives the pointwise convergence
$h_{n_k}(s) = (\varrho(\cdot, s), h_{n_k})_\mC \to (\varrho(\cdot, s), h^*)_\mC = h^*(s)$ for every $s \in \barD$. We
claim that $h^* = f$. If not, then since $f - h^*$ is continuous, there are $\varepsilon > 0$ and an open ball $B$ with
$B \cap \mathcal{D} \neq \emptyset$ such that $|f - h^*| \ge \varepsilon$ on $B \cap \barD$. By Assumption~\ref{ass:design}, $B$ contains design points $s_{i_1}, s_{i_2}, \ldots$, and for any fixed $N$, the pointwise convergence at the finitely many sites $s_{i_1}, \ldots, s_{i_N}$
gives $(f - h_{n_k})^2(s_{i_j}) \ge \nicefrac{\varepsilon^2}{4}$ for every $j \le N$ once $k$ is large enough that also $n_k \ge i_N$,
whence $\sum_{i \le n_k} (f - h_{n_k})^2(s_i) \ge N\tfrac{\varepsilon^2}{4}$, which contradicts \eqref{eq:two-bounds-app} once
$N > 4\sigma_\epsilon^{2}\|f\|_\mC^2/\varepsilon^2$. Hence $h^* = f$, and weak lower semicontinuity of the norm yields
$\liminf_n \|h_n\|_\mC^2 = \lim_k \|h_{n_k}\|_\mC^2 \ge \|h^*\|_\mC^2 = \|f\|_\mC^2$. Therefore
$Q_\infty(f) = \lim_n J_n(h_n; f) \ge \liminf_n \|h_n\|_\mC^2 \ge \|f\|_\mC^2$. Combining this with the fact that $Q_n(f, f) \le \|f\|_\mC^2$ yields $Q_n(f, f) \to \|f\|_\mC^2$.

Finally, $Q_n$ is a symmetric bilinear form, so for $f, g \in \HH$, using the polarisation identity twice yields the final conclusion:
\begin{align*}
Q_n(f, g) &= \tfrac14\bigl[Q_n(f{+}g, f{+}g) - Q_n(f{-}g, f{-}g)\bigr] 
\rightarrow
\tfrac14\bigl[\|f{+}g\|_\mC^2 - \|f{-}g\|_\mC^2\bigr] = (f, g)_\mC.
\end{align*}
\end{proof}

\begin{lemma}\label{lem:qform-rough}
Let Assumptions~\ref{ass:A} and~\ref{ass:design} hold and let $f \in C(\barD)$ with $f \notin \HH$. Then (i) $Q_n(f, f) \to \infty$ as $n \to \infty$; and (ii) if moreover $f \in L^2(\mathcal{D}) \setminus \{0\}$ and
$g := \mathcal{S}f$, then
$$
Q_n(g,g) - \frac{Q_n(f, g)^2}{Q_n(f,f)} \longrightarrow \|f\|_{L^2}^2,
\qquad
Q_n(f,f) - \frac{Q_n(f, g)^2}{Q_n(g,g)} \longrightarrow \infty.
$$
\end{lemma}

\begin{proof}
(i) The identity \eqref{eq:variational-app} holds for every $f \in C(\barD)$, and $Q_n(f,f)$ is non-decreasing in $n$, so $Q_n(f,f) \uparrow Q_\infty(f) \in [0, \infty]$. Suppose $Q_\infty(f) < \infty$ and let $h_n \in \HH$ denote the minimiser of $J_n(\cdot\,; f)$. As in \eqref{eq:two-bounds-app},
$\|h_n\|_\mC^2 \le Q_\infty(f)$ and $\sum_{i \le n}(f - h_n)^2(s_i) \le \sigma_\epsilon^2  Q_\infty(f)$ for every $n$. The sequence $(h_n)$ is then bounded in $\HH$ and has a weakly convergent subsequence $h_{n_k} \rightharpoonup h^* \in \HH$, and the argument in the proof of Lemma~\ref{lem:qform} (with $\sigma_\epsilon^2  Q_\infty(f)$ in place of $\sigma_\epsilon^2 \|f\|_\mC^2$ in the bound on the data-fit term)
shows that $h^* = f$. Thus $f = h^* \in \HH$, contradicting $f \notin \HH$. Hence $Q_\infty(f) = \infty$.

(ii) For $a \in \R$ let $F_n(a) := Q_n(g - a f,  g - a f) \ge 0$, a quadratic polynomial in $a$, and let $m_n := \min_a F_n(a) = Q_n(g,g) - Q_n(f,g)^2/Q_n(f,f)$ denote the Schur complement in the first limit of (ii), with minimiser $a_n = Q_n(f,g)/Q_n(f,f)$ (this is well defined for all large $n$ as $f$ is continuous and non-zero, the space-filling design eventually contains points with $f(s_i) \neq 0$, whence $Q_n(f,f) > 0$). For each fixed $a$ the function $g - af$ is continuous, so $Q_n(g-af, g-af)$ is non-decreasing in $n$ by \eqref{eq:variational-app}. Taking minima
preserves monotonicity, so $m_n \uparrow m_\infty$ with $m_\infty \le \lim_n F_n(0) = \lim_n Q_n(g,g) = \|g\|_\mC^2
= \|f\|_{L^2}^2$ by Lemma~\ref{lem:qform} and Proposition~\ref{prop:rep}. In particular $m_n \le \|f\|_{L^2}^2$ for every $n$. We now show $m_\infty \ge \|f\|_{L^2}^2$. Since $Q_n$ is a positive semidefinite bilinear form, the Cauchy--Schwarz inequality gives
$$
|a_n| = \frac{|Q_n(f,g)|}{Q_n(f,f)}
\le \Bigl(\frac{Q_n(g,g)}{Q_n(f,f)}\Bigr)^{\nicefrac1{2}}
\longrightarrow 0,
$$
as $Q_n(g,g) \to \|f\|_{L^2}^2 < \infty$ while $Q_n(f,f) \to \infty$ by part~(i). For every fixed $n_0$ and all $n \ge n_0$, monotonicity at fixed $a$ gives $m_n = F_n(a_n) \ge F_{n_0}(a_n)$, and letting $n \to \infty$, continuity of the quadratic $F_{n_0}$ together with $a_n \to 0$ yields $m_\infty \ge F_{n_0}(0) = Q_{n_0}(g,g)$. Letting $n_0 \to \infty$ gives $m_\infty \ge \|f\|_{L^2}^2$, proving the first limit.

For the second limit, write $\rho_n := Q_n(f,g)/\{Q_n(f,f)  Q_n(g,g)\}^{1/2} \in [-1, 1]$, so that $m_n = Q_n(g,g)(1 - \rho_n^2)$ and
$$
Q_n(f,f) - \frac{Q_n(f,g)^2}{Q_n(g,g)} = Q_n(f,f) (1 - \rho_n^2)
= m_n \frac{Q_n(f,f)}{Q_n(g,g)}
\longrightarrow \infty,
$$
since $m_n \to \|f\|_{L^2}^2 > 0$, $Q_n(g,g) \to \|f\|_{L^2}^2$, and $Q_n(f,f) \to \infty$ by part~(i).
\end{proof}

\begin{proof}[Proof of Proposition~\ref{prop:frontier}]

The two field laws are $N(m, \mC)$ and $N(m', \mC')$ on $L^2(\mathcal{D})$, with covariance operators $\mC = \tau^{-2}L^{-\alpha}$ and $\mC' = \tau'^{-2}L'^{-\alpha'}$ and means $m = \beta_{\mathrm{pw}}X + \beta_{\mathrm{fc}}\mathcal{S}X$ and $m' = \beta'_{\mathrm{pw}}X + \beta'_{\mathrm{fc}}\mathcal{S}'X$, where $L = \kappa^2 - \Delta$, $L' = \kappa'^2 - \Delta$, $\mathcal{S} = \tau^{-1}L^{-\nicefrac{\alpha}{2}}$, and $\mathcal{S}' = \tau'^{-1}L'^{-\nicefrac{\alpha'}{2}}$.
These are the shifted-Laplacian Whittle--Mat\'ern measures of \citet[Cor.~3.3]{BolinKirchner2023}, with fractional orders
$\nicefrac{\alpha}{2}, \nicefrac{\alpha'}{2} \in (\nicefrac{d}{4}, \infty)$ since $\alpha, \alpha' > \nicefrac{d}{2}$.
That corollary is stated for the Dirichlet Laplacian, but its proof uses only that $L$ and $L'$ share an eigenbasis together with the Weyl asymptotics $\omega_j \asymp j^{\nicefrac{2}{d}}$ of Lemma~\ref{lem:embed}(a), which both hold here so the corollary applies verbatim. Since $d \le 3$ (Assumption~\ref{ass:A}), its part~II gives $N(m, \mC) \sim N(m', \mC')$ if and only if $\alpha = \alpha'$, $\tau = \tau'$, and $m - m' \in \HH_0$, with $\HH_0 := \mC^{1/2}(L^2(\mathcal{D}))$, the laws being mutually singular otherwise.
The mean condition is automatic when  $\alpha = \alpha'$, $\tau = \tau'$ as $\HH_0 = D(L^{\nicefrac{\alpha}{2}}) = D(L'^{\nicefrac{\alpha}{2}})$ and by Proposition~\ref{prop:rep} both $\mathcal{S}X$ and $\mathcal{S}'X$ lie in $D(L^{\nicefrac{\alpha}{2}})$, as does $X$ whenever a pointwise channel is present (Assumption~\ref{ass:A}). The laws are therefore equivalent  when $(\tau, \alpha) = (\tau', \alpha')$ and singular when $\tau \neq \tau'$ or $\alpha \neq \alpha'$. 
\end{proof}

\begin{proof}[Proof of Theorem~\ref{thm:infill}]

Under the stated assumptions, $\widetilde{\bm{Y}}_n = \bm{D}_n \bbeta_0 + \bm{\eta}_n$ with
$\bm{\eta}_n = \bm{u}_n + \bm{\epsilon}_n \sim N(0, \bm{\Sigma}_n)$ and $\bm{\Sigma}_n = \bm{K}_n + \sigma_\epsilon^2 \bm{I}$. With
$(\kappa, \tau, \alpha, \sigma_\epsilon^2)$ known, twice the negative log-likelihood is, up to an additive constant,
$(\widetilde{\bm{Y}}_n - \bm{D}_n\bbeta)^\top \bm{\Sigma}_n^{-1}(\widetilde{\bm{Y}}_n - \bm{D}_n\bbeta)$. This is a quadratic function of $\bbeta$ whose gradient vanishes if and only if $\bm{M}_n \bbeta = \bm{D}_n^\top \bm{\Sigma}_n^{-1} \widetilde{\bm{Y}}_n$. Thus, whenever
$\bm{M}_n$ is non-singular the unique maximiser is $\hat\bbeta_n = \bm{M}_n^{-1} \bm{D}_n^\top \bm{\Sigma}_n^{-1} \widetilde{\bm{Y}}_n$.
Therefore $\hat\bbeta_n - \bbeta_0 = \bm{M}_n^{-1} \bm{D}_n^\top \bm{\Sigma}_n^{-1} \bm{\eta}_n$ is a linear transformation of the Gaussian vector $\bm{\eta}_n$, hence Gaussian with mean zero and covariance
\begin{equation}\label{eq:sandwich}
\bm{M}_n^{-1} \bm{D}_n^\top \bm{\Sigma}_n^{-1}\bm{\Sigma}_n\bm{\Sigma}_n^{-1} \bm{D}_n \bm{M}_n^{-1}
= \bm{M}_n^{-1} \bm{M}_n \bm{M}_n^{-1} = \bm{M}_n^{-1},
\end{equation}
proving (i). 
Further, by definition, $(\bm{M}_n)_{k\ell} = Q_n(f_k, f_\ell)$ with $(f_1, f_2) = (X, \mathcal{S}X)$. Both functions lie in $\HH$. Lemma~\ref{lem:qform} therefore gives $\bm{M}_n \to \bm{G}$ entrywise. By Theorem~\ref{thm:geometry}(i), $\vartheta(X) > 0$ implies $\bm{G} \succ 0$ and since eigenvalues are continuous functions of the entries, $\lambda_{\min}(\bm{M}_n) \to \lambda_{\min}(\bm{G}) > 0$, so $\bm{M}_n$ is
non-singular for all $n$ large enough, and $\bm{M}_n^{-1} \to \bm{G}^{-1}$ by continuity of matrix inversion at non-singular arguments. For any
$\bm{t} \in \R^2$ the characteristic function of $\hat\bbeta_n - \bbeta_0$ is $\exp(-\tfrac12 \bm{t}^\top \bm{M}_n^{-1} \bm{t}) \to \exp(-\tfrac12 \bm{t}^\top \bm{G}^{-1} \bm{t})$, so $\hat\bbeta_n - \bbeta_0 \xrightarrow{d} N(0, \bm{G}^{-1})$, proving (ii).

For the forced model $\bm{D}_n$ has a single column $(\mathcal{S}X)_n$, and $\mathcal{S}X \in \HH$ for every $X \in L^2(\mathcal{D})$ by Proposition~\ref{prop:rep}. The scalar $\bm{M}_n = Q_n(\mathcal{S}X, \mathcal{S}X) \to \|\mathcal{S}X\|_\mC^2 = \|X\|_{L^2}^2 > 0$
by Lemma~\ref{lem:qform} and Proposition~\ref{prop:rep}, and the steps above apply verbatim.
\end{proof}

\begin{proof}[Justification of Remark~\ref{rem:multivariate}]
With $p$ covariates the mean of \eqref{eq:additive} is a linear combination of the $2p$ functions $(X_1, \ldots, X_p, \mathcal{S}X_1, \ldots, \mathcal{S}X_p)$, which lie in $\HH$ by Assumption~\ref{ass:A} and Proposition~\ref{prop:rep}, so $\bm{D}_n$ becomes the $n \times 2p$ design built from these functions and $\bm{G}_p$ their Cameron--Martin Gram matrix. The argument for Theorem~\ref{thm:infill} is dimension-free. Specifically, the GLS computation \eqref{eq:sandwich} gives exact normality $\hat\bbeta_n \sim N(\bbeta_0, \bm{M}_n^{-1})$ for every $n$ at which $\bm{M}_n$ is non-singular, and Lemma~\ref{lem:qform} applies entrywise to $(\bm{M}_n)_{k\ell} = Q_n(f_k, f_\ell)$ and gives $\bm{M}_n \to \bm{G}_p$; and if $\bm{G}_p \succ 0$ then $\lambda_{\min}(\bm{M}_n) \to \lambda_{\min}(\bm{G}_p) > 0$, so $\hat\bbeta_n \xrightarrow{d} N(\bbeta_0, \bm{G}_p^{-1})$ exactly as before. Since $\bm{G}_p$ is the Gram matrix of the $2p$ functions in $(\HH, (\cdot,\cdot)_\mC)$, it is non-singular if and only if these functions are linearly independent, which is the joint identifiability of the two channels across covariates. Finally, $\vartheta_k = (\bm{R}_p^{-1})_{kk}^{-1}$ is the standard variance-inflation identity: writing $\bm{\Delta} = \operatorname{diag}(\bm{G}_p)$ and $\bm{R}_p = \bm{\Delta}^{-1/2}\bm{G}_p\bm{\Delta}^{-1/2}$, we have $(\bm{G}_p^{-1})_{kk} = (\bm{R}_p^{-1})_{kk} (\bm{G}_p)_{kk}^{-1}$, so the $k$th asymptotic variance exceeds the single-channel value $(\bm{G}_p)_{kk}^{-1}$ by  the factor $(\bm{R}_p^{-1})_{kk} \ge 1$.
\end{proof}

\begin{proof}[Proof of Proposition~\ref{prop:rough}]
(i) Since $X \notin \HH = D(L^{\nicefrac{\alpha}{2}})$, we have $\E_w(\lambda^{\alpha}) = \|X\|_{L^2}^{-2}\sum_j \lambda_j^{\alpha} x_j^2
= \infty$, while $\E_w(\lambda^{\nicefrac{\alpha}{2}}) < \infty$ by the hypothesis of part~(i). Hence, by Theorem~\ref{thm:geometry},
$\vartheta(X) = 1 - [\E_w(\lambda^{\nicefrac{\alpha}{2}})]^2/\E_w(\lambda^{\alpha}) = 1$.

(ii) By Proposition~\ref{prop:rep}, $\mathcal{S}X \in \HH$. In the joint GLS fit of the two channels with known covariance, $\hat\bbeta_n$ is, as in the proof of Theorem~\ref{thm:infill}, Gaussian and unbiased for every $n$ with covariance $\bm{M}_n^{-1}$, where $\bm{M}_n$ is the $2 \times 2$ Gram matrix of $(X, \mathcal{S}X)$ under the form $Q_n$. By the Schur-complement identity,
$(\bm{M}_n^{-1})_{22} = \{Q_n(\mathcal{S}X, \mathcal{S}X) - Q_n(X, \mathcal{S}X)^2 / Q_n(X, X)\}^{-1}$, and Lemma~\ref{lem:qform-rough}(ii) applied with $f = X$, $g = \mathcal{S}X$ gives $(\bm{M}_n^{-1})_{22} \to \|X\|_{L^2}^{-2}$. Hence $\hat\beta_{\mathrm{fc}}$ is Gaussian and unbiased with variance converging to $\|X\|_{L^2}^{-2}$, which is the  limit law of Theorem~\ref{thm:infill} with the stated limiting variance.

(iii) 
$(\bm{M}_n^{-1})_{11} = \{Q_n(X, X) - Q_n(X, \mathcal{S}X)^2 / Q_n(\mathcal{S}X, \mathcal{S}X)\}^{-1} \to 0$ by the second limit of Lemma~\ref{lem:qform-rough}(ii). Thus $\hat\beta_{\mathrm{pw}}$ is Gaussian and unbiased with variance tending to zero, so $\hat\beta_{\mathrm{pw}} \to \beta_{\mathrm{pw},0}$ in $L^2$ and hence in probability. 
\end{proof}

\begin{proof}[Proof of Theorem~\ref{thm:misspec}]

We have $\widetilde{\bm{Y}}_n = \bm{m}_{0,n} + \bm{\eta}_n$ with $\bm{\eta}_n \sim N(0, \bm{\Sigma}_n)$ and $\bm{m}_{0,n} = (m_0(s_i))_{i \le n}$, and the GLS estimator is $\hat\bbeta_n = \bm{M}_n^{-1}\bm{D}_n^\top\bm{\Sigma}_n^{-1}\widetilde{\bm{Y}}_n$ (for $n$ large enough). Substituting $\widetilde{\bm{Y}}_n$ and writing $\bm{v}_n := \bm{D}_n^\top\bm{\Sigma}_n^{-1}\bm{m}_{0,n}$, $\hat\bbeta_n = \bm{M}_n^{-1}\bm{v}_n + \bm{M}_n^{-1}\bm{D}_n^\top\bm{\Sigma}_n^{-1}\bm{\eta}_n$.
This is a deterministic vector plus a linear image of the Gaussian $\bm{\eta}_n$; hence $\hat\bbeta_n$ is Gaussian with mean $\bm{M}_n^{-1}\bm{v}_n$ and, by the same computation as in \eqref{eq:sandwich}, covariance $\bm{M}_n^{-1}$. As the columns of $\bm{D}_n$ are $\bm{X}_n$ and
$(\mathcal{S}X)_n$ and $Q_n(f, g) = \bm{f}_n^\top\bm{\Sigma}_n^{-1}\bm{g}_n$,
$$
\hat\bbeta_n \sim N\bigl( \bm{M}_n^{-1} \bm{v}_n, \bm{M}_n^{-1} \bigr),
\qquad
\bm{v}_n =
\begin{pmatrix}
   Q_n(X, m_0) \\ Q_n(\mathcal{S}X, m_0)
\end{pmatrix}.
$$
As $X, \mathcal{S}X, m_0 \in \HH$, Lemma~\ref{lem:qform} gives
$\bm{v}_n \to \bm{v} := \bigl((X, m_0)_\mC, (\mathcal{S}X, m_0)_\mC\bigr)^\top$ and $\bm{M}_n \to \bm{G}$, hence
$\bm{M}_n^{-1}\bm{v}_n \to \bm{G}^{-1}\bm{v} =: \bbeta_\infty$ and
$\bm{M}_n^{-1} \to \bm{G}^{-1}$. 
The characteristic-function argument in the proof of Theorem~\ref{thm:infill}, now with mean $\bm{M}_n^{-1}\bm{v}_n$, then yields
$\hat\bbeta_n \xrightarrow{d} N(\bbeta_\infty, \bm{G}^{-1})$.
Since $\bm{G}(a, b)^\top = \bm{v}$ are the normal equations of the orthogonal projection of $m_0$ onto $\spn\{X, \mathcal{S}X\}$ in
$(\cdot,\cdot)_\mC$, $\bbeta_\infty = \bm{G}^{-1}\bm{v}$ is the coefficient vector of that projection, and the unique minimiser of
$\|m_0 - aX - b\mathcal{S}X\|_\mC^2$ over $(a, b) \in \R^2$. If $m_0 = \beta_0 X$ (i.e.\ $\mathcal{S}_0 = \mathcal{I}$), then
$\bm{v} = \beta_0(\|X\|_\mC^2,  (\mathcal{S}X, X)_\mC)^\top$ is $\beta_0$ times the first column of $\bm{G}$ in \eqref{eq:gram}, so
$\bbeta_\infty = (\beta_0, 0)^\top$, and if $m_0 = \beta_0\mathcal{S}X$ the second column gives $\bbeta_\infty = (0, \beta_0)^\top$.
Finally, with the single regressor $\mathcal{S}X$, the previous steps collapse to
$$
\hat\beta_{\mathrm{fc},n} \xrightarrow{d}
N\Bigl(\frac{(\mathcal{S}X,  m_0)_\mC}{\|\mathcal{S}X\|_\mC^2},
   \frac{1}{\|\mathcal{S}X\|_\mC^2}\Bigr)
=
N\Bigl(\frac{(\mathcal{S}X,  m_0)_\mC}{\|X\|_{L^2}^2},
   \frac{1}{\|X\|_{L^2}^2}\Bigr),
$$
and the mean is finite by the Cauchy--Schwarz inequality.
\end{proof}

\begin{proof}[Proof of Corollary~\ref{cor:confounding}]
We have $X \in \HH$ by Assumption~\ref{ass:A}, so $m_0 = \beta_0 X + Z \in \HH$ and Theorem~\ref{thm:misspec} applies. Let 
$\bm{v}(\cdot) = ((X, \cdot)_\mC   (\mathcal{S}X, \cdot)_\mC)^\top$. Then by linearity of the inner product and becuase 
as $\bm{v}(X)$ is the first column of $\bm{G}$ in \eqref{eq:gram},
$\bbeta_\infty = \beta_0  \bm{G}^{-1}\bm{v}(X) + \bm{G}^{-1}\bm{v}(Z)
=  \beta_0  (1,0)^\top + \bm{G}^{-1}\bm{v}(Z)$.
\end{proof}

\begin{proof}[Proof of Corollary~\ref{cor:signs}]

Write $m_0 = \beta_0 \sigma_0(L)X$, so that $(m_0)_j = \beta_0 \sigma_0(\lambda_j) x_j$. Then,
$$
(\mathcal{S}X, m_0)_\mC
= \tau^2 \sum_j \lambda_j^{\alpha}
   \bigl(\tau^{-1}\lambda_j^{-\nicefrac{\alpha}{2}} x_j\bigr)
   \bigl(\beta_0 \sigma_0(\lambda_j)  x_j\bigr)
= \beta_0  \tau \sum_j \lambda_j^{\nicefrac{\alpha}{2}}
   \sigma_0(\lambda_j)  x_j^2 .
$$
The sum converges absolutely since $M = \sum_j \lambda_j^{\alpha}\sigma_0(\lambda_j)^2 x_j^2 < \infty$ as $m_0 \in \HH$, and by Cauchy--Schwarz inequality $\sum_j \lambda_j^{\nicefrac{\alpha}{2}}\sigma_0(\lambda_j)  x_j^2 \le M^{\nicefrac12}(\sum_j x_j^2)^{\nicefrac12} < \infty$.
By assumption, at least one $j$ satisfies $\sigma_0(\lambda_{j}) > 0$, $x_{j} \neq 0$, and as every term is non-negative because $\sigma_0 \ge 0$, the sum is strictly positive. Dividing by $\|X\|_{L^2}^2$,
$$
\beta_{\mathrm{fc},\infty} = 
\beta_0 \tau \frac{\sum_j \lambda_j^{\nicefrac{\alpha}{2}}\sigma_0(\lambda_j)x_j^2}{\sum_j x_j^2}
= \beta_0 \tau \E_w\bigl[\lambda^{\nicefrac{\alpha}{2}}\sigma_0(\lambda)\bigr],
$$
which is finite, non-zero, and has the sign of $\beta_0$.
\end{proof}

\begin{lemma}\label{lem:incdom-bounds}
Let $\psi = (\beta_{\mathrm{pw}}, \beta_{\mathrm{fc}}, \kappa, \tau, \sigma_\epsilon^2)$ be the parameter vector, $\alpha > \nicefrac{d}{2}$ fixed and known, and $\Psi_0$ a compact neighbourhood of the true value $\psi_0$ in the interior of the parameter set $\Psi$ of Theorem~\ref{thm:incdom}. For the increasing-domain model, the mean $\bm{\mu}_n(\psi)$ and covariance $\bm{\Sigma}_n(\psi)$ are three times continuously
differentiable in $\psi$ and under Assumption~\ref{ass:B}, uniformly in $n$ and $\psi \in \Psi_0$:
(i) $\lambda_{\min}(\bm{\Sigma}_n(\psi)) \ge \sigma_\epsilon^2$ and $\lambda_{\max}(\bm{\Sigma}_n(\psi)) \le C < \infty$;
(ii) the spectral norms of all parameter derivatives of $\bm{\Sigma}_n(\psi)$ up to third order are bounded by a constant;
(iii) all parameter derivatives of the mean up to third order are bounded entrywise: $\|\partial^m \bm{\mu}_n(\psi)\|_\infty \le C$ for $|m| \le 3$.
\end{lemma}

\begin{proof}

(i) The lower bound $\lambda_{\min}(\bm{\Sigma}_n) \ge \sigma_\epsilon^2$ is immediate from $\bm{\Sigma}_n = \bm{K}_n + \sigma_\epsilon^2 \bm{I}$ with $\bm{K}_n \succeq 0$. For the upper bound, since $\bm{K}_n$ is symmetric, Gershgorin's theorem gives
\begin{equation}\label{eq:lambda_bound}
   \lambda_{\max}(\bm{K}_n)
   \le \max_{i \le n} \sum_{j \le n} |\varrho(s_i - s_j)|
   \le \sum_{z \in \delta\mathbb{Z}^d} |\varrho(z)| ,
\end{equation}
where $\varrho$ now denotes the stationary Mat\'ern covariance
function
\begin{equation}\label{eq:matern_cov}
\varrho(r) = \sigma^2  \frac{2^{1-\nu}}{\Gamma(\nu)}
(\kappa\|r\|)^{\nu} K_{\nu}(\kappa\|r\|) =  \frac{1}{\tau^{2}(2\pi)^d}
     \int_{\R^d} \bigl(\kappa^2 + \|\xi\|^2\bigr)^{-\alpha}
     e^{\mathrm{i}\xi^\top r}\, d\xi,
\end{equation}
with $\nu = \alpha - \nicefrac{d}{2}$ and $\sigma^2 = \Gamma(\nu) / \bigl(\Gamma(\nu + \nicefrac{d}{2})(4\pi)^{\nicefrac{d}{2}}\kappa^{2\nu}\tau^2\bigr)$. Writing $z = \kappa\|r\|$ and $g(z) := z^\nu K_\nu(z)$, we have $\varrho(r) = \sigma^2 \tfrac{2^{1-\nu}}{\Gamma(\nu)} g(z)$, and the
map $z \mapsto g(z) e^{z/2}$ is continuous on $(0,\infty)$ with the limit $\Gamma(\nu)2^{\nu-1} < \infty$ as $z \to 0$, since
$K_\nu(z) \sim \tfrac{\Gamma(\nu)}{2}(z/2)^{-\nu}$. By the asymptotics $K_\nu(z) \sim \sqrt{\pi/(2z)} e^{-z}$ as $z \to \infty$
\citep[Eq.~9.7.2]{AbramowitzStegun1964} it behaves like $\sqrt{\pi/2} z^{\nu-1/2}e^{-z/2} \to 0$, so it is bounded by some
$M(\nu) < \infty$ and therefore $g(z) \le M(\nu) e^{-z/2}$. On  $\Psi_0$, $\sigma^2 2^{1-\nu}/\Gamma(\nu)$ and $M(\nu)$ are bounded and
$\kappa \ge \kappa_{\min} > 0$, whence $|\varrho(r)| \le C e^{-c\|r\|}$ uniformly on $\Psi_0$ with $c = \kappa_{\min}/2$. As the lattice sum
$\sum_{z \in \delta\mathbb{Z}^d} e^{-c\|z\|}$ is finite, this proves (i).

We next establish the differentiability, beginning with the covariance. The entries of $\bm{\Sigma}_n$ do not depend on $(\beta_{\mathrm{pw}}, \beta_{\mathrm{fc}})$, and the $\sigma_\epsilon^2$-derivative of $\bm{\Sigma}_n$ is $\bm{I}$, with all higher $\sigma_\epsilon^2$-derivatives zero. For the $\tau$-derivatives, $\partial_\tau \varrho = -2\tau^{-1}\varrho$ and all higher $\tau$-derivatives are also multiples of $\varrho$. By  the spectral representation \eqref{eq:matern_cov},
$$
\partial_\kappa \varrho(r) = \frac{-2\alpha\kappa}{\tau^{2}(2\pi)^d}
   \int_{\R^d} \bigl(\kappa^2 + \|\xi\|^2\bigr)^{-\alpha-1} e^{\mathrm{i}\xi^\top r}\, d\xi ,
$$
which is proportional to a Mat\'ern covariance of order $\alpha + 1$ and thus satisfies the same uniform exponential bound $|\partial_\kappa\varrho(r)| \le Ce^{-c\|r\|}$. Iterating, every mixed parameter derivative of $\varrho$ up to third order exists, is continuous in $\psi$ (the differentiated integrands are continuous in $\psi$ and dominated), and is a finite linear combination (with coefficients bounded on $\Psi_0$) of
Mat\'ern-type kernels of orders $\alpha, \ldots, \alpha+3$, all exponentially decaying uniformly on $\Psi_0$.

Turning to the mean, by \eqref{eq:additive},
$\bm{\mu}_n(\psi) = \beta_{\mathrm{pw}} [X(s_i)]_{i \le n} + \beta_{\mathrm{fc}} [(\mathcal{S}_\psi X)(s_i)]_{i \le n}$, with
$\mathcal{S}_\psi = \tau^{-1}L^{-\nicefrac{\alpha}{2}}$, is linear in $(\beta_{\mathrm{pw}}, \beta_{\mathrm{fc}})$ and depends on $(\kappa, \tau)$ only through $\mathcal{S}_\psi X$, which is the convolution $g_\psi * X$ with the kernel $g_\psi$ whose Fourier transform is the symbol
$s_\psi(\xi) = \tau^{-1}(\kappa^2 + \|\xi\|^2)^{-\nicefrac{\alpha}{2}}$.
Differentiating in $\kappa$ and $\tau$ shows that every $\partial^m s_\psi$ with $|m| \le 3$ is a linear combination
$\partial^m s_\psi(\xi) = \sum_{j=0}^{3} c_j(\kappa,\tau)\bigl(\kappa^2 + \|\xi\|^2\bigr)^{-\nicefrac{\alpha}{2} - j}$, whose coefficients $c_j$ are continuous in $(\kappa,\tau)$ and hence bounded on $\Psi_0$. Each term is the symbol of a nonnegative kernel: by the subordination identity
$$
\bigl(\kappa^2 + \|\xi\|^2\bigr)^{-\beta}
= \frac{1}{\Gamma(\beta)}\int_0^\infty t^{\beta-1} e^{-t\kappa^2} e^{-t\|\xi\|^2}\, dt,
\qquad \beta = \tfrac{\alpha}{2} + j > 0,
$$
its kernel $g_{\beta}(r) = \Gamma(\beta)^{-1}\int_0^\infty t^{\beta-1} e^{-t\kappa^2}(4\pi t)^{-\nicefrac{d}{2}} e^{-\|r\|^2/(4t)}\, dt$ is a positive mixture of Gaussian densities, so $g_\beta \ge 0$ and, by Tonelli's theorem, $\|g_\beta\|_{L^1} =  \kappa^{-2\beta}$, the symbol evaluated at the origin. As $g_\beta$ is pointwise decreasing in $\kappa$, the kernels $\partial^m g_\psi$  are dominated on $\Psi_0$ by the fixed integrable function $C\sum_{j\le3} g_{\nicefrac{\alpha}{2}+j}(\cdot\,; \kappa_{\min})$, so the differentiations may be taken under the subordination and
convolution integrals. Thus, $\partial^m(\mathcal{S}_\psi X) = (\partial^m g_\psi) * X$ exists and is continuous in $\psi$ for $|m| \le 3$. Together with the linearity in $(\beta_{\mathrm{pw}}, \beta_{\mathrm{fc}})$, this proves that $\bm{\mu}_n$ and $\bm{\Sigma}_n$ are three times continuously differentiable.

Now (ii) follows by applying \eqref{eq:lambda_bound} to the derivatives: every parameter derivative of $\varrho$ up to third order decays exponentially uniformly on $\Psi_0$, and the $\sigma_\epsilon^2$-derivative of $\bm{\Sigma}_n$ is $\bm{I}$, so the spectral norms of all parameter derivatives of $\bm{\Sigma}_n(\psi)$ up to third order are bounded uniformly in $n$ and on $\Psi_0$.
For (iii), every parameter derivative of the mean up to third order is $[X(s_i)]_i$ or $[((\partial^m\mathcal{S}_\psi) X)(s_i)]_i$
carrying at most one factor $\beta_{\mathrm{fc}}$. As $|\beta_{\mathrm{fc}}| \le C$ on $\Psi_0$ and $\|X\|_\infty \le C$ by Assumption~\ref{ass:B}, we have on $\Psi_0$ 
\begin{align*}
\|(\partial^m\mathcal{S}_\psi) X\|_\infty
   &= \|(\partial^m g_\psi) * X\|_\infty
   \le \|\partial^m g_\psi\|_{L^1(\R^d)} \|X\|_\infty \\
    & \le C \sum_{j=0}^{3} |c_j(\kappa,\tau)|\kappa^{-\alpha-2j}
   \le C
\end{align*}
by Young's inequality using $\kappa \ge \kappa_{\min} > 0$ and boundedness of the coefficients.
\end{proof}

The maximum likelihood estimator of Theorem~\ref{thm:incdom} is studied through the Gaussian log-likelihood and its derivatives. For
$\widetilde{\bm{Y}}_n \sim N(\bm{\mu}_n(\psi), \bm{\Sigma}_n(\psi))$, the log-likelihood is
$$
\ell_n(\psi) = -\tfrac{n}{2}\log(2\pi) - \tfrac12 \log\det\bm{\Sigma}_n(\psi) - \tfrac12\bigl(\widetilde{\bm{Y}}_n - \bm{\mu}_n(\psi)\bigr)^\top \bm{\Sigma}_n(\psi)^{-1} \bigl(\widetilde{\bm{Y}}_n - \bm{\mu}_n(\psi)\bigr),
$$
with score and Fisher information given by standard Gaussian identities \citep[see, e.g.,][]{MardiaMarshall1984}:
\begin{align}
\partial_k \ell_n
   &= (\partial_k\bm{\mu}_n)^\top \bm{\Sigma}_n^{-1} \bm{r}_n
   + \tfrac12 \bm{r}_n^\top \bm{\Sigma}_n^{-1}(\partial_k\bm{\Sigma}_n)
     \bm{\Sigma}_n^{-1} \bm{r}_n
   - \tfrac12 \operatorname{tr}\bigl(\bm{\Sigma}_n^{-1}
     \partial_k\bm{\Sigma}_n\bigr),\label{eq:score}\\
(\bm{I}_n)_{k\ell}
   &= (\partial_k \bm{\mu}_n)^\top \bm{\Sigma}_n^{-1} (\partial_\ell \bm{\mu}_n) + \tfrac12 \operatorname{tr}\bigl(
       \bm{\Sigma}_n^{-1} (\partial_k \bm{\Sigma}_n)
       \bm{\Sigma}_n^{-1} (\partial_\ell \bm{\Sigma}_n) \bigr). \label{eq:fisher-app}
\end{align}
Here $\bm{r}_n := \widetilde{\bm{Y}}_n - \bm{\mu}_n(\psi)$ and $\bm{\mu}_n(\psi) = \beta_{\mathrm{pw}} \bm{X}_n + \beta_{\mathrm{fc}}(\mathcal{S}_\psi X)_n$ depends on $(\beta_{\mathrm{pw}}, \beta_{\mathrm{fc}}, \kappa, \tau)$ and $\bm{\Sigma}_n(\psi)$ on $(\kappa, \tau, \sigma_\epsilon^2)$, so all blocks of \eqref{eq:fisher-app} are non-zero except the $(\bbeta, \sigma_\epsilon^2)$ block. The two summands give the split $\bm{I}_n = \bm{M}_n + \bm{C}_n$ into the mean-information Gram matrix $\bm{M}_n$ and the covariance-information $\bm{C}_n$. We now identify the limit of $n^{-1}\bm{I}_n$ and establish its positive-definiteness.
To define the limit, let $c(z) = \varrho(\delta z) + \sigma_\epsilon^2\mathbf{1}\{z=0\}$, $z\in\mathbb{Z}^d$, be the covariance function of the observations on the lattice $\delta\mathbb{Z}^d$, where $\varrho$ is the stationary Mat\'ern covariance, with spectral density (symbol) $f_\psi(\xi) := \sum_{z\in\mathbb{Z}^d} c(z) e^{-\mathrm{i}\delta z^\top\xi}$, $\xi \in [-\nicefrac{\pi}{\delta}, \nicefrac{\pi}{\delta})^d$ and corresponding covariance operator  $\bm{T}$. Thus, $\bm{T}$ is the convolution by $c$ on $\ell^2(\mathbb{Z}^d)$, $(\bm{T}v)_z = \sum_{z'} c(z-z') v_{z'}$, which is self-adjoint and bounded, with $\|\bm{T}\| \le \sum_z |c(z)| < \infty$ by Schur's test and since $\varrho$ decays exponentially. We first show that $\bm{T}^{-1}$ is the convolution by an exponentially decaying precision kernel $w$ whose symbol is $f_\psi^{-1}$.

\begin{lemma}\label{lem:wdecay}
Let Assumption~\ref{ass:B} hold. Then $f_\psi \ge \sigma_\epsilon^2$, $\bm{T}$ is invertible, and $\bm{T}^{-1}$ is the convolution by a kernel $w$, i.e., $(\bm{T}^{-1})_{zz'} = w(z - z')$ and $\sum_{z'\in\mathbb{Z}^d} c(z-z') w(z') = \mathbf{1}\{z=0\}$, which satisfies $\sum_{z\in\mathbb{Z}^d} w(z) e^{-\mathrm{i}\delta z^\top\xi} = f_\psi(\xi)^{-1}$ and $|w(z)| \le C  e^{-c_2\|z\|}$ for constants $C, c_2 > 0$, uniformly on $\Psi_0$. Moreover,
for every finite $\Lambda \subset \mathbb{Z}^d$ the principal submatrix $\bm{T}_\Lambda := (c(z - z'))_{z, z' \in \Lambda}$ satisfies $\lambda_{\min}(\bm{T}_\Lambda) \ge \sigma_\epsilon^2$ and $|(\bm{T}_\Lambda^{-1})_{zz'}| \le \tfrac{2}{\sigma_\epsilon^2} e^{-c_2\|z - z'\|}$, uniformly in $\Lambda$ and on $\Psi_0$.
\end{lemma}

\begin{proof}
By the proof of Lemma~\ref{lem:incdom-bounds}(i), $|c(z)| \le C_0 e^{-c_1\|z\|}$ uniformly on $\Psi_0$. For every finitely supported $v$,
$$
(\bm{T}v, v) = \sigma_\epsilon^2\|v\|^2 + \sum_{z, z'} \bar{v}_z  v_{z'}
\varrho\bigl(\delta(z - z')\bigr) \ge \sigma_\epsilon^2\|v\|^2,
$$
since $\varrho$ is positive definite, so by density $\bm{T} \succeq \sigma_\epsilon^2 I$ and $\bm{T}$ is invertible with $\|\bm{T}^{-1}\| \le \sigma_\epsilon^{-2}$. Under the unitary Fourier map $v \mapsto \sum_z v_z e^{-\mathrm{i}\delta z^\top\xi}$ from $\ell^2(\mathbb{Z}^d)$ onto
$L^2([-\nicefrac{\pi}{\delta}, \nicefrac{\pi}{\delta})^d)$, the operator $\bm{T}$ becomes multiplication by $f_\psi$. Since $f_\psi$
is continuous and its range is the spectrum of $\bm{T}$, this gives $f_\psi \ge \sigma_\epsilon^2$, and $\bm{T}^{-1}$ becomes
multiplication by $f_\psi^{-1}$, i.e., $\bm{T}^{-1}$ is the convolution by the kernel $w$ with symbol $\sum_z w(z) e^{-\mathrm{i}\delta z^\top\xi} = f_\psi(\xi)^{-1}$. 
Applying $\bm{T}\bm{T}^{-1} = I$ to the coordinate vector at the origin gives $\sum_{z'\in\mathbb{Z}^d} c(z-z') w(z') = \mathbf{1}\{z = 0\}$.

For the finite-section bound, let $\Lambda \subset \mathbb{Z}^d$. Then the corresponding submatrix $\bm{T}_\Lambda$ of $\bm{T}$ satisfies 
$\lambda_{\min}(\bm{T}_\Lambda) \ge \sigma_\epsilon^2$ as it is a compression of $\bm{T}$. For $\theta \in \R^d$ let
$\bm{D}_\theta := \operatorname{diag}(e^{\theta^\top z})$, so that $\bm{D}_\theta\bm{T}_\Lambda\bm{D}_\theta^{-1} - \bm{T}_\Lambda$
has entries $c(z - z')\bigl(e^{\theta^\top(z-z')} - 1\bigr)$. By the bound $|e^{x} - 1| \le |x|e^{|x|}$ and Schur's test,
$\|\bm{D}_\theta\bm{T}_\Lambda\bm{D}_\theta^{-1} - \bm{T}_\Lambda\| \le C_0\|\theta\| \sum_z \|z\| e^{-c_1\|z\|/2} =: C'\|\theta\|$ for $\|\theta\| \le \nicefrac{c_1}{2}$. 
Thus, for $\|\theta\| \le c_2 := \min\{\nicefrac{c_1}{2}, \sigma_\epsilon^2/(2C')\}$ a Neumann series around $\bm{T}_\Lambda$
shows that $\bm{D}_\theta\bm{T}_\Lambda\bm{D}_\theta^{-1}$ is invertible with inverse norm at most $\nicefrac{2}{\sigma_\epsilon^2}$. Taking $\theta = c_2 (z - z')/\|z - z'\|$ in $(\bm{T}_\Lambda^{-1})_{zz'} = e^{-\theta^\top(z - z')} \bigl(\bm{D}_\theta\bm{T}_\Lambda\bm{D}_\theta^{-1}\bigr)^{-1}_{zz'}$ yields $|(\bm{T}_\Lambda^{-1})_{zz'}| \le \tfrac{2}{\sigma_\epsilon^2} e^{-c_2\|z - z'\|}$, uniformly in $\Lambda$ and on $\Psi_0$ (cf.\ \citealp{CombesThomas1973}, where this conjugation argument originates).
Finally, the inverse of a finite section is not the corresponding section of $\bm{T}^{-1}$, so we pass to the limit: for
$\Lambda_L := [-L, L]^d \cap \mathbb{Z}^d$, $\bm{T}_{\Lambda_L}^{-1} \to \bm{T}^{-1}$ entrywise as $L \to \infty$. Indeed, for fixed $z'$, the vectors
$u_L := \bm{T}_{\Lambda_L}^{-1}e_{z'}$ satisfy $\|u_L\| \le \sigma_\epsilon^{-2}$, and any weak subsequential limit $u$ satisfies, for every fixed $z$, $(\bm{T}u)_z = \langle u, \bm{T}e_z\rangle = \lim_L \langle u_L, \bm{T}e_z\rangle = \lim_L (\bm{T}u_L)_z = \mathbf{1}\{z = z'\}$, since
$(\bm{T}u_L)_z = (\bm{T}_{\Lambda_L}u_L)_z = \mathbf{1}\{z = z'\}$ once $\Lambda_L$ contains $z$ and the origin. Hence $u = \bm{T}^{-1}e_{z'}$, and $w(z - z') = (\bm{T}^{-1})_{zz'} = \lim_L (\bm{T}_{\Lambda_L}^{-1})_{zz'}$ inherits the bound $|w(z)| \le \tfrac{2}{\sigma_\epsilon^2}  e^{-c_2\|z\|}$.
\end{proof}

\begin{lemma}\label{lem:info-structure}
Define $\bm{I}(\psi) = \bm{M}(\psi) + \bm{C}(\psi)$ on $\Psi_0$ by
\begin{equation}
\label{eq:info-limit}
\bm{C}(\psi)_{k\ell} = \tfrac12\bigl(\bm{T}^{-1}\partial_k c,
     \bm{T}^{-1}\partial_\ell c\bigr)_{\ell^2(\mathbb{Z}^d)},
   \qquad
\bm{M}(\psi)_{k\ell} = \bigl\langle \partial_k\bm{\mu}, w * \partial_\ell\bm{\mu}\bigr\rangle,
\end{equation}
where $*$ is convolution over $\delta\mathbb{Z}^d$ and $\langle u, v\rangle := \lim_{n\rightarrow\infty} \langle u, v\rangle_n$, with $\langle u, v\rangle_n := n^{-1}\sum_{s_i\in\mathcal{D}_n} u(s_i) v(s_i)$, is the per-site average along the design. Then
$\bm{I}(\psi) = \bm{M}(\psi) + \bm{C}(\psi)$ is continuous and positive-definite on $\Psi_0$.
\end{lemma}

\begin{proof}
We first note that $\bm{C}$ is well defined and establish that the per-site averages defining $\bm{M}(\psi)$ exist. By Lemma~\ref{lem:wdecay}, $\bm{T}^{-1}\partial_k c = w * \partial_k c$, which lies in $\ell^1(\mathbb{Z}^d) \subset \ell^2(\mathbb{Z}^d)$ as $w$ and $\partial_k c$ (Lemma~\ref{lem:incdom-bounds}(ii)) decay exponentially; thus $\bm{C}(\psi)$ is one half the Gram matrix of the sequences $w * \partial_k c$ in $\ell^2(\mathbb{Z}^d)$.
For $\bm{M}$, each $\partial_k\bm{\mu}$ samples a continuum convolution $g_k * X$ with $g_k$ an exponentially decaying kernel (Lemma~\ref{lem:incdom-bounds}(iii)) whose transform $\widehat{\partial_k\bm{\mu}} := \hat g_k$ equals $1$, $s_\psi$, $\beta_{\mathrm{fc}} \partial_\kappa s_\psi$ and $\beta_{\mathrm{fc}} \partial_\tau s_\psi$ for $k = \beta_{\mathrm{pw}}, \beta_{\mathrm{fc}}, \kappa, \tau$, as a
function of $\xi \in \R^d$. Since $w$ is summable (Lemma~\ref{lem:wdecay}) and the lag averages $\langle \partial_k\bm{\mu}, \partial_\ell\bm{\mu}(\cdot - \delta z)\rangle_n$ are bounded by $\|\partial_k\bm{\mu}\|_\infty \|\partial_\ell\bm{\mu}\|_\infty$, uniformly in $n$, $z$ and $\Psi_0$
(Lemma~\ref{lem:incdom-bounds}(iii)), dominated convergence over the lags shows that, as soon as each fixed-lag average $\langle \partial_k\bm{\mu}, \partial_\ell\bm{\mu}(\cdot - \delta z)\rangle$ exists, $\langle \partial_k\bm{\mu},  w * \partial_\ell\bm{\mu}\rangle = \sum_z w(z) \langle \partial_k\bm{\mu}, \partial_\ell\bm{\mu}(\cdot - \delta z)\rangle$ exists as well. It therefore suffices to compute the fixed-lag averages. For $z \in \mathbb{Z}^d$,
$$
\langle \partial_k\bm{\mu}, \partial_\ell\bm{\mu}(\cdot - \delta z)\rangle_n
= \int\limits_{\R^d}\!\int\limits_{\R^d} g_k(u) g_\ell(v)
\frac1n \!\!\sum_{s_i\in\mathcal{D}_n} X(s_i - u) X(s_i - u - \delta z + u - v) du\, dv,
$$
where for fixed $(u,v)$ the inner average is a shifted design average $\hat r_n(u', -\delta z + u - v)$ of Assumption~\ref{ass:B}, with
$u' \equiv -u \pmod{\delta\mathbb{Z}^d}$, up to a boundary correction on a fraction $O((\|u\| + \|v\| + \delta\|z\|)/a_n)$ of sites. Further, it is bounded by $\|X\|_\infty^2$ and converges to $R_X(-\delta z + u - v)$, so dominated convergence over the integrable kernels gives $\langle \partial_k\bm{\mu}, \partial_\ell\bm{\mu}(\cdot - \delta z)\rangle
= \int\!\!\int g_k(u) g_\ell(v) R_X(-\delta z + u - v)\,du\,dv$.
Collecting the Fourier transforms, using
$\sum_z w(z)  e^{-\mathrm{i}\delta z^\top\xi} = \tilde f_\psi^{-1}(\xi)$, where $\tilde f_\psi^{-1}$ is the $(\nicefrac{2\pi}{\delta})$-periodic extension of $f_\psi^{-1}$ to $\R^d$, and $R_X(h) = \int e^{\mathrm{i}\xi^\top h}\, dW_X(\xi)$,  yields
$$
\bm{M}(\psi)_{k\ell} = \int_{\R^d}
\widehat{\partial_k\bm{\mu}} \overline{\widehat{\partial_\ell\bm{\mu}}}
\tilde f_\psi^{-1}\,dW_X.
$$

Now, both $\bm{C}$ and $\bm{M}$ are Gram matrices and therefore positive semidefinite, and the mean rows of $\bm{C}$ vanish, since
$\partial_k c = 0$ for $k \in \{\beta_{\mathrm{pw}}, \beta_{\mathrm{fc}}\}$. If $v$ is a null vector of the covariance block of $\bm{C}$,
the Gram structure gives $\sum_k v_k  (w * \partial_k c) = 0$, and convolving with $c$ gives, by Lemma~\ref{lem:wdecay},
$\sum_k v_k  \partial_k c = 0$. At lags $z \neq 0$, where $\partial_{\sigma_\epsilon^2} c = 0$ and $\partial_\tau c = -\tfrac{2}{\tau}\varrho(\delta z)$, this reads $v_\kappa  \partial_\kappa\varrho(\delta z) = \tfrac{2 v_\tau}{\tau}  \varrho(\delta z)$ for all $z \neq 0$. By \eqref{eq:matern_cov} and the expression for $\partial_\kappa\varrho$, $\partial_\kappa\varrho(r)/\varrho(r) = -r K_{\nu+1}(\kappa r)/K_\nu(\kappa r) \le -r$ for $r > 0$, since $K_{\nu+1} \ge K_\nu$ for $\nu > 0$ \citep[Eq.~9.6.24]{AbramowitzStegun1964}. This ratio is thus unbounded, and not constant on the lattice, and $v_\kappa = 0$. Then $\varrho > 0$ forces $v_\tau = 0$, and the equation at $z = 0$ gives $v_{\sigma_\epsilon^2} = 0$. Thus $\bm{C}(\psi)$ is positive semidefinite with null space exactly the mean coordinates.
For $\bm{M}(\psi)$, its $(\beta_{\mathrm{pw}}, \beta_{\mathrm{fc}})$ sub-block is the Gram matrix of the symbols $1$ and $s_\psi$ in $L^2(\tilde f_\psi^{-1}\,dW_X)$. Since $s_\psi$ is strictly decreasing in $\|\xi\|$ and $\tilde f_\psi^{-1} > 0$, this sub-block is singular only if $W_X$ is zero or concentrated on a single modulus $\|\xi\|$, excluded by Assumption~\ref{ass:B}, hence it is positive definite. 

As $\bm{M}(\psi)$ and $\bm{C}(\psi)$ are positive semidefinite, $\bm{v}^\top\bm{I}(\psi)\bm{v} = 0$ forces $\bm{v}^\top\bm{C}(\psi)\bm{v} = \bm{v}^\top\bm{M}(\psi)\bm{v} = 0$. The first pins the covariance coordinates of $\bm{v}$ to zero, and the surviving $(\beta_{\mathrm{pw}}, \beta_{\mathrm{fc}})$ part then lies in the kernel of the positive-definite $(1, s_\psi)$ Gram block, so $\bm{v} = \bm{0}$ and $\bm{I}(\psi) \succ 0$. Finally, by
Lemma~\ref{lem:incdom-bounds} the symbols $s_\psi$, $f_\psi^{-1}$ and their $\psi$-derivatives are bounded and equicontinuous in $\psi$
uniformly in the frequency variable (and $W_X$ is a finite measure), so \eqref{eq:info-limit} is Lipschitz in $\psi$. Thus $\bm{I}(\psi)$ is continuous and $\psi\mapsto\lambda_{\min}(\bm{I}(\psi))$ attains a positive minimum $c_0 > 0$ on $\Psi_0$.
\end{proof}

\begin{lemma}\label{lem:Sigmainv}
Under Assumption~\ref{ass:B}, constants $c_3, c_4 > 0$ exist such that, uniformly in $n$ and $\psi \in \Psi_0$,
\begin{align}
|(\bm{\Sigma}_n^{-1})_{ij}| &\le (2/\sigma_\epsilon^2)
e^{-c_3\|s_i - s_j\|}\label{eq:Sinvbound}\\
(\bm{\Sigma}_n^{-1})_{ij} &= w\bigl((s_i - s_j)/\delta\bigr)
   + O\bigl(e^{-c_4(\|s_i - s_j\|
   + \mathrm{dist}(\{s_i, s_j\},\partial\mathcal{D}_n))}\bigr) \label{eq:finite-section}
\end{align}
Further, summed over all pairs $i, j \le n$, the error term in \eqref{eq:finite-section} contributes
$O\bigl(\sum_i e^{-c_4 \mathrm{dist}(s_i,\partial\mathcal{D}_n)}\bigr)
= O(a_n^{d-1}) = o(n)$.
\end{lemma}

\begin{proof}
The kernels satisfy $|c(z)| \le Ce^{-c_1\|z\|}$ (by the proof of Lemma~\ref{lem:incdom-bounds}(i)) and $|w(z)| \le Ce^{-c_2\|z\|}$ (Lemma~\ref{lem:wdecay}), uniformly on $\Psi_0$. Indexed by the sites $s_i = \delta z_i$, $\bm{\Sigma}_n$ is the principal submatrix of $\bm{T}$ on the design sites, $\bm{\Sigma}_n = \bm{P}_n\bm{T}\bm{P}_n$, where $\bm{P}_n$ is the coordinate projection onto $\mathcal{D}_n$, and $(\bm{T}^{-1})_{ij} = w((s_i - s_j)/\delta)$. Hence \eqref{eq:Sinvbound} is the finite-section bound of Lemma~\ref{lem:wdecay} written in physical distance, $\|s_i - s_j\| = \delta\|z_i - z_j\|$, with $c_3 = \nicefrac{c_2}{\delta}$. For \eqref{eq:finite-section}, with $\bm{Q}_n = \bm{I} - \bm{P}_n$ we have that $\bm{\Sigma}_n^{-1} - \bm{P}_n\bm{T}^{-1}\bm{P}_n = \bm{\Sigma}_n^{-1} \bm{P}_n\bm{T}\bm{Q}_n\bm{T}^{-1}\bm{P}_n$. Because $(\bm{P}_n\bm{T}^{-1}\bm{P}_n)_{ij} = w\bigl((s_i - s_j)/\delta\bigr)$, we just have to bound the right-hand side to obtain \eqref{eq:finite-section}. The $(i,j)$ entry of the right-hand side is
$$
\sum_{k: s_k \in \mathcal{D}_n}\sum_{l: s_l \notin \mathcal{D}_n}
   (\bm{\Sigma}_n^{-1})_{ik} T_{kl} (\bm{T}^{-1})_{lj},
$$
a sum over chains $s_i \to s_k \to s_l \to s_j$ whose middle site $s_l$ lies outside $\mathcal{D}_n$, because of the factor $\bm{Q}_n$. Each factor decays exponentially in the length of its hop (by \eqref{eq:Sinvbound} for the first, and by the decay of $c$ and of $w$ for the other two) so each term is bounded by $C''e^{-cL}$, where $L := \|s_i - s_k\| + \|s_k - s_l\| + \|s_l - s_j\|$ is the total length of the chain and $c > 0$ is the smallest of the three rates. The triangle inequality gives $L \ge \|s_i - s_j\|$.
Moreover, since the chain starts and ends in $\mathcal{D}_n$ but visits $s_l \notin \mathcal{D}_n$, its outbound and return legs satisfy $\|s_i - s_k\| + \|s_k - s_l\| \ge \|s_i - s_l\| \ge \mathrm{dist}(s_i, \partial\mathcal{D}_n)$ and $\|s_l - s_j\| \ge \mathrm{dist}(s_j, \partial\mathcal{D}_n)$, so also $L \ge \mathrm{dist}(s_i, \partial\mathcal{D}_n) + \mathrm{dist}(s_j, \partial\mathcal{D}_n)$. Averaging the two lower bounds gives 
$L \ge L_{\min} := \tfrac12\bigl(\|s_i - s_j\| + \mathrm{dist}(\{s_i, s_j\}, \partial\mathcal{D}_n)\bigr)$.
Finally, split each rate as $\varepsilon + (c - \varepsilon)$ with $0 < \varepsilon < c$: since $\sup_{i,j,n}\sum_{k,l} e^{-\varepsilon L} < \infty$ (a double convolution of summable kernels), the double sum is bounded by a constant multiple of $e^{-(c-\varepsilon)L_{\min}}$. This proves \eqref{eq:finite-section}, with the constant $c_4$ there given by $(c - \varepsilon)/2$.
The final statement follows as a boundary layer of thickness $t$ contains $O(t  a_n^{d-1}/\delta^d)$ sites, so $\sum_i e^{-c_4 \mathrm{dist}(s_i,\partial\mathcal{D}_n)} = O(a_n^{d-1})$, while $n \asymp (a_n/\delta)^d$.
\end{proof}

\begin{lemma}\label{lem:info-conv}
Under Assumption~\ref{ass:B}, $n^{-1}\bm{I}_n(\psi) \to \bm{I}(\psi)$ locally uniformly on a neighbourhood of $\psi_0$.
\end{lemma}

\begin{proof}
The trace terms of \eqref{eq:fisher-app} depend only on $(\kappa, \tau, \sigma_\epsilon^2)$, and $\partial_k\bm{\Sigma}_n$ is the finite section of the lattice convolution by $\partial_k c$, which decays exponentially by Lemma~\ref{lem:incdom-bounds}(ii). 
Write $\bm{\Sigma}_n^{-1} = \bm{W}_n + \bm{E}_n$, where $\bm{W}_n := \bm{P}_n\bm{T}^{-1}\bm{P}_n$ is the finite section of the convolution
by $w$ ($\bm{P}_n$ as in the proof of Lemma~\ref{lem:Sigmainv}) and, by \eqref{eq:finite-section} and the final statement of Lemma~\ref{lem:Sigmainv}, the error matrix satisfies $\sum_{i,j \le n}|(\bm{E}_n)_{ij}| = o(n)$. Expanding the trace multilinearly in the two factors $\bm{\Sigma}_n^{-1}$, every term containing $\bm{E}_n$ is bounded by a constant multiple of $\sum_{i,j}|(\bm{E}_n)_{ij}|$, since the remaining factors have row and column sums bounded uniformly in $n$. Therefore,
$$
\operatorname{tr}\bigl(\bm{\Sigma}_n^{-1}(\partial_k\bm{\Sigma}_n)
\bm{\Sigma}_n^{-1}(\partial_\ell\bm{\Sigma}_n)\bigr)
= \operatorname{tr}\bigl(\bm{W}_n(\partial_k\bm{\Sigma}_n)
\bm{W}_n(\partial_\ell\bm{\Sigma}_n)\bigr) + o(n).
$$
All four factors on the right are finite sections of lattice convolutions with exponentially decaying kernels.  Sections do not multiply exactly, but the discrepancy is of boundary order by the same argument as in the proof of \eqref{eq:finite-section}. Specifically, by the identity $(\bm{P}_n\bm{A}\bm{P}_n)(\bm{P}_n\bm{B}\bm{P}_n) = \bm{P}_n\bm{A}\bm{B}\bm{P}_n - \bm{P}_n\bm{A}\bm{Q}_n\bm{B}\bm{P}_n$, each replacement of a product of sections by the section of the product produces a correction whose entries are sums over chains passing through a site outside $\mathcal{D}_n$.  Its diagonal entries are $O(e^{-c \mathrm{dist}(s_i, \partial\mathcal{D}_n)})$, and their sum over the $n$ diagonal sites is $O(a_n^{d-1}) = o(n)$. Applying this three times, once at each junction of the four factors in the trace, $\operatorname{tr}\bigl(\bm{W}_n(\partial_k\bm{\Sigma}_n)
\bm{W}_n(\partial_\ell\bm{\Sigma}_n)\bigr) = \operatorname{tr}\bigl(\bm{P}_n \bm{T}^{-1}(\partial_k\bm{T})\bm{T}^{-1}(\partial_\ell\bm{T}) \bm{P}_n\bigr) + o(n)$, where $\partial_k\bm{T}$ denotes the convolution by $\partial_k c$ on $\ell^2(\mathbb{Z}^d)$. By stationarity, each of the $n$ diagonal entries of $\bm{T}^{-1}(\partial_k\bm{T})\bm{T}^{-1}(\partial_\ell\bm{T})$ equals its entry at $z = z' = 0$, which is
$(w * \partial_k c * w * \partial_\ell c)(0) = (\bm{T}^{-1}\partial_k c, \bm{T}^{-1}\partial_\ell c)_{\ell^2(\mathbb{Z}^d)}$, as $c$ and its parameter derivatives, and hence also $w$, are even kernels.
Dividing by $2n$ and letting $n \to \infty$ therefore yields, uniformly on $\Psi_0$,
$$
\tfrac{1}{2n}\operatorname{tr}\bigl(\bm{\Sigma}_n^{-1} (\partial_k\bm{\Sigma}_n)\bm{\Sigma}_n^{-1} (\partial_\ell\bm{\Sigma}_n)\bigr)
\longrightarrow
\tfrac12\bigl(\bm{T}^{-1}\partial_k c, \bm{T}^{-1}\partial_\ell c\bigr)_{\ell^2(\mathbb{Z}^d)} = \bm{C}(\psi)_{k\ell}.
$$
This result is similar to \citet[Proposition~3.2]{Bachoc2014}, who establishes the analogous almost-sure trace convergence on (possibly perturbed) regular grids in a zero-mean, nugget-free setting.

For the mean block, the entries are
$$
n^{-1}(\partial_k\bm{\mu}_n)^\top\bm{\Sigma}_n^{-1}(\partial_\ell\bm{\mu}_n) = n^{-1}\sum_{i,j \le n} \partial_k\bm{\mu}(s_i) (\bm{\Sigma}_n^{-1})_{ij} \partial_\ell\bm{\mu}(s_j), 
$$
with the regressors $\partial_k\bm{\mu} = g_k * X$ of the proof of Lemma~\ref{lem:info-structure}, all bounded uniformly by Lemma~\ref{lem:incdom-bounds}(iii).
We now split $\bm{\Sigma}_n^{-1} = \bm{W}_n + \bm{E}_n$ as in the covariance block. The contribution of $\bm{E}_n$ is bounded by $\|\partial_k\bm{\mu}\|_\infty \|\partial_\ell\bm{\mu}\|_\infty n^{-1}\sum_{i,j}|(\bm{E}_n)_{ij}| = o(1)$, by the final statement of Lemma~\ref{lem:Sigmainv}. In the contribution of $\bm{W}_n$, whose entries are $w((s_i - s_j)/\delta)$, reindex the sum over $j$ by the lag $z = (s_i - s_j)/\delta$:
$$
n^{-1}(\partial_k\bm{\mu}_n)^\top\bm{W}_n(\partial_\ell\bm{\mu}_n)
= \sum_{z\in\mathbb{Z}^d} w(z) n^{-1}\!\!\sum_{i: s_i, s_i - \delta z \in \mathcal{D}_n} \partial_k\bm{\mu}(s_i)\partial_\ell\bm{\mu}(s_i - \delta z).
$$
Now, the restriction $s_i - \delta z \in \mathcal{D}_n$ may be dropped since for each lag it affects a fraction $O(\delta\|z\|/a_n)$ of the
sites, and $\sum_z |w(z)| \|z\| < \infty$ makes the total cost $O(1/a_n)$. Furthermore, the lag sum may be truncated at $\|z\| \le R$, since the lag averages are uniformly bounded and $w$ decays exponentially, so the tail is $O(\sum_{\|z\| > R}|w(z)|) = O(e^{-cR})$, uniformly in $n$. Combining these steps yields
$$
\frac1n(\partial_k\bm{\mu}_n)^\top\bm{\Sigma}_n^{-1}(\partial_\ell\bm{\mu}_n) = \sum_{\|z\|\le R} w(z) \Bigl[ \frac1n\!\!\sum_{s_i\in\mathcal{D}_n}
\partial_k\bm{\mu}(s_i) \partial_\ell\bm{\mu}(s_i - \delta z)\Bigr]
+ O(e^{-cR}) + o(1),
$$
uniformly on $\Psi_0$. By the computation in the proof of Lemma~\ref{lem:info-structure}, the average in the bracket converges to
$\langle\partial_k\bm{\mu}, \partial_\ell\bm{\mu}(\cdot - \delta z)\rangle$. Hence, for each fixed $R$, letting $n \to \infty$ leaves
$\sum_{\|z\|\le R} w(z) \langle\partial_k\bm{\mu}, \partial_\ell\bm{\mu}(\cdot - \delta z)\rangle + O(e^{-cR})$.
Then letting $R \to \infty$ recovers the full lag series $\sum_{z} w(z)\langle\partial_k\bm{\mu}, \partial_\ell\bm{\mu}(\cdot - \delta z)\rangle
= \langle\partial_k\bm{\mu},  w * \partial_\ell\bm{\mu}\rangle$. This is absolutely convergent, and equal to $\bm{M}(\psi)_{k\ell}$ of
\eqref{eq:info-limit}, by the lag-sum identity established in the proof of Lemma~\ref{lem:info-structure}, so $n^{-1}(\partial_k\bm{\mu}_n)^\top\bm{\Sigma}_n^{-1} (\partial_\ell\bm{\mu}_n) \to \bm{M}(\psi)_{k\ell}$.

Finally, every error term above was bounded uniformly on $\Psi_0$, and the convergence of the bracketed lag averages is uniform on $\Psi_0$ as well as at each lag the error is dominated by the integral of the $\psi$-uniform envelope of the kernels $g_k$ (Lemma~\ref{lem:incdom-bounds}(iii)) against the $\psi$-independent error of the shifted design averages, which vanishes by the dominated convergence argument in the proof of Lemma~\ref{lem:info-structure}. Hence $\sup_{\psi\in\Psi_0}\|n^{-1}\bm{I}_n(\psi) - \bm{I}(\psi)\| \to 0$. In particular, with $c_0 > 0$ from Lemma~\ref{lem:info-structure}, $n^{-1}\bm{I}_n(\psi) \succeq \tfrac12
c_0\bm{I}$ on $\Psi_0$ for all $n$ large.
\end{proof}

\begin{remark}\label{rem:ass-examples}\label{rem:ergodic-X}
Assumption~\ref{ass:B} admits a rich class of covariates, such as quasi-periodic functions $X(s) = m + \sum_j a_j\cos(\omega_j^\top s + \varphi_j)$ whose frequencies are non-resonant with the sampling lattice and include at least two distinct $\|\omega_j\|$. 
The condition on $X$ is stated deterministically for transparency, but extend to the random case. If $X$ is a realisation of a bounded, mean-square
continuous, strictly stationary random field on $\R^d$, independent of $(\W, \epsilon)$, whose shift action is ergodic along $\delta\mathbb{Z}^d$ (e.g.\ any mixing field), and whose spectral measure is not concentrated on a single modulus, then for each fixed $(u, h)$ a multiparameter pointwise ergodic theorem \citep[Theorem~6.2.8]{Krengel1985} applied to the field $s \mapsto X(s+u)X(s+u+h)$ gives $\hat r_n(u, h) \to \E[X(0)X(h)] = R_X(h)$ almost surely. By Fubini's theorem the convergence holds almost surely for almost every $(u, h)$, which is all that the dominated-convergence step in the proof of Lemma~\ref{lem:info-structure} uses. Lemmas~\ref{lem:info-structure}--\ref{lem:info-conv} therefore apply to almost every realisation, with the same limit $\bm{I}(\psi)$. We omit the routine details. 
\end{remark}

\begin{proof}[Proof of Theorem~\ref{thm:incdom}]
We start by showing consistency. Define the normalised contrast $\mathcal{K}_n(\psi) := n^{-1}\E_{\psi_0}[\ell_n(\psi_0) - \ell_n(\psi)] \ge 0$. Throughout this step, the bounds of Lemma~\ref{lem:incdom-bounds} are used with $\Psi$ in place of $\Psi_0$, which is legitimate by compactness of $\Psi$. Under $\psi_0$, write $\widetilde{\bm{Y}}_n = \bm{\mu}_n(\psi_0) + \bm{r}_n$ with $\bm{r}_n \sim N(\bm{0}, \bm{\Sigma}_n(\psi_0))$, so that
$$
n^{-1}\ell_n(\psi) = -\tfrac12\log(2\pi) - \tfrac{1}{2n}\log\det\bm{\Sigma}_n(\psi) - \tfrac{1}{2n}\bigl(\bm{r}_n + \bm{\Delta}_n(\psi)\bigr)^\top \bm{\Sigma}_n(\psi)^{-1} \bigl(\bm{r}_n + \bm{\Delta}_n(\psi)\bigr),
$$
where $\bm{\Delta}_n(\psi) := \bm{\mu}_n(\psi_0) - \bm{\mu}_n(\psi)$, and 
$$
\Var\bigl(n^{-1}\ell_n(\psi)\bigr) = \frac1{4n^2}\Var(\bm{r}_n^\top\bm{\Sigma}_n(\psi)^{-1}\bm{r}_n) + \frac1{n^2}\Var(\bm{r}_n^\top\bm{\Sigma}_n(\psi)^{-1}\bm{\Delta}_n).
$$
With $\bm{A} := \bm{\Sigma}_n(\psi)^{-1}$, $\Var(\bm{r}_n^\top\bm{A}\bm{r}_n) = 2\operatorname{tr}\bigl((\bm{A}\bm{\Sigma}_n(\psi_0))^2\bigr)
\le 2n\|\bm{A}\|^2\|\bm{\Sigma}_n(\psi_0)\|^2$, and  $\Var(\bm{\Delta}_n^\top\bm{A}\bm{r}_n) \le \|\bm{A}\|^2\|\bm{\Sigma}_n(\psi_0)\| \|\bm{\Delta}_n\|^2$ with $\|\bm{\Delta}_n\|^2 \le Cn$. By the eigenvalue bounds of Lemma~\ref{lem:incdom-bounds}(i) and the entrywise mean bounds of
Lemma~\ref{lem:incdom-bounds}(iii), both variances are $O(n)$, so $\Var\bigl(n^{-1}\ell_n(\psi)\bigr) = O(n^{-1})$, uniformly on $\Psi$. 
Similarly, by \eqref{eq:score} evaluated at a general $\psi$ with $\bm{r}_n$ replaced by $\widetilde{\bm{r}}_n := \bm{r}_n + \bm{\Delta}_n(\psi)$ and all matrices evaluated at $\psi$,
$$
\partial_k  \tfrac1n\ell_n(\psi)
= \tfrac1n(\partial_k\bm{\mu}_n)^\top\bm{\Sigma}_n^{-1} \widetilde{\bm{r}}_n + \tfrac{1}{2n} \widetilde{\bm{r}}_n^\top\bm{\Sigma}_n^{-1} (\partial_k\bm{\Sigma}_n)\bm{\Sigma}_n^{-1}\widetilde{\bm{r}}_n - \tfrac{1}{2n}\operatorname{tr}\bigl(\bm{\Sigma}_n^{-1} \partial_k\bm{\Sigma}_n\bigr).
$$
By Lemma~\ref{lem:incdom-bounds}, $\|\partial_k\bm{\mu}_n\| \le C\sqrt{n}$, $\|\bm{\Sigma}_n^{-1}\| \le \sigma_\epsilon^{-2}$,
$\|\partial_k\bm{\Sigma}_n\| \le C$ and $\|\bm{\Delta}_n(\psi)\|^2 \le Cn$, uniformly on $\Psi$. Hence, the three terms are bounded in absolute value by $Cn^{-1/2}\|\widetilde{\bm{r}}_n\|$, $Cn^{-1}\|\widetilde{\bm{r}}_n\|^2$,  and $\tfrac12\|\bm{\Sigma}_n^{-1}\partial_k\bm{\Sigma}_n\| \le C$,
using $n^{-1}|\operatorname{tr}(\bm{M})| \le \|\bm{M}\|$ for the trace term. Since $n^{-1/2}\|\widetilde{\bm{r}}_n\| \le \tfrac12\bigl(1 + n^{-1}\|\widetilde{\bm{r}}_n\|^2\bigr)$ and $\|\widetilde{\bm{r}}_n\|^2 \le 2\|\bm{r}_n\|^2 + 2\|\bm{\Delta}_n(\psi)\|^2 \le 2\|\bm{r}_n\|^2 + Cn$, all three bounds are dominated by a multiple of $1 + n^{-1}\|\bm{r}_n\|^2$, so
$$
\sup_{\psi\in\Psi} \bigl\|\nabla_\psi  n^{-1}\ell_n(\psi)\bigr\|
\le C\bigl(1 + n^{-1}\|\bm{r}_n\|^2\bigr) =: L_n,
$$
with $\E L_n \le C'$ since $\E\|\bm{r}_n\|^2 = \operatorname{tr}\bm{\Sigma}_n(\psi_0) \le Cn$. 
The bounds of Lemma~\ref{lem:incdom-bounds}, and hence the gradient bound above, hold equally on the convex hull of $\Psi$ (again by compactness), so the mean value theorem gives, for all $\psi, \psi' \in \Psi$, $|n^{-1}\ell_n(\psi) - n^{-1}\ell_n(\psi')| \le L_n \|\psi - \psi'\|$. 
This is the stochastic Lipschitz condition (Assumption~3A) of \citet{Newey1991}, with $B_n = L_n$ and $h(t) = t$, where $L_n = O_p(1)$ since, by Markov's inequality, $\sup_n P(L_n > M) \le \nicefrac{C'}{M}$. Taking expectations in the same inequality shows that the functions $\psi \mapsto \E_{\psi_0} n^{-1}\ell_n(\psi)$ are Lipschitz with the $n$-free constant $C'$, hence equicontinuous, and $\Var\bigl(n^{-1}\ell_n(\psi)\bigr) = O(n^{-1})$ gives $n^{-1}\ell_n(\psi) - \E_{\psi_0} n^{-1}\ell_n(\psi) \xrightarrow{p} 0$ for each fixed $\psi$. All conditions of \citet[Corollary~2.2]{Newey1991} are therefore met on the compact $\Psi$, and its conclusion is the uniform law of large numbers $\sup_{\psi\in\Psi}|n^{-1}\ell_n(\psi)
- \E_{\psi_0}  n^{-1}\ell_n(\psi)| \xrightarrow{p} 0$.

Since $\hat\psi_n$ maximises $\ell_n$, we have $n^{-1}\ell_n(\hat\psi_n) \ge n^{-1}\ell_n(\psi_0)$. Replacing both sides by their expectations, at the cost of twice the uniform error, gives $\mathcal{K}_n(\hat\psi_n) \le o_p(1)$. The estimator can therefore only be located where the contrast is asymptotically negligible. To conclude the consistency $\hat\psi_n \xrightarrow{p} \psi_0$, it thus suffices to show that $\liminf_n \mathcal{K}_n(\psi) > 0$ for every fixed $\psi \neq \psi_0$, together with the same property locally uniformly outside any ball around $\psi_0$, so that a non-vanishing contrast excludes every fixed distance from $\psi_0$. Decompose $\mathcal{K}_n(\psi)$ into a mean part $\tfrac1{2n} \bm{\Delta}_n(\psi)^\top \bm{\Sigma}_n(\psi)^{-1} \bm{\Delta}_n(\psi)$ plus a covariance part  
\begin{align*}
\tfrac1{2n}\Bigl[\log\frac{\det\bm{\Sigma}_n(\psi)}{\det\bm{\Sigma}_n(\psi_0)} + \operatorname{tr}\bigl(\bm{\Sigma}_n(\psi)^{-1}\bm{\Sigma}_n(\psi_0)\bigr) - n\Bigr].
\end{align*}
Both parts are non-negative as the mean part is a nonnegative quadratic form, and the covariance part equals $\tfrac{1}{2n}\sum_{j \le n}(\lambda_j - \log\lambda_j - 1) \ge 0$, where $\lambda_j$ are the eigenvalues of $\bm{\Sigma}_n(\psi)^{-1/2}\bm{\Sigma}_n(\psi_0) \bm{\Sigma}_n(\psi)^{-1/2}$. It therefore suffices to bound one of the two parts away from zero, and which part is used may depend on $\psi$.

If the covariance parameters of $\psi$ differ from those of $\psi_0$, the covariance part is bounded away from zero. Indeed, by the eigenvalue bounds of Lemma~\ref{lem:incdom-bounds}(i), the $\lambda_j$ above lie in $(0, b]$ with $b := \nicefrac{C}{\sigma_\epsilon^2} \ge 1$ uniformly in $n$ and $\psi \in \Psi$, and a second-order Taylor expansion of $x \mapsto x - \log x - 1$ around $x = 1$ gives $\lambda_j - \log\lambda_j - 1 \ge (\lambda_j - 1)^2/(2b^2)$.
Summing over $j$, the covariance part is at least
$$
\frac{1}{4nb^2} \bigl\|\bm{\Sigma}_n(\psi)^{-\nicefrac12}
\bigl(\bm{\Sigma}_n(\psi_0) - \bm{\Sigma}_n(\psi)\bigr)
\bm{\Sigma}_n(\psi)^{-\nicefrac12}\bigr\|_F^2
\ge \frac{c}{n} \bigl\|\bm{\Sigma}_n(\psi_0) - \bm{\Sigma}_n(\psi)\bigr\|_F^2,
$$
where the second inequality pulls the outer factors out through $\|\bm{\Sigma}_n(\psi)^{\nicefrac12}\|^2 = \|\bm{\Sigma}_n(\psi)\| \le C$, so that $c > 0$ is uniform in $n$ and $\psi \in \Psi$. Writing $\Delta c := c(\cdot\,;\psi_0) - c(\cdot\,;\psi)$, the squared Frobenius norm is $\sum_{i,j \le n} \Delta c(z_i - z_j)^2$, and for each fixed lag $z$ the number of site pairs with $z_i - z_j = z$ is $n(1 - O(\delta\|z\|/a_n))$, so
$$
\liminf_{n\to\infty} \tfrac1n\bigl\|\bm{\Sigma}_n(\psi_0)
- \bm{\Sigma}_n(\psi)\bigr\|_F^2
\ge \sum_{z\in\mathbb{Z}^d} \Delta c(z)^2.
$$
This is strictly positive unless $c(\cdot\,;\psi) = c(\cdot\,;\psi_0)$ on all of $\mathbb{Z}^d$, that is, unless $\varrho_\psi(\delta z) =
\varrho_{\psi_0}(\delta z)$ for all $z \neq 0$ and $\varrho_\psi(0) + \sigma_\epsilon^2 = \varrho_{\psi_0}(0) + \sigma_{e,0}^2$ at $z = 0$. By the Bessel asymptotics used in the proof of Lemma~\ref{lem:incdom-bounds}, $\varrho_\psi(r) = C(\kappa,\tau) r^{\nu - 1/2} e^{-\kappa r}(1 + o(1))$ as $r \to \infty$, so the decay rate along the lattice identifies $\kappa$, the prefactor then identifies $\tau$, and the value at $z = 0$ identifies $\sigma_\epsilon^2$. Hence $\liminf_n$ of the covariance part is strictly positive whenever the covariance parameters differ; cf.\ \citet[Proposition~3.1]{Bachoc2014} for a spectral-domain analogue.

If instead the covariance parameters agree, then $\mathcal{S}_\psi = \mathcal{S}_{\psi_0}$, $\bm{\Sigma}_n(\psi) = \bm{\Sigma}_n(\psi_0)$, and
$\bm{\Delta}_n = (\beta_{\mathrm{pw},0} - \beta_{\mathrm{pw}})\bm{X}_n + (\beta_{\mathrm{fc},0} - \beta_{\mathrm{fc}})(\mathcal{S}_{\psi_0}X)_n$. Note that for $k, \ell \in \{\beta_{\mathrm{pw}}, \beta_{\mathrm{fc}}\}$, the trace term vanishes in \eqref{eq:fisher-app} since $\partial_\beta\bm{\Sigma}_n = 0$. Therefore, since $\bm{X}_n = \partial_{\beta_{\mathrm{pw}}}\bm{\mu}_n$ and $(\mathcal{S}_{\psi_0}X)_n = \partial_{\beta_{\mathrm{fc}}} \bm{\mu}_n$ at $\psi_0$, the mean part equals $\tfrac12\bm{b}^\top\bm{G}_n\bm{b}$, where $\bm{b} := (\beta_{\mathrm{pw},0} - \beta_{\mathrm{pw}}, \beta_{\mathrm{fc},0} - \beta_{\mathrm{fc}})^\top$ and $\bm{G}_n$ is the $(\beta_{\mathrm{pw}}, \beta_{\mathrm{fc}})$ block of
$n^{-1}\bm{I}_n(\psi_0)$. By Lemma~\ref{lem:info-conv}, $\bm{G}_n$ converges to the corresponding block of $\bm{I}(\psi_0)$, which is
positive definite by Lemma~\ref{lem:info-structure}; hence $\liminf_n \tfrac12\bm{b}^\top\bm{G}_n\bm{b} \ge \lambda_{\min} \|\bm{b}\|^2 > 0$ for some $\lambda_{\min} > 0$ unless $\bbeta = \bbeta_0$.

Finally, $\mathcal{K}_n$ is Lipschitz on $\Psi$ with a constant independent of $n$ as its gradient obeys the same bounds as that of
$\E_{\psi_0} n^{-1}\ell_n$. The pointwise bound therefore also implies that $\liminf_n \inf_{\psi \in \Psi \setminus B(\psi_0,\epsilon)}
\mathcal{K}_n(\psi) > 0$ for every $\epsilon > 0$: otherwise some sequence $\psi_{n_j}$ in the compact set $\Psi \setminus B(\psi_0,\epsilon)$ would satisfy $\mathcal{K}_{n_j}(\psi_{n_j}) \to 0$ and, along a subsequence, $\psi_{n_j} \to \psi^*$ with $\psi^* \neq \psi_0$. Equi-Lipschitz continuity then transfers the vanishing contrast to $\psi^*$, contradicting $\liminf_n \mathcal{K}_n(\psi^*) > 0$. Combined with $\mathcal{K}_n(\hat\psi_n) \le o_p(1)$, this forces $P\bigl(\hat\psi_n \notin B(\psi_0,\epsilon)\bigr) \to 0$ for every
$\epsilon > 0$. 

We now prove the asymptotic normality.
Write $\bm{J}_n(\psi) = -\nabla^2\ell_n(\psi)$ for the observed information, and set $\psi_t := \psi_0 + t(\hat\psi_n - \psi_0)$
and $\bar{\bm{J}}_n := \int_0^1 \bm{J}_n(\psi_t)\, dt$. Since $\psi_0$ is interior to $\Psi$ and $\hat\psi_n \xrightarrow{p}
\psi_0$, the event that $\hat\psi_n$ (and the whole segment $\{\psi_t : t \in [0,1]\}$) lies in a fixed ball $B(\psi_0, r) \subset \operatorname{int}\Psi$ has probability tending to one. We argue on this event, which does not affect distributional limits. On it the maximiser is interior, so $\nabla\ell_n(\hat\psi_n) = 0$, and the fundamental theorem of calculus applied to the $C^1$ map $t \mapsto \nabla\ell_n(\psi_t)$ (first-order Taylor formula with integral remainder, the smoothness follows from Lemma~\ref{lem:incdom-bounds}) gives 
$$
0 = \nabla\ell_n(\psi_0) - \bar{\bm{J}}_n (\hat\psi_n - \psi_0),
\qquad
\sqrt n (\hat\psi_n - \psi_0) = \bigl(n^{-1}\bar{\bm{J}}_n\bigr)^{-1}
n^{-1/2}\nabla\ell_n(\psi_0).
$$
It remains to identify the limits of $\nabla\ell_n(\psi_0)$ and $\bar{\bm{J}}_n$.

For $\bar{\bm{J}}_n$, we have $\E_{\psi_0}\bm{J}_n(\psi_0) = \bm{I}_n(\psi_0)$,  the fluctuation $n^{-1}(\bm{J}_n(\psi_0) - \bm{I}_n(\psi_0))$ has variance $O(n^{-1})$ by Lemma~\ref{lem:incdom-bounds}, and $n^{-1}\bm{I}_n(\psi_0) \to \bm{I}(\psi_0)$ by Lemma~\ref{lem:info-conv}. Hence, $n^{-1}\bm{J}_n(\psi_0) \xrightarrow{p}\bm{I}(\psi_0)$. Further, differentiating \eqref{eq:score} twice, the entries of $\nabla^3\ell_n(\psi)$ are again linear and quadratic forms in $\bm{r}_n$ whose coefficient matrices have operator norms bounded uniformly on $B(\psi_0, r)$, since Lemma~\ref{lem:incdom-bounds} bounds three derivatives of $\bm{\mu}_n$ and $\bm{\Sigma}_n$. Thus, as in the bound for $L_n$, $\sup_{B(\psi_0, r)} n^{-1}\|\nabla^3\ell_n\| \le L_n'$ with $L_n' = O_p(1)$, and the mean value theorem on the (convex) ball gives $\|n^{-1}\bar{\bm{J}}_n - n^{-1}\bm{J}_n(\psi_0)\| \le L_n' \|\hat\psi_n - \psi_0\| \xrightarrow{p} 0$. Therefore $n^{-1}\bar{\bm{J}}_n \xrightarrow{p}\bm{I}(\psi_0)\succ 0$. In
particular $\bar{\bm{J}}_n$ is invertible with probability tending to one, justifying the inversion above.

By \eqref{eq:score} evaluated at $\psi_0$, each coordinate of the score $\nabla\ell_n(\psi_0)$ is a linear form plus a centred quadratic form in the Gaussian vector $\bm{r}_n \sim N(0, \bm{\Sigma}_n(\psi_0))$. We show that $n^{-1/2}\nabla\ell_n(\psi_0) \xrightarrow{d} N(0, \bm{I}(\psi_0))$ by the method of cumulants. Related normality arguments for increasing-domain Gaussian-likelihood scores are given by \citet{MardiaMarshall1984} and \citet[Proposition~3.2]{Bachoc2014}. By the Cram\'er--Wold device it suffices to fix $\bm{\lambda} \neq \bm{0}$ and show $n^{-1/2}S_n \xrightarrow{d} N(0, \bm{\lambda}^\top\bm{I}(\psi_0)\bm{\lambda})$ for $S_n := \bm{\lambda}^\top\nabla\ell_n(\psi_0)$, which takes the form
$S_n = \bm{a}^\top\bm{r}_n + \tfrac12\bigl(\bm{r}_n^\top\bm{B}\bm{r}_n - \operatorname{tr}(\bm{B}\bm{\Sigma}_n)\bigr)$, 
with $\bm{a} := \bm{\Sigma}_n^{-1}\sum_k \lambda_k \partial_k\bm{\mu}_n$, $\bm{B} := \bm{\Sigma}_n^{-1}\Bigl(\sum_k \lambda_k \partial_k\bm{\Sigma}_n\Bigr)\bm{\Sigma}_n^{-1}$ and $\bm{\Sigma}_n = \bm{\Sigma}_n(\psi_0)$. Note that $\operatorname{tr}(\bm{B}\bm{\Sigma}_n)
= \E \bm{r}_n^\top\bm{B}\bm{r}_n$, so the quadratic form is centred. Set $\tilde{\bm{a}} := \bm{\Sigma}_n^{1/2}\bm{a}$ and
$\tilde{\bm{B}} := \bm{\Sigma}_n^{1/2}\bm{B}\bm{\Sigma}_n^{1/2}$, which is symmetric. A Gaussian completion of squares gives, for
$|t| < \|\tilde{\bm{B}}\|^{-1}$, the cumulant generating function
$$
\log \E  e^{tS_n} = -\tfrac12\log\det(\bm{I} - t\tilde{\bm{B}})
+ \tfrac{t^2}{2} \tilde{\bm{a}}^\top(\bm{I} - t\tilde{\bm{B}})^{-1}\tilde{\bm{a}} - \tfrac{t}{2}\operatorname{tr}\tilde{\bm{B}},
$$
and expanding $-\tfrac12\log\det(\bm{I} - t\tilde{\bm{B}}) = \tfrac12\sum_{p \ge 1} t^p \operatorname{tr}(\tilde{\bm{B}}^p)/p$
and the resolvent in powers of $t$ yields $\operatorname{cum}_1(S_n) = 0$ and $\operatorname{cum}_p(S_n) = \tfrac{(p-1)!}{2}\operatorname{tr}(\tilde{\bm{B}}^p) + \tfrac{p!}{2} \tilde{\bm{a}}^\top\tilde{\bm{B}}^{p-2} \tilde{\bm{a}}$ for $p \ge 2$. For $p = 2$, this equals $\bm{a}^\top\bm{\Sigma}_n\bm{a} + \tfrac12\operatorname{tr}\bigl((\bm{B}\bm{\Sigma}_n)^2\bigr) = \bm{\lambda}^\top\bm{I}_n(\psi_0)\bm{\lambda}$, by \eqref{eq:fisher-app}, and $\operatorname{cum}_2(n^{-1/2}S_n) = n^{-1}\bm{\lambda}^\top\bm{I}_n(\psi_0)\bm{\lambda} \to \bm{\lambda}^\top\bm{I}(\psi_0)\bm{\lambda}$ by Lemma~\ref{lem:info-conv}. For $p \ge 3$, let $d_1, \ldots, d_n$ be the eigenvalues of $\tilde{\bm{B}}$. By Lemma~\ref{lem:incdom-bounds}(i)--(ii), $\max_i |d_i| = \|\tilde{\bm{B}}\| \le \|\bm{\Sigma}_n\| \|\bm{B}\| \le K$ with $K$ independent of $n$, while $\sum_i d_i^2 = \operatorname{tr}\bigl((\bm{B}\bm{\Sigma}_n)^2\bigr) \le 2\bm{\lambda}^\top\bm{I}_n(\psi_0)\bm{\lambda} \le Cn$ and $\|\tilde{\bm{a}}\|^2 = \bm{a}^\top\bm{\Sigma}_n\bm{a} \le Cn$. Hence
$$
\bigl|\operatorname{tr}(\tilde{\bm{B}}^p)\bigr| = \Bigl|\sum_i d_i^p\Bigr|
\le K^{p-2}\sum_i d_i^2 \le C K^{p-2} n,
$$
and $\bigl|\tilde{\bm{a}}^\top\tilde{\bm{B}}^{p-2}\tilde{\bm{a}}\bigr| \le K^{p-2} \|\tilde{\bm{a}}\|^2 \le C K^{p-2} n$. Thus, for every $p \ge 3$, $\operatorname{cum}_p(n^{-1/2}S_n) = n^{-p/2}\operatorname{cum}_p(S_n) = O(n^{1-p/2}) \to 0$. All cumulants of $n^{-1/2}S_n$ therefore converge to those of $N(0, \bm{\lambda}^\top\bm{I}(\psi_0)\bm{\lambda})$. Since moments are polynomials in the cumulants and the normal distribution is determined by its moments, the Fr\'echet--Shohat theorem gives $n^{-1/2}S_n \xrightarrow{d} N(0, \bm{\lambda}^\top\bm{I}(\psi_0)\bm{\lambda})$ and finally, Slutsky's theorem applied to $\nabla\ell_n(\psi_0)$ and $\bar{\bm{J}}_n$ yields $\sqrt n (\hat\psi_n - \psi_0) \xrightarrow{d} \bm{I}(\psi_0)^{-1} N(0, \bm{I}(\psi_0)) = N(0, \bm{I}(\psi_0)^{-1})$.
\end{proof}

\subsection{Proofs for Section~\ref{sec:inference}}
We need the following basic finite element result.

\begin{proposition}\label{prop:fem-error}
Let Assumption~\ref{ass:A} hold with $\alpha = 2$, let $\mathcal{D}$ be polygonal, and let $\{V_h\}$, $\mathcal{S}_h$ and $\varrho_h$ be as in Section~\ref{sec:fem-error}. There is a constant $C$, depending only on $\mathcal{D}$, $\kappa$, $\tau$, and the mesh, such that for every $X \in L^2(\mathcal{D})$ and all sufficiently small $h$, 
\begin{enumerate}[label=(\roman*),leftmargin=*,itemsep=1pt]
\item $\|\mathcal{S}X - \mathcal{S}_h X\|_{L^2(\mathcal{D})} \le C h^2  \|X\|_{L^2(\mathcal{D})}$;
\item $\|\mathcal{S}X - \mathcal{S}_h X\|_{L^\infty(\mathcal{D})} \le C h^{2 - d/2}  \|X\|_{L^2(\mathcal{D})}$;
\item $\sup_{s, t \in \barD} |\varrho(s,t) - \varrho_h(s,t)| \le C h^{2 - d/2}$.
\end{enumerate}
\end{proposition}

\begin{proof}
Throughout this proof, $a(v, w) := \kappa^2 (v,w)_{L^2} + (\nabla v, \nabla w)_{L^2}$ denotes the bilinear form of $L$ on $H^1(\mathcal{D})$, which is bounded, symmetric, and coercive since $\kappa > 0$. We let $C$ denote a generic constant with the stated dependencies and  repeatedly use that for quasi-uniform, shape-regular families and $d \le 3$ \citep[see][Theorem~4.4.20]{BrennerScott2008}, the nodal interpolant $I_h : H^2(\mathcal{D}) \to V_h$ satisfies
\begin{equation}\label{eq:interp}
\|u - I_h u\|_{L^2} + h \|u - I_h u\|_{H^1} \le C_I  h^2 |u|_{H^2},
\quad
\|u - I_h u\|_{L^\infty} \le C_I  h^{2 - \nicefrac{d}{2}} |u|_{H^2},
\end{equation}
Further, by Assumption~\ref{ass:A} the Neumann problem is $H^2$-regular, i.e., for $g \in L^2(\mathcal{D})$, $u = L^{-1}g$ satisfies
$\|u\|_{H^2} \le C_R \|g\|_{L^2}$ \citep[Theorems~2.2.2.5 and~3.2.1.3]{Grisvard1985}.

Fix $X \in L^2(\mathcal{D})$, let $u = L^{-1}X$ and $u_h = L_h^{-1}P_h X$, so that $\mathcal{S}X - \mathcal{S}_h X = \tau^{-1}(u - u_h)$, and write $e = u - u_h$. Since $(P_h X, v_h)_{L^2} = (X, v_h)_{L^2}$ for all $v_h \in V_h$ by the definition of $P_h$, $u_h$ is the conforming Galerkin approximation in $V_h$ of the variational solution $u$ of $a(u, v) = (X, v)_{L^2}$, $v \in H^1(\mathcal{D})$. Part (i) is therefore the classical $L^2$ error estimate for an $H^2$-regular second-order problem: since $L$ is self-adjoint the dual problem coincides with the primal one, and \citet[Thm.~5.7.6]{BrennerScott2008} together with \eqref{eq:interp} gives $\|e\|_{L^2} \le C h \|e\|_{H^1} \le C h^2 |u|_{H^2} \le C h^2 \|X\|_{L^2}$. Multiplying by $\tau^{-1}$ gives (i). 

For (ii), split $e = (u - I_h u) + (I_h u - u_h)$. The first term is bounded by \eqref{eq:interp}: $\|u - I_h u\|_{L^\infty} \le C h^{2 - \nicefrac{d}{2}} |u|_{H^2} \le C h^{2 - \nicefrac{d}{2}}\|X\|_{L^2}$. The second term is in $V_h$, so by the inverse estimate $\|v_h\|_{L^\infty} \le C_{\mathrm{inv}} h^{-\nicefrac{d}{2}} \|v_h\|_{L^2}$, $v_h \in V_h$ \citep[Theorem~4.5.11]{BrennerScott2008}, the triangle inequality, \eqref{eq:interp}, and (i),
\begin{align*}
\|I_h u - u_h\|_{L^\infty} &\le C_{\mathrm{inv}} h^{-\frac{d}{2}}
\bigl( \|I_h u - u\|_{L^2} + \|u - u_h\|_{L^2} \bigr) 
\le C h^{-\frac{d}{2}} \cdot h^2 \|X\|_{L^2}.
\end{align*}
Adding the two bounds and multiplying by $\tau^{-1}$ gives (ii).

To prove (iii) note that for $s \in \barD$,  $w \mapsto (L^{-1}w)(s)$ is a bounded linear functional on $L^2(\mathcal{D})$, since
$\|L^{-1}w\|_{L^\infty} \le C_S \|L^{-1}w\|_{H^2} \le C_S C_R \|w\|_{L^2}$ by the Sobolev embedding and $H^2$-regularity. A  $g_s \in L^2(\mathcal{D})$ exists with
\begin{equation}\label{eq:green-cont}
(g_s, w)_{L^2} = (L^{-1}w)(s) \quad \text{for all } w \in L^2(\mathcal{D}),
\qquad
\|g_s\|_{L^2} \le C_S C_R, 
\end{equation}
by Riesz representation theorem. In the eigenbasis, $(g_s, e_j) = (L^{-1}e_j)(s) = \lambda_j^{-1}e_j(s)$, so $g_s = \sum_j \lambda_j^{-1} e_j(s) e_j$.
Thus, with $U = \tau^{-1}L^{-1}\W = \tau^{-1}\sum_j \lambda_j^{-1}\xi_j e_j$,
$$
\E\bigl[U(s)  U(t)\bigr] = \tau^{-2} \sum_j \lambda_j^{-2} e_j(s) e_j(t)
= \tau^{-2} (g_s, g_t)_{L^2},
$$
that is, $\varrho(s,t) = \tau^{-2}(g_s, g_t)_{L^2}$. Discretely, the functional $w \mapsto (L_h^{-1}P_h w)(s) = (L^{-1}w)(s) - \tau (\mathcal{S}w - \mathcal{S}_h w)(s)$ is also bounded on $L^2(\mathcal{D})$, by \eqref{eq:green-cont} and (ii), and it vanishes for $w \perp V_h$, since then $P_h w = 0$. Its Riesz representer $g^h_s$ therefore lies in $V_h$, and expanding $U_h = \tau^{-1}L_h^{-1}P_h\W$ in any orthonormal basis of $V_h$
gives, exactly as in the continuous case, $\varrho_h(s,t) = \tau^{-2}(g^h_s, g^h_t)_{L^2}$. Moreover, by construction, $g_s - g^h_s$ is the representer of $w \mapsto \tau (\mathcal{S}w - \mathcal{S}_h w)(s)$, so (ii) gives, uniformly in $s \in \barD$,
$$
\|g_s - g^h_s\|_{L^2} = \sup_{\|w\|_{L^2} = 1} \tau \bigl|(\mathcal{S}w - \mathcal{S}_h w)(s)\bigr|
\le \tau C h^{2 - \nicefrac{d}{2}}.
$$
Combining the previous steps, the Cauchy--Schwarz inequality, and $\|g^h_s\|_{L^2} \le \|g_s\|_{L^2} + \|g_s - g^h_s\|_{L^2} \le C$ for $h$ small,
\begin{align*}
|\varrho(s,t) - \varrho_h(s,t)| &\le \tau^{-2}\Bigl(
   \bigl|\bigl(g_s - g^h_s,  g_t\bigr)\bigr|
+ \bigl|\bigl(g^h_s,  g_t - g^h_t\bigr)\bigr| \Bigr) \\
&\le \tau^{-2}\bigl( C h^{2-\nicefrac{d}{2}}  \|g_t\|_{L^2}
   + C \|g_t - g^h_t\|_{L^2} \bigr) \le C h^{2 - \nicefrac{d}{2}},
\end{align*}
uniformly over $s, t \in \barD$.
\end{proof}

\begin{proof}[Proof of Theorem~\ref{thm:fem-estimator}]
The rates $\varepsilon_h = O(h^{2 - \nicefrac{d}{2}})$ and $\delta_h = O(h^{2 - \nicefrac{d}{2}})$ are parts~(ii) and~(iii) of Proposition~\ref{prop:fem-error}, the first with a constant proportional to $\|X\|_{L^2(\mathcal{D})}$. It remains to prove (i) and (ii).
Write $\bm{D} = \bm{D}_n$ and $\widetilde{\bm{D}} = \bm{D}_{n,h}$ for the $n \times 2$ design matrices with columns $(X(s_i))_i$ and, respectively,
$(\mathcal{S}X(s_i))_i$ and $((\mathcal{S}_hX)(s_i))_i$. Further write $\bm{\Sigma} = \bm{\Sigma}_n = \bm{K}_n + \sigma_\epsilon^2 \bm{I}$ and
$\widetilde{\bm{\Sigma}} = \bm{\Sigma}_{n,h} = \bm{K}_{n,h} + \sigma_\epsilon^2 \bm{I}$, and let $\bm{M} = \bm{D}^\top \bm{\Sigma}^{-1} \bm{D}$ and $\widetilde{\bm{M}} = \widetilde{\bm{D}}^\top \widetilde{\bm{\Sigma}}^{-1}\widetilde{\bm{D}}$. Throughout, $\|\cdot\|$ denotes the spectral norm and $C$ a constant depending only on $\sigma_\epsilon^2$, $\|X\|_{L^\infty}$, $\|\mathcal{S}X\|_{L^\infty}$, $\|\bbeta_0\|$, and $\bm{G}$ (both supremum norms are finite because $X, \mathcal{S}X \in \HH \hookrightarrow C(\barD)$).
Since $\bm{K}_n, \bm{K}_{n,h} \succeq 0$, we have $\|\bm{\Sigma}^{-1}\| \le \sigma_\epsilon^{-2}$ and $\|\widetilde{\bm{\Sigma}}^{-1}\| \le \sigma_\epsilon^{-2}$. The difference $\bm{D} - \widetilde{\bm{D}}$ has first column zero and second column with entries $(\mathcal{S}X - \mathcal{S}_hX)(s_i)$, so $\|\bm{D} - \widetilde{\bm{D}}\| \le \|\bm{D} - \widetilde{\bm{D}}\|_F \le \sqrt{n}  \varepsilon_h$; and $\bm{\Sigma} - \widetilde{\bm{\Sigma}} = \bm{K}_n - \bm{K}_{n,h}$ is symmetric with entries bounded by $\delta_h$, so $\|\bm{\Sigma} - \widetilde{\bm{\Sigma}}\| \le \max_i \sum_{j \le n} |(\bm{K}_n - \bm{K}_{n,h})_{ij}| \le n \delta_h$.
Moreover $\|\bm{D}\| \le \|\bm{D}\|_F \le \sqrt{2n} \max(\|X\|_\infty, \|\mathcal{S}X\|_\infty) \le C\sqrt n$, and hence also $\|\widetilde{\bm{D}}\| \le C\sqrt n$, since $\varepsilon_h \le n\varepsilon_h \le \eta_{n,h} \le \eta_0 \le 1$ for the choice of $\eta_0$ below. Finally, by the resolvent identity and the above bounds, $\|\bm{\Sigma}^{-1} - \widetilde{\bm{\Sigma}}^{-1}\| = \|\bm{\Sigma}^{-1}(\widetilde{\bm{\Sigma}} - \bm{\Sigma})\widetilde{\bm{\Sigma}}^{-1}\| \le \sigma_\epsilon^{-4}  n  \delta_h$.
Now, decompose
$$
\widetilde{\bm{M}} - \bm{M}
= (\widetilde{\bm{D}} - \bm{D})^\top \widetilde{\bm{\Sigma}}^{-1} \widetilde{\bm{D}}
+ \bm{D}^\top \bigl(\widetilde{\bm{\Sigma}}^{-1} - \bm{\Sigma}^{-1}\bigr)
\widetilde{\bm{D}} + \bm{D}^\top \bm{\Sigma}^{-1} (\widetilde{\bm{D}} - \bm{D}),
$$
and bound the three terms by $\sqrt n \varepsilon_h \cdot \sigma_\epsilon^{-2} \cdot C\sqrt n = C n \varepsilon_h$, $C\sqrt n \cdot \sigma_\epsilon^{-4} n \delta_h \cdot C \sqrt n = C n^2 \delta_h$, and $C n \varepsilon_h$ respectively, using the above bounds. Hence
\begin{equation}\label{eq:DeltaM}
\|\widetilde{\bm{M}} - \bm{M}\| \le C \eta_{n,h},
\qquad \eta_{n,h} = n\varepsilon_h + n^2\delta_h.
\end{equation}
By Theorem~\ref{thm:infill}, $\bm{M} \to \bm{G} \succ 0$, so there are $n_0$ and $c_0 > 0$ with $\lambda_{\min}(\bm{M}) \ge c_0$ for $n \ge n_0$. Choosing $\eta_0 = \min\{1,  c_0/(2C)\}$ in \eqref{eq:DeltaM} guarantees $\lambda_{\min}(\widetilde{\bm{M}}) \ge c_0/2$, hence
$\widetilde{\bm{M}}$ is invertible with $\|\widetilde{\bm{M}}^{-1}\| \le 2/c_0$, and writing $\bm{M}^{-1} - \widetilde{\bm{M}}^{-1} = \bm{M}^{-1}(\bm{M} - \widetilde{\bm{M}})\widetilde{\bm{M}}^{-1}$ gives  $\|\widetilde{\bm{M}}^{-1} - \bm{M}^{-1}\| \le \|\bm{M}^{-1}\| \|\widetilde{\bm{M}}^{-1}\| \|\widetilde{\bm{M}} - \bm{M}\| \le C  \eta_{n,h}$.

The data satisfy $\widetilde{\bm{Y}}_n = \bm{D}\bbeta_0 + \bm{\eta}_n$ with $\bm{\eta}_n \sim N(0, \bm{\Sigma})$, and
$\hat\bbeta_{n,h} = \widetilde{\bm{M}}^{-1}\widetilde{\bm{D}}^\top \widetilde{\bm{\Sigma}}^{-1}\widetilde{\bm{Y}}_n$ is a fixed linear map applied to a Gaussian vector, hence Gaussian. Its mean is
$$
\widetilde{\bm{M}}^{-1}\widetilde{\bm{D}}^\top\widetilde{\bm{\Sigma}}^{-1} \bm{D} \bbeta_0
= \bbeta_0 + \widetilde{\bm{M}}^{-1}\widetilde{\bm{D}}^\top
\widetilde{\bm{\Sigma}}^{-1}\bigl(\bm{D} - \widetilde{\bm{D}}\bigr)
\bbeta_0 = \bbeta_0 + \bm{b}_{n,h},
$$
using $\widetilde{\bm{M}}^{-1}\widetilde{\bm{D}}^\top\widetilde{\bm{\Sigma}}^{-1} \widetilde{\bm{D}} = I$. Splitting
$\widetilde{\bm{\Sigma}}^{-1} = \widetilde{\bm{\Sigma}}^{-1/2}\cdot
\widetilde{\bm{\Sigma}}^{-1/2}$ and using
$\|\widetilde{\bm{\Sigma}}^{-1/2}\widetilde{\bm{D}}\|
= \|\widetilde{\bm{M}}\|^{1/2} \le C$, $\|\widetilde{\bm{\Sigma}}^{-1/2}\|
\le \sigma_\epsilon^{-1}$, and
$\|\bm{D} - \widetilde{\bm{D}}\| \le \sqrt n \varepsilon_h$,
$$
\|\bm{b}_{n,h}\| \le \|\widetilde{\bm{M}}^{-1}\|\cdot
\|\widetilde{\bm{\Sigma}}^{-1/2}\widetilde{\bm{D}}\|\cdot
\|\widetilde{\bm{\Sigma}}^{-1/2}\|\cdot
\|\bm{D} - \widetilde{\bm{D}}\|\cdot\|\bbeta_0\|
\le C \sqrt n \varepsilon_h .
$$
In particular, $n \varepsilon_h \le \eta_{n,h} \le \eta_0$ gives $\|\bm{b}_{n,h}\| \le C \eta_0/\sqrt{n}$, the second bias bound of part~(i).
The covariance is
$$
\bm{V}_{n,h} = \widetilde{\bm{M}}^{-1} \widetilde{\bm{D}}^\top \widetilde{\bm{\Sigma}}^{-1} \bm{\Sigma} \widetilde{\bm{\Sigma}}^{-1}\widetilde{\bm{D}} \widetilde{\bm{M}}^{-1}
= \widetilde{\bm{M}}^{-1} + \widetilde{\bm{M}}^{-1} \widetilde{\bm{D}}^\top \widetilde{\bm{\Sigma}}^{-1} \bigl(\bm{\Sigma} - \widetilde{\bm{\Sigma}}\bigr) \widetilde{\bm{\Sigma}}^{-1}\widetilde{\bm{D}} \widetilde{\bm{M}}^{-1}.
$$
Using $\|\widetilde{\bm{\Sigma}}^{-1}\widetilde{\bm{D}}\| \le \tfrac1{\sigma_\epsilon}\|\widetilde{\bm{\Sigma}}^{-1/2}\widetilde{\bm{D}}\| = \tfrac1{\sigma_\epsilon}\|\widetilde{\bm{M}}\|^{1/2} \le C$, the second term is bounded by $\|\widetilde{\bm{M}}^{-1}\|^2 \|\widetilde{\bm{\Sigma}}^{-1}\widetilde{\bm{D}}\|^2 \|\bm{\Sigma} - \widetilde{\bm{\Sigma}}\| \le C n \delta_h \le C \eta_{n,h}$. Combining with the bound on
$\|\widetilde{\bm{M}}^{-1} - \bm{M}^{-1}\|$ gives $\|\bm{V}_{n,h} - \bm{M}_n^{-1}\| \le C\eta_{n,h}$, completing (i).

For (ii), if $\eta_{n,h_n} \to 0$ then $\|\bm{b}_{n,h_n}\| \le C \eta_{n,h_n}/\sqrt n \to 0$ and $V_{n,h_n} \to \bm{G}^{-1}$ by (i) and $\bm{M}_n^{-1} \to \bm{G}^{-1}$ (Theorem~\ref{thm:infill}). Gaussianity from (i) and convergence of mean and covariance give $\hat\bbeta_{n,h_n} \xrightarrow{d} N(\bbeta_0, \bm{G}^{-1})$ via characteristic functions. For the coverage claim, the Wald pivot built from $\widetilde{\bm{M}}$ is
$\widetilde{\bm{M}}^{1/2}(\hat\bbeta_{n,h} - \bbeta_0) \sim N(\widetilde{\bm{M}}^{1/2} \bm{b}_{n,h}, \widetilde{\bm{M}}^{1/2} \bm{V}_{n,h} \widetilde{\bm{M}}^{1/2})$, whose mean tends to zero and whose covariance tends to $\bm{G}^{1/2} \bm{G}^{-1} \bm{G}^{1/2} = I$. Hence,  coverage converges to the nominal level.

For the forced model, $\bm{D}$ and $\widetilde{\bm{D}}$ have the single columns $(\mathcal{S}X(s_i))_i$ and $((\mathcal{S}_hX)(s_i))_i$, the
scalar $\bm{M} \to \|X\|_{L^2}^2 > 0$ replaces $\bm{G}$ (Theorem~\ref{thm:infill}), and only $\|\mathcal{S}X\|_{L^\infty}$ enters the constants, which is finite for any $X \in L^2(\mathcal{D})$ since $\mathcal{S}X \in H^2(\mathcal{D}) \hookrightarrow C(\barD)$. The arguments above thus apply verbatim.
\end{proof}

\begin{remark}\label{rem:fem-fractional}
The rate in Proposition~\ref{prop:fem-error}(iii) could likely be obtained through the methods of \citet{CoxKirchner2020}, who prove
sup-norm covariance convergence at rates up to $h^{2\alpha - d}$, up to an arbitrarily small polynomial loss, for $\alpha > d/2$ and Dirichlet boundary conditions. Since Theorem~\ref{thm:fem-estimator} enters the discretisation only through $\varepsilon_h$ and $\delta_h$, its proof applies unchanged, and with $\delta_h = O(h^{4-d})$ the mesh requirement stated after Theorem~\ref{thm:fem-estimator} sharpens from $o(n^{-4/(4-d)})$ to $h_n = o(n^{-2/(4-d)})$.
\end{remark}

\subsection{Proofs for Appendix~\ref{app:colored}}

\begin{proof}[Proof of Proposition~\ref{prop:colored}]
We modify the proof of Proposition~\ref{prop:timeavg}, keeping its notation. In the eigenbasis $\{e_j\}$ of $L$, the coordinates of
$L^{-\zeta/2} \mathrm{d}W_t$ are independent scaled Brownian increments $\lambda_j^{-\nicefrac{\zeta}{2}}\,\mathrm{d}w_j$, so the modes are independent OU processes
$$
\mathrm{d}y_j = \bigl(-a_j y_j + f_j\bigr)\,\mathrm{d}t + \sqrt{q_j}\,\mathrm{d}w_j,
$$
where $q_j := \lambda_j^{-\zeta}$. The mean dynamics do not involve the noise, so $\E[\bar y_j] = f_j/a_j$ and (i) follows as in the proof of
Proposition~\ref{prop:timeavg}. If $y^{(1)}$ denotes the stationary OU process with rate $a$ and unit noise intensity, then by linearity of the OU map in the driving noise, the centred stationary solution with intensity $q$ has the same law as $\sqrt{q}$ times the centred $y^{(1)}$. Thus, every second-moment formula of the white-noise case holds with an additional factor $q_j$, while means are unchanged.

The stationary variance of mode $j$ is $q_j/(2a_j) = (2D)^{-1}\lambda_j^{-(1+\zeta)}$, so the stationary field has covariance operator $(2D)^{-1}L^{-(1+\zeta)}$. By Lemma~\ref{lem:embed}(b), this operator is trace class (and the field is $L^2(\mathcal{D})$-valued) if and only if $1 + \zeta > \nicefrac{d}{2}$, i.e.\ $\zeta > \nicefrac{d}{2} - 1$, in which case it is the Whittle--Mat\'ern field with $\alpha = 1+\zeta$ and $\tau^2 = 2D$ by the proof of Proposition~\ref{prop:rep}, which applies to any exponent exceeding $\nicefrac{d}{2}$. This proves (ii). 

By the scaling identity and the exact variance computation of the proof of Proposition~\ref{prop:timeavg},
$$
\Var(\bar y_j) = \frac{q_j}{a_j^2 T} \bigl(1 + r_j\bigr)
= \frac{\lambda_j^{-(2+\zeta)}}{D^2 T} \bigl(1 + r_j\bigr),
\qquad |r_j| \le (a_j T)^{-1} \le (\rho T)^{-1},
$$
so the covariance operator of $\sqrt{T}(\bar Y_T - A^{-1}f)$ is $D^{-2}L^{-(2+\zeta)}(I + R_T)$ with $\|R_T\| \le (\rho T)^{-1}$. Since $2 + \zeta \ge 2 > \nicefrac{d}{2}$ for $d \le 3$, the operator $L^{-(2+\zeta)}$ is trace class, and the Hilbert--Schmidt and weak-convergence arguments of the proof of Proposition~\ref{prop:timeavg} apply, giving the asymptotic Whittle--Mat\'ern law with $\alpha = 2 + \zeta$ and
$\tau = D$.

Finally, the standardisation computation in the proof of Proposition~\ref{prop:timeavg} involves only the mean and is unchanged: $\beta_{\mathrm{fc}}^{\star} = \gamma_0/\rho$. By the proof of Proposition~\ref{prop:rep}, the Cameron--Martin space of the limit field in (iii) is $D(L^{(2+\zeta)/2})$, and for $f = \gamma_0 X$,
\begin{equation*}
A^{-1}f \in D\bigl(L^{\frac{2+\zeta}{2}}\bigr)
\Longleftrightarrow \sum_j \lambda_j^{2+\zeta} \bigl(\lambda_j^{-1} \frac{f_j}{D}\bigr)^2 \!< \!\infty
\Longleftrightarrow \sum_j \lambda_j^{\zeta} x_j^2 < \infty
\Longleftrightarrow X \in D\bigl(L^{\frac{\zeta}{2}}\bigr).
\end{equation*}
\end{proof}

\begin{proof}[Proof of Corollary~\ref{cor:cascade}]

We induct on $m$, with all covariance operators identified up to level-dependent positive multiplicative constants, and with the induction hypothesis that the level-$m$ balance is driven by temporally white noise with spatial covariance proportional to $L^{-(m-1)}$, i.e., the setting of Proposition~\ref{prop:colored} with $\zeta = m - 1$. The hypothesis holds for $m = 1$, whose driving noise is white ($L^0 = I$). The driving noise is well defined for every $m$, since $L^{-(m-1)}$ is a bounded positive operator on $L^2(\mathcal{D})$, and the construction uses only the covariance operator of the level-$(m-1)$ equilibrium, which exists even when that field itself is function-valued only for $m - 1 > \nicefrac{d}{2}$. Part (ii) gives instantaneous covariance proportional to $L^{-(1 + (m-1))} = L^{-m}$, which by Lemma~\ref{lem:embed}(b) is trace class so the field is $L^2(\mathcal{D})$-valued and Whittle--Mat\'ern with $\alpha = m$ precisely when $m > \nicefrac{d}{2}$. As this is proportional to the level-$m$ equilibrium covariance, it is also the driving covariance of the level-$(m + 1)$ balance, which closes the induction. Part (iii) gives time-average covariance proportional to $L^{-(2 + (m-1))} = L^{-(m+1)}$, with the same operator-norm error bound and weak-convergence conclusion.
\end{proof}

\end{appendices}

\bibliographystyle{abbrvnat}
\bibliography{bibliography}

\end{document}